\documentclass[reqno, a4paper]{amsart}
\usepackage{amsmath}
\usepackage{amssymb}
\usepackage{amsthm}

\usepackage[scale=0.9]{geometry}

\usepackage{natbib}
\usepackage{bibentry} 

\usepackage[english]{babel}
\usepackage[utf8]{inputenc}

\usepackage{subfig}
\usepackage{graphicx}

\usepackage[unicode,hypertexnames=false]{hyperref}
\usepackage[active]{srcltx} 

\usepackage{bm} 

\usepackage{tensor}
\usepackage{accents}

\usepackage{enumitem}

\usepackage{lineno}

\usepackage{thmtools}

\usepackage{multicol}

\usepackage{todonotes}

\usepackage{placeins}

\usepackage{booktabs}

\usepackage{listings}
\usepackage{xcolor}

\usepackage{mdframed}

\DeclareMathOperator{\diag}{diag}

\DeclareMathOperator{\sinc}{sinc}

\newcommand{\bigo}[1]{\ensuremath{O\left(#1 \right)}}

\newcommand{\exponential}[1]{\ensuremath{{\mathrm e}^{#1}}}

\newcommand{\iunit}{\ensuremath{\mathrm{i}}}

\newcommand{\bydefinition}{\mathrm{def}}

\newcommand{\diff}{\mathrm{d}}

\renewcommand{\vec}[1]{\ensuremath{\mathbf{#1}}}

\makeatletter
\@ifpackageloaded{bm}%
{\renewcommand{\vec}[1]{\ensuremath{\bm{#1}}}%
}{%
\relax
}
\makeatother

\newcommand{\tensorq}[1]{\ensuremath{\mathbb{#1}}}      
\newcommand{\tensorc}[1]{\ensuremath{\mathrm{#1}}}      

\newcommand{\transpose}[1]{#1^\top}

\newcommand{\inverse}[1]{#1^{-1}}

\newcommand{\identity}{\ensuremath{\tensorq{I}}} 
\makeatletter
\@ifpackageloaded{bm}%
{%
}{%
}

\@ifpackageloaded{bm}%
{%
 
}{%

}

\@ifpackageloaded{bm}%
{%
}{%
}
\makeatother

\newcommand{\generictensor}{{\tensorq{A}}}
\newcommand{\diracdelta}{\delta}

\newcommand{\UnitTriangle}{U_{\mathrm{Triangle}}}

\newcommand{\R}{\ensuremath{{\mathbb R}}}

\newcommand{\N}{\ensuremath{{\mathbb N}}}
\newcommand{\Z}{\ensuremath{{\mathbb Z}}}

\newcommand{\ccontinuous}{c_{\text{C}}}

\newcommand{\pd}[2]{\ensuremath{\frac{\partial {#1}}{\partial {#2}}}}
\newcommand{\ppd}[2]{\ensuremath{\frac{\partial^2 {#1}}{\partial {#2^2}}}}
\newcommand{\dd}[2]{\ensuremath{\frac{\diff {#1}}{\diff {#2}}}}

\newcommand{\ddd}[2]{\ensuremath{\frac{\diff^2 {#1}}{\diff {#2}^2}}}

\makeatletter
\@ifpackageloaded{tensor}
{

}{%

}
\makeatother

\makeatletter
\@ifpackageloaded{tensor}
{

}{%

}
\makeatother

\makeatletter
\@ifundefined{volume}{%
}%
{%
}
\makeatother

\makeatletter

\@ifpackageloaded{MnSymbol} 
{
 
}{%
 
}

\@ifpackageloaded{MnSymbol} 
{
   
}{%
   
}
\makeatother

\newcommand{\vectordot}[2]{\ensuremath{#1 \bullet #2}}

\newcommand{\tensorschur}[2]{\ensuremath{#1 \circ #2}} 

\newcommand{\fouriertransforms}{{\mathcal F}}

\newcommand{\convolution}[2]{#1 \ast #2}

\newcommand{\FourierTransform}[1]{\fouriertransforms \left[#1\right]}
\newcommand{\InverseFourierTransform}[1]{\inverse{\fouriertransforms} \left[#1\right]}

\newcommand{\FourierTransformSemidiscrete}[1]{\fouriertransforms_h \left[#1\right]}
\newcommand{\InverseFourierTransformSemidiscrete}[1]{\inverse{\fouriertransforms}_h \left[#1\right]}

\newcommand{\discreteconvolution}[2]{#1 \ast_h #2}
\renewcommand{\tensorschur}[2]{\ensuremath{#1 \odot #2}} 

\newcommand{\FourierTransformDiscrete}[1]{\fouriertransforms_h^h \left[ #1 \right]}
\newcommand{\InverseFourierTransformDiscrete}[1]{{\inverse{\fouriertransforms}}_h^h \left[ #1 \right]}

\newcommand{\discreteperiodicconvolution}[2]{#1 \ast_h^h #2}

\newcommand{\FourierTransformDiscreteMatrix}{\tensorq{F}_h^h}

\newcommand{\sinetransforms}{\mathcal{S}}

\newcommand{\SineTransformDiscrete}[1]{\sinetransforms_h^h \left[ #1 \right]}
\newcommand{\InverseSineTransformDiscrete}[1]{{\inverse{\sinetransforms}}_h^h \left[ #1 \right]}

\newcommand{\SineTransformDiscreteMatrix}{\tensorq{S}_h^h}

\newcommand{\Sdiracdelta}{S}

\newcommand{\jmatrix}{\tensorq{J}}

\declaretheorem[shaded]{theorem}
\declaretheorem[shaded,sibling=theorem]{definition}
\declaretheorem[shaded,sibling=theorem]{corollary}
\declaretheorem[shaded,sibling=theorem]{lemma}

\numberwithin{equation}{section}

\makeatletter

\providecommand{\solutionname}{Solution}
\makeatother

\definecolor{codegreen}{rgb}{0,0.6,0}
\definecolor{codegray}{rgb}{0.5,0.5,0.5}
\definecolor{codepurple}{rgb}{0.58,0,0.82}
\definecolor{backcolour}{rgb}{0.95,0.95,0.92}

\lstdefinestyle{mycodestyle}{
backgroundcolor=\color{backcolour},   
commentstyle=\color{codegreen},
keywordstyle=\color{magenta},
numberstyle=\tiny\color{codegray},
stringstyle=\color{codepurple},
frame=single,
basicstyle=\ttfamily,
numberstyle=\tiny,
breakatwhitespace=false,         
breaklines=true,                 
captionpos=b,                    
keepspaces=true,                 
numbers=left,                    
numbersep=5pt,                  
showspaces=false,                
showstringspaces=false,
showtabs=false,                  
tabsize=2
}

\numberwithin{equation}{section}

\title[Discrete versus continuous]{Discrete versus continuous---linear lattice models and their exact continuous counterparts}

\date{\today}

\author{Lorenzo Fusi}
\address{
  Dipartimento di Matematica e Informatica ``Ulisse Dini''
  Università degli Studi di Firenze\\
  Viale Morgagni 67/A\\
  Firenze\\
  I 50134\\
  Italy
}
\email{lorenzo.fusi@unifi.it}

\author{Oliver K\v{r}enek}
\address{
Faculty of Mathematics and Physics\\
Charles University\\
Sokolovsk\'a 83\\
Praha 8 -- Karl\'{\i}n\\
CZ 186\;75\\
Czech Republic
}
\email{oliver.krenek741@student.cuni.cz}

\author{V\'{\i}t Pr\r{u}\v{s}a}
\address{
Faculty of Mathematics and Physics\\
Charles University\\
Sokolovsk\'a 83\\
Praha 8 -- Karl\'{\i}n\\
CZ 186\;75\\
Czech Republic
}
\email{prusv@karlin.mff.cuni.cz}

\author{Casey Rodriguez}
\address{
University of North Carolina\\
Department of Mathematics\\
329 Phillips Hall \\
Chapel Hill NC 27599 \\
United States of America
}
\email{crodrig@email.unc.edu}

\author{Rebecca Tozzi}
\address{
  Dipartimento di Matematica e Informatica ``Ulisse Dini''
  Università degli Studi di Firenze\\
  Viale Morgagni 67/A\\
  Firenze\\
  I 50134\\Italy
}
\email{rebecca.tozzi@unifi.it}

\author{Martin Vejvoda}
\address{
Faculty of Mathematics and Physics\\
Charles University\\
Sokolovsk\'a 83\\
Praha 8 -- Karl\'{\i}n\\
CZ 186\;75\\
Czech Republic
}
\email{martin.vejvoda@mff.cuni.cz}

\thanks{V\'{\i}it Pr\r{u}\v{s}a thanks the Czech Science Foundation, grant no. 25-16592S, for its support. Martin Vejvoda thanks the Charles University Grant Agency, grant no. 306-10/252731 and grant no. UNCE24/SCI/005, for its support.}

\begin{document}

\begin{abstract}
  We review and study the correspondence between discrete linear lattice/chain models of interacting particles and their continuous counterparts represented by linear partial differential equations. In particular, we study the correspondence problem for linear nearest neighbour interaction lattice models as well as for linear multiple-neighbour interaction lattice models, while we gradually proceed from infinite lattices to  periodic lattices and finally to finite lattices with fixed ends/zero Dirichlet boundary conditions. The whole study is framed as a systematic specialisation of Fourier analysis tools from the continuous to the discrete setting and vice versa, and the correspondence between the discrete and continuous models is examined primarily with regard to the dispersion relation.


\end{abstract}

\maketitle

\begin{mdframed}[hidealllines=true, backgroundcolor=yellow]
  The manuscript has been published as \href{https://doi.org/10.1016/j.ijengsci.2026.104530}{Lorenzo Fusi, Oliver Křenek, Vít Průša, Casey Rodriguez, Rebecca Tozzi, Martin Vejvoda: Discrete versus continuous—Linear lattice models and their exact continuous counterparts, International Journal of Engineering Science, Volume 224, 2026, 104530}, \href{https://doi.org/10.1016/j.ijengsci.2026.104530}{10.1016/j.ijengsci.2026.104530}. The content of this manuscript version is identical to the published version, but with a more detailed explanations at several places. Furthermore, definitions, theorems and lemmas are typeset in visually distinct boxes in order to improve readability. This manuscript version also contains an Appendix on implementation in \textsc{Matlab} and \textsc{Wolfram Language}.
\end{mdframed}

\tableofcontents


\section{Introduction}
\label{sec:introduction}

We study the correspondence between discrete lattice models and their (quasi)continuous counterparts. This problem is of interest from two viewpoints. First, in continuum mechanics setting, one wants to identify a continuous model (partial differential equation) out of a discrete model describing the interaction between a chain (lattice) of discrete particles provided that interparticle distance tends to zero. This setting is referred to as a \emph{continualisation} problem. Second, in numerical analysis setting, one wants to approximate the solution of a given partial differential equation by the solution to a (finite) system of ordinary differential equations in order to solve the original partial differential equation numerically. This setting is referred to as a \emph{discretisation problem}. Clearly, the two settings are complementary, and it would be interesting to see whether techniques/approaches used in numerical analysis can be exploited in continuum mechanics and \emph{vice versa}. In what follows we answer this call. In particular, we first provide a summary of existing diverse approaches in both fields, and then we gradually proceed to new results concerning the correspondence between the discrete and continuous models. 

In the abstract setting the problem is the following. Assume that the function $u_h(t, x)$ solves a continuous problem
\begin{equation}
  \label{eq:1}
  \ppd{u_h}{t} - \mathcal{L}_h u_h = 0,
\end{equation}
wherein $\mathcal{L}_h$ is a linear operator acting in the spatial/physical domain. Assume further that a (possibly infinite) collection of functions $\{u_{h, j} (t)\}_{j=-\infty}^{+\infty}$ represents grid values of a function, and that the collection $\{u_{h, j} (t)\}_{j=-\infty}^{+\infty}$ solves the discrete system of ordinary differential equations
\begin{subequations}
  \label{eq:2}  
  \begin{equation}
    \label{eq:3}
    \ddd{u_{h, j}}{t} - \sum_{k=-\infty}^{\infty} \tensor{\left(\tensorc{L}_h \right)}{_j_k}u_{h, j} = 0,
  \end{equation}
  wherein $\{\tensor{\left(\tensorc{L}_h \right)}{_j_k}\}_{j, k = -\infty}^{+\infty}$ are elements of (possibly infinite) matrix $\tensorq{L}_h$. (If the elements of collection $\{u_{h, j} (t)\}_{j=-\infty}^{+\infty}$ are interpreted as displacements of discrete particles from their equilibrium positions, then $\tensorc{L}_h$ describes the interactions between the particles in a lattice.)  In the matrix form the problem~\eqref{eq:3} would thus read
  \begin{equation}
    \label{eq:4}
    \ddd{\vec{u}_h}{t} - \tensorq{L}_h \vec{u}_h = \vec{0}.
  \end{equation}
\end{subequations}
Finally, assume that we have a reconstruction procedure that allows us to reconstruct the continuous function~$u_h(t, x)$ out of the discrete grid values $\{u_{h, j} (t)\}_{j=-\infty}^{+\infty}$ and vice versa---one can think of an interpolation and a sampling at a given spatial grid.

If we have such a reconstruction procedure, we can ask whether the function~$u_h(t, x)$ reconstructed from the solution~$\{u_{h, j} (t)\}_{j=-\infty}^{+\infty}$ to the discrete system~\eqref{eq:3} with a given interaction matrix $\tensorq{L}_h$ is an \emph{exact} solution to some partial differential equation of type~\eqref{eq:1}, and whether it is possible to identify the corresponding operator~$\mathcal{L}_h$. This is the~\emph{continualisation problem} wherein the number of particles/distance between the discrete particles $h$ is \emph{a small but yet finite number}, and it is a problem investigated in continuum mechanics.

On the other hand, assume that we want to find a solution $u_h$ to~\eqref{eq:1}, but we can only solve discrete problems of type~\eqref{eq:3}. The question is how should we design the reconstruction procedure and the elements of the interaction matrix~$\tensorq{L}_h$ such that the solution $\tilde{u}_h$ obtained by the reconstruction of the discrete solution inherits qualitative properties of the \emph{exact} solution $u_h$ to the partial differential equation~\eqref{eq:1}. This is the \emph{discretisation problem} encountered in numerical analysis.

Clearly, the \emph{discretisation problem} is similar to the \emph{continualisation problem}. The difference is in what is perceived as the ground truth---the partial differential equation represented by the operator $\mathcal{L}_h$ or the lattice model represented by the matrix~$\tensorq{L}_h$---and what is an approximation---the lattice model represented by the matrix~$\tensorq{L}_h$ or the partial differential equation represented by the operator $\mathcal{L}_h$.

We frame our study of continualisation/discretetisation problem as a study on systematic use of Fourier transform,  while the correspondence between the discrete and continuous models is examined primarily with regard to the dispersion relation/eigenvalues of $\tensorq{L}_h$ and $\mathcal{L}_h$, see Section~\ref{sec:discr-cont-probl} for rationale of this approach. We gradually proceed from infinite lattices, Section~\ref{sec:infinite-lattice}, to periodic lattices, Section~\ref{sec:periodic-lattice}, and finally to finite lattices with fixed ends/zero Dirichlet boundary conditions, Section~\ref{sec:finite-lattice-with}. Most of the Fourier transform based tools used in sections on infinite/periodic lattices are classical, see, for example, \cite{vichnevetsky.r.bowles.jb:fourier}, \cite{boyd.jp:chebyshev} and \cite{trefethen.ln:spectral}, but we introduce them is a coherent form emphasising the interplay between continuous and semidiscrete/discrete versions of Fourier transform. With these known tools we derive \emph{complete} characterisation of the equivalence between continuous and discrete models for \emph{finite} interparticle distance $h$, see Theorem~\ref{thr:2} and Theorem~\ref{thr:3} respectively. To our best knowledge such theorems are not available in the literature.

Interestingly, the Fourier transform based approach can also be extended to \emph{non-periodic finite lattices} with fixed ends/Dirichlet boundary conditions, which is the setting wherein the use of Fourier based methods is rather discouraged, see, for example, \cite[Chapter 3, page 47]{shen.j.tang.t.ea:spectral}. We discuss finite lattices in Section~\ref{sec:finite-lattice-with}, and we again derive \emph{complete} characterisation of the equivalence between continuous and discrete models for \emph{finite} interparticle distance $h$, see Theorem~\ref{thr:4}. (We have complete characterisation for the nearest neighbour interaction model only, for the multiple-neighbour interaction models we only discuss the quality of eigenvalues approximation.) Finally, in Section~\ref{sec:remark-more-complex} we discuss the benefits of Fourier transform based methods in discretisation of a generic regular Sturm--Liouville operator. In particular, we propose a new and rather simple discrete sine transform based method for computation of eigenvalues of Sturm--Liouville operators, while the method seems to be competitive with the state-of-the-art specialised algorithms.

\section{A further look into discretisation and continualisation problems}
\label{sec:discr-cont-probl}
The key issues in the continualisation/discretisation problem are the best documented on a simple one-dimensional infinite lattice model of nearest neighbour interacting particles. Let us thus consider the one-dimensional infinite lattice model of equal mass particles $m$, wherein the nearest neighbours are connected with springs of stiffness $k$, see Figure~\ref{fig:simple-lattice}, while the equilibrium position of $j$-th particle is  $x_{h, j} = jh$, $h$ being the equilibrium interparticle distance. The governing equations for the longitudinal displacements $\{u_{h, j} (t)\}_{j=-\infty}^{+\infty}$ read
\begin{equation}
  \label{eq:5}
  m \ddd{u_{h, j}}{t} = k \left(u_{h, j+1} - 2 u_{h, j} + u_{h, j-1}\right).
\end{equation}
If the individual masses are given in terms of the linear density $\rho$ as $m=\rho h$, where $h$ is the interparticle distance, and if the spring stiffness $k$ scales with the interparticle distance $h$ as $k = \frac{K}{h}$, where $\rho$ and $K$ are some constants, then~\eqref{eq:5} can be rewritten as
\begin{equation}
  \label{eq:6}
  \ddd{u_{h, j}}{t} = \frac{K}{\rho} \frac{u_{h, j+1} - 2 u_{h, j} + u_{h, j-1}}{h^2}
\end{equation}
or in the matrix form as
\begin{equation}
  \label{eq:7}
  \ddd{}{t}
  \begin{bmatrix}
    \vdots \\
    u_{h, -2} \\
    u_{h, -1} \\
    u_{h, 0} \\
    u_{h, 1} \\
    u_{h, 2} \\
    \vdots
  \end{bmatrix}
  =
  \frac{K}{\rho}
  \frac{1}{h^2}
  \begin{bmatrix}
    \ddots & \ddots &  &  &  &  &  \\
    \ddots & -2 & 1 &  &  &  &  \\
           & 1 & -2 & 1 &  &  &  \\
           &  & 1 & -2 & 1 &  &  \\
           &  &  & 1 & -2 & 1 &  \\
           &  &  &  & 1 & -2 & \ddots \\
           &  &  &  &  & \ddots & \ddots
  \end{bmatrix}
  \begin{bmatrix}
    \vdots \\
    u_{h, -2} \\
    u_{h, -1} \\
    u_{h, 0} \\
    u_{h, 1} \\
    u_{h, 2} \\
    \vdots
  \end{bmatrix}
  .
\end{equation}
This is a system of type~\eqref{eq:3}. If the interparticle distance $h$ tends to zero, then one can expect that the behaviour of the discrete lattice model increasingly corresponds to the continuous one-dimensional wave equation
\begin{equation}
  \label{eq:8}
  \ppd{u}{t} = \ccontinuous^2 \ppd{u}{x}
\end{equation}
with the wave speed
\begin{equation}
  \label{eq:9}
  \ccontinuous^2 =_{\bydefinition} \frac{K}{\rho}.  
\end{equation}
(Recall that the fraction on the right-hand side of~\eqref{eq:6} is in fact the centered finite differences formula for the second derivative, hence in the limit $h \to 0+$ it converges to the second derivative $\ppd{}{x}$.) Equation~\eqref{eq:8} is the partial differential equation of type~\eqref{eq:1}. As we have already pointed out, this discrete-to-continuous transition is in the continuum mechanics literature referred to as the \emph{continualisation}, see, for example, \cite{kunin.ia:elastic*1}, \cite{challamel.n.zhang.yp.ea:discrete,challamel.n.picandet.v.ea:revisiting}.

On the other hand, if the problem of interest is the continuous wave equation~\eqref{eq:8}, then~\eqref{eq:6} can be seen as a \emph{spatial discretisation} of the wave equation by the second order centred finite differences scheme, see, for example, \cite{vichnevetsky.r.bowles.jb:fourier}, \cite{trefethen.ln:spectral}, \cite{leveque.rj:finite} or \cite{jovanovic.bs.suli.e:analysis}. (We emphasise that the time variable is kept continuous, thus from the perspective of numerical analysis we deal with a \emph{semi}discretisation of the corresponding equation.) Unfortunately, systems~\eqref{eq:6} and \eqref{eq:8} are equivalent only in the limit $h \to 0+$, while their behaviour is \emph{substantially different} for any finite $h$.

\begin{figure}[t]
  \centering
  \includegraphics[width=0.6\textwidth]{figures/simple-chain-lattice}
  \caption{Simple lattice model.}
  \label{fig:simple-lattice}
\end{figure}

\subsection{Dispersion relation and wave propagation}
\label{sec:disp-relat-wave}

The key issue is that the discrete/continuous models differ regarding the wave propagation. The wave propagation is fully characterised by the~\emph{dispersion relation}, see, for example, \cite{brillouin.l:wave*2}, and the \emph{dispersion relation is different} for the discrete model~\eqref{eq:6} and for the continuous model~\eqref{eq:8}. Indeed, if we are interested in the wave of the form $\exponential{\iunit \left(\xi x - \omega t \right)}$, then a simple substitution into~\eqref{eq:6} and~\eqref{eq:8} reveals that the dispersion relation for the discrete model reads
\begin{subequations}
  \label{eq:10}
  \begin{equation}
    \label{eq:11}
    \omega^2 = 2 \ccontinuous^2 \frac{1 - \cos \left( \xi h \right)}{h^2},
  \end{equation}
  while for the continuous model we get
  \begin{equation}
    \label{eq:12}
    \omega ^2 = \ccontinuous^2 \xi^2,
  \end{equation}
\end{subequations}
see Section~\eqref{sec:search-cont-anal} for details. Clearly, if $h$ is small, then we can approximate the discrete dispersion relation~\eqref{eq:11} as $\omega^2 \approx \ccontinuous^2 \xi^2$, which is the same expression as the continuous dispersion relation~\eqref{eq:12}, and in the limit $h \to 0+$ we arrive to the same dispersion relation. \emph{However, the difference between the discrete dispersion relation~\eqref{eq:11} and the continuous~\eqref{eq:12} prevails for any \emph{finite yet arbitrary small} $h$, and, furthermore, the difference is not uniform with respect to $\xi$}. The agreement between the dispersion relations is good for small wavenumbers $\xi$, but it worsens for large wavenumbers $\xi$.

Consequently, if a precise characterisation of wave motion is necessary, then this difference is a serious \emph{qualitative} problem. For example, the dispersion relation~\eqref{eq:11} describes for any finite yet arbitrary small $h$ a \emph{dispersive wave propagation} wherein the waves of different wavelengths travel at different phase speeds, while the dispersion relation~\eqref{eq:12} predicts the same phase speed for all wavelengths.  

The problem regarding the difference in dispersion relations is known, and much effort has been devoted to mitigate it. The approach to the \emph{continualisation} problem, wherein the discrete lattice model is the ground truth, and the task is to find a continuous version of the lattice model~\eqref{eq:6}, moved in the direction of more involved partial differential equations replacing the standard wave equation~\eqref{eq:8}, see, for example, \cite{challamel.n.picandet.v.ea:revisiting,challamel.n.zhang.yp.ea:discrete}. Popular refined versions of the standard wave equation~\eqref{eq:8} are, for example,
\begin{subequations}
  \label{eq:13}
  \begin{align}
    \label{eq:14}
    \left(
    \ppd{}{t}
    -
    \ccontinuous^2
    \left(
    \ppd{}{x}
    +
    \frac{h^2}{12} \pd{^4}{x^4}
    +
    \cdots
    \right)
    \right)
    u_h
    &=
      0
      ,
    \\
    \label{eq:15}
    \left(
    \left( 1 - \frac{h^2}{12} \ppd{}{x} \right)
    \ppd{}{t}
    -
    \ccontinuous^2 \ppd{}{x}
    \right)
    u_h
    &=
      0
      .
  \end{align}
\end{subequations}
(Note that unlike in the simple wave equation~\eqref{eq:8} the interparticle distance $h$ now explicitly enters the corresponding partial differential equation.) As we shall see later, equations~\eqref{eq:13} are indeed better than the simple wave equation~\eqref{eq:8} in the sense that they lead to dispersion relations that are closer to the discrete dispersion relation~\eqref{eq:11}. 

On the other hand, the \emph{discretisation problem} wherein the ground truth is the partial differential equation (wave equation), and the objective is to develop a more sophisticated discretisation scheme that inherits the qualitative properties of the corresponding partial differential equation, has been addressed in works on numerical analysis, see, for example, \cite{trefethen.ln:group} and~\cite{tam.ckw.webb.jc:dispersion-relation-preserving}.

Having briefly introduced the wave equation/nearest neighbour lattice models, we can now make our quest more concrete. We are interested in the correspondence between lattice models of type~\eqref{eq:3} and the continuous models of type~\eqref{eq:1}, while the quantity of interest is the \emph{dispersion relation}. We shall investigate the problem in one spatial dimension, and we shall proceed from the infinite spatial domain to a spatial domain with periodic boundary conditions and finally to a bounded spatial domain with zero Dirichlet boundary conditions.

\subsection{Eigenvalues approximation} We note that the dispersion relation problem in fact boils down to the analysis of eigenvalues of the operator $\mathcal{L}_h$ and the matrix $\tensorq{L}_h$. Indeed, if we assume that the time dependence of the displacement field is of type $\exponential{\iunit \omega t}$, then~\eqref{eq:1} and~\eqref{eq:2} lead to eigenvalue problems
\begin{subequations}
  \label{eq:16}
  \begin{align}
    \label{eq:17}
    - \omega^2 \widehat{u_h} &= \mathcal{L}_h \widehat{u_h}, \\
    \label{eq:18}
    - \omega^2 \widehat{\vec{u}_h} &= \tensorq{L}_h \widehat{\vec{u}_h},
  \end{align}
\end{subequations}
where $\widehat{u_h}$ and $\widehat{\vec{u}_h}$ denote time Fourier transform of $u_h$ and $\vec{u}_h$. Clearly, the dispersion relation is well-preserved if the eigenvalues of the matrix $\tensorq{L}_h$ match the corresponding eigenvalues of $\mathcal{L}_h$. The question how to efficiently discretise spatial operators with regard to eigenvalue approximation has been studied in numerical analysis, see, for example, \cite{paine.jw.hoog.fr.ea:on}, \cite{andrew.al.paine.jw:correction} and \cite{ledoux.v.van-daele.m.ea:efficient}, and we shall exploit this perspective as well.

\section{Fourier transform}
\label{sec:fourier-transform}
Since we are interested in wave propagation and especially the dispersion relation, our main tool is naturally the Fourier transform. 
The Fourier transform and the inverse Fourier transform are defined as follows.
\begin{definition}[Fourier transform]
  \label{dfn:1}
  Let $f$ and $g$ be functions such that the following integrals are well defined. We denote
  \begin{subequations}
  \label{eq:fourier-transform}
  \begin{align}
    \label{eq:19}
    \FourierTransform{f} (\xi)
    &=_{\bydefinition}
      \int_{x = - \infty}^{+ \infty}
      f(x)
      \exponential{-\iunit \xi x}
      \,
      \diff x
      ,
    \\
    \label{eq:20}
    \InverseFourierTransform{g} (x)
    &=_{\bydefinition}
      \frac{1}{2\pi}
      \int_{\xi = - \infty}^{+ \infty}
      g(\xi)
      \exponential{\iunit \xi x}
      \,
      \diff \xi
      ,
  \end{align}
  and we call the function $\FourierTransform{f}$ the Fourier transform of $f$, while $\InverseFourierTransform{g}$ is called the inverse Fourier transform of $g$.
\end{subequations}
\end{definition}
The definition is that same as in~\cite{trefethen.ln:spectral}, and we follow notation and conventions used in~\cite{trefethen.ln:spectral} whenever possible. (Following the practice in applied mathematics, we also formally employ the same definition even for distributions, in particular for the Dirac distribution. For a rigorous treatment see, for example, \cite{schwartz.l:theorie} or~\cite{grafakos.l:classical}.) Besides the Fourier transform we can also introduce the convolution of two functions.
\begin{definition}[Convolution]
  \label{dfn:2}
  Let $f$ and $g$ be functions such that the following integral is well defined,
\begin{equation}
  \label{eq:21}
  \left( \convolution{f}{g} \right) (x) =_{\bydefinition} \int_{y=-\infty}^{\infty} f(x-y) g(y) \, \diff y.
\end{equation}
The function $\convolution{f}{g}$ is called the convolution of $f$ and $g$.
\end{definition}
Further, let~$\diracdelta(x)$ denotes the Dirac distribution concentrated at the origin. (If no confusion can arise, we do not write the arguments of the functions, for example, we write $\diracdelta$ instead of~$\diracdelta(x)$ and so forth.) We also use the notation $\diracdelta_{y}(x)$ or $\diracdelta_{y}$ for the shifted Dirac distribution $\diracdelta(x - y)$ concentrated at point $y$. We recall that the Dirac distribution is an identity element with respect to the convolution operation,
  $
  \convolution{\diracdelta}{f} = f,
  $
  and that
  \begin{equation}
    \label{eq:28}
    \convolution{\diracdelta_y(x)}{f} (x) = f(x-y).
  \end{equation}
  Finally, we recall that convolution is an associative operation, $\convolution{\left( \convolution{f}{g} \right)}{h} = \convolution{f}{\left(\convolution{g}{h}\right)}$. The Fourier transform has these properties.
  \begin{lemma}[Properties of Fourier transform]
    \label{lm:4}
    Let $f$ and $g$ be functions such that the expressions on both sides of the following equations are well defined, then
    \begin{subequations}
      \label{eq:22}
      \begin{align}
        \label{eq:23}
        \FourierTransform{\dd{^nf}{x^n}} &= \left( \iunit \xi \right)^n \FourierTransform{f}, \\
        \label{eq:24}
        \FourierTransform{f(x-y)} &= \exponential{- \iunit \xi y} \FourierTransform{f}, \\
        \label{eq:25}
      \FourierTransform{\convolution{f}{g}} &= \FourierTransform{f} \FourierTransform{g}, \\
        \label{eq:26}
        \FourierTransform{f g}  &= \convolution{\FourierTransform{f}}{\FourierTransform{g}}, \\
        \label{eq:27}
        1 &=  \FourierTransform{ \diracdelta}, \\
        \label{eq:29}
        \FourierTransform{\exponential{\iunit \eta x}} &= 2 \pi \diracdelta_{\eta}.
      \end{align}
    \end{subequations}
  \end{lemma}
  Formulae summarised in~\eqref{eq:22} are the well known derivative, shift and convolution-to-multiplication theorems for the Fourier transform, see, for example, \cite{bracewell.rn:fourier} or~\cite{grafakos.l:classical}.

A notable property of the Fourier transform is that it ``diagonalises'' all derivatives in the sense that the differentiation reduces to the multiplication by powers of the wavenumber $\xi$. In other words, we can formally solve the eigenvalue problem
\begin{equation}
  \label{eq:30}
  \dd{^nf}{x^n} = \lambda f
\end{equation}
as follows. We take the Fourier transform of~\eqref{eq:30}, which reduces it to
\begin{equation}
  \label{eq:31}
  \left(\left( \iunit \xi \right)^n - \lambda\right) \FourierTransform{f} = 0,
\end{equation}
where we have used the derivative rule~\eqref{eq:23}. This equation in the Fourier space is formally solved provided that we set
\begin{equation}
  \label{eq:32}
  \FourierTransform{f} = \diracdelta_\eta (\xi), 
\end{equation}
where $\diracdelta_\eta (\xi)$ denotes the shifted Dirac function in the Fourier space, $\diracdelta_\eta (\xi)= \diracdelta(\xi - \eta)$, with $\eta$ being a fixed but otherwise arbitrary wavenumber $\eta \in \R$. The corresponding eigenvalue is then
\begin{equation}
  \label{eq:33}
  \lambda = \left(\iunit \eta \right)^n.  
\end{equation}
The eigenfunctions $f$ in problem~\eqref{eq:30} are then obtained by the inverse Fourier transform of $\diracdelta_\eta (\xi)$, see~\eqref{eq:32}, and we get
\begin{equation}
  \label{eq:34}
  f = \frac{1}{2 \pi} \exponential{\iunit \eta x}.
\end{equation}
Note that \emph{all derivative operators} $\dd{^n}{x^n}$, $n \in \N$, \emph{share the same eigenfunctions} $f$. 

Furthermore, the formal eigenvalue problem for the linear operator defined as the convolution with a given function $g$ reads
\begin{equation}
  \label{eq:35}
  \convolution{g}{f} = \lambda f, 
\end{equation}
and we can solve this eigenvalue problem using the Fourier transform as well. We take the Fourier transform, and we get
\begin{equation}
  \label{eq:36}
  \left(\FourierTransform{g} - \lambda\right) \FourierTransform{f} = 0
\end{equation}
where we have used the convolution-to-multiplication property~\eqref{eq:25}. This equation in the Fourier space is formally solved provided that we again set
\begin{equation}
  \label{eq:37}
  \FourierTransform{f} = \diracdelta_\eta (\xi),
\end{equation}
while the corresponding eigenvalue is now obtained by the evaluation of Fourier transform of $g$ and the corresponding wavenumber,
\begin{equation}
  \label{eq:38}
  \lambda = \left. \FourierTransform{g} \right|_{\xi = \eta}.
\end{equation}
The operator ``convolution with a known function $g$'' thus formally has the same eigenfunctions as (all) derivative operators of arbitrary order, the operators differ in eigenvalues only. This is not surprising since the derivative of a function can be in fact expressed as the convolution with the derivative of the Dirac distribution, for example
\begin{equation}
  \label{eq:39}
  \ddd{f}{x} = \convolution{\ddd{\diracdelta}{x}}{f}.
\end{equation}

Naturally all the manipulations shown above are purely formal as work with the Dirac distribution in a rather naive way. However, the formal infinite dimensional manipulations directly suggest which manipulations could be of interest on the discrete (finite dimensional) level, and they give us hints concerning the eigenvalues/eigenvectors of the corresponding finite dimensional operators (matrices). As we shall see preserving an analogue of the ``diagonalisation'' property of Fourier transform on the discrete level is crucial in order to get a good correspondence between the discrete and continuous dispersion relation. Furthermore, the discrete versions (finite dimensional setting) of the ``diagonalisation'' property are easy to make rigorous.

\section{Infinite lattice}
\label{sec:infinite-lattice}

\subsection{Infinite lattice model of interacting particles}
\label{sec:lattice-model}
We now consider an infinite lattice (chain) of equal mass particles that are in the equilibrium positioned at \emph{equispaced grid points}  $\{x_{h, j}\}_{j=-\infty}^{+\infty}$,
\begin{equation}
  \label{eq:40}
  x_{h, j} = jh,
\end{equation}
$j \in \Z$, with a fixed equilibrium interparticle distance $h$, see Figure~\ref{fig:simple-lattice}. We assume that the longitudinal displacements of the individual particles $\{u_{h, j} (t)\}_{j=-\infty}^{+\infty}$ are given by the following (infinite) system of ordinary differential equations,
\begin{equation}
  \label{eq:41}
  \ddd{u_{h,j}}{t}
  -
  \sum_{m=-\infty}^{+ \infty} c_{h, j-m}u_{h, m}
  =
  0
\end{equation}
with the property $c_{h, i} = c_{h, -i}$, that is by the system
\begin{equation}
  \label{eq:42}
  \ddd{}{t}
  \begin{bmatrix}
    \vdots \\
    u_{h, -2} \\
    u_{h, -1} \\
    u_{h, 0} \\
    u_{h, 1} \\
    u_{h, 2} \\
    \vdots
  \end{bmatrix}
  -
  \begin{bmatrix}
    \ddots & \vdots & \vdots & \vdots & \vdots & \vdots & \reflectbox{$\ddots$} \\
    \cdots & c_{h, 0} & c_{h, 1} & c_{h, 2} & c_{h, 3} & c_{h, 4} & \cdots \\
    \cdots & c_{h, 1} & c_{h, 0} & c_{h, 1} & c_{h, 2} & c_{h, 3} & \cdots \\
    \cdots & c_{h, 2} & c_{h, 1} & c_{h, 0} & c_{h, 1} & c_{h, 2} & \cdots \\
    \cdots & c_{h, 3} & c_{h, 2} & c_{h, 1} & c_{h, 0} & c_{h, 1} & \cdots \\
    \cdots & c_{h, 4} & c_{h, 3} & c_{h, 2} & c_{h, 1} & c_{h, 0} & \cdots \\
    \reflectbox{$\ddots$} & \vdots & \vdots & \vdots & \vdots & \vdots & \ddots
  \end{bmatrix}
  \begin{bmatrix}
    \vdots \\
    u_{h, -2} \\
    u_{h, -1} \\
    u_{h, 0} \\
    u_{h, 1} \\
    u_{h, 2} \\
    \vdots
  \end{bmatrix}
  =
  \begin{bmatrix}
    \vdots \\
    0 \\
    0 \\
    0 \\
    0 \\
    0 \\
    \vdots
  \end{bmatrix}
  ,
\end{equation}
which we also write in the matrix--vector form as
\begin{equation}
  \label{eq:43}
  \ddd{\vec{u}_h}{t} - \tensorq{L}_h^{\text{infinite}} \vec{u}_h = \vec{0}
\end{equation}
with the obvious identification of (infinite) vector $\vec{u}_h$ and matrix $\tensorq{L}_h^{\text{infinite}}$. The symmetry $c_{h, i} = c_{h, -i}$ and the convolution-type sum in~\eqref{eq:41} follow from basic physics considerations. Namely, we want the interparticle interactions to be symmetric meaning that the interaction of $j$-th particle with its $m$-th neighbour to the left should be the same as its interaction to its $m$-th neighbour to the right. Furthermore, the interactions should be translationally invariant meaning that the interaction between the particles should depend only on their mutual position. We note that these physics based requirements yield a convenient structure from the mathematical point of view. The infinite matrix in~\eqref{eq:42} is an \emph{infinite circulant matrix}, see~\cite{davis.pj:circulant}, and it represents a Laurent operator, see, for example, \cite[Chapter 3]{gohberg.i.goldberg.s.ea:basic}.

The coefficients $\left\{ c_{h, i} \right\}_{i=1}^{+\infty}$ can be indeed put into one-to-one correspondence with spring stiffnesses $\left\{ k_{h, i} \right\}_{k=1}^{+\infty}$. If the $j$-th particle in the lattice is connected to its nearest neighbours $j+1$ and $j-1$ with the spring of stiffness $k_{h, 1}$, its second nearest neighbours $j+2$ and $j-2$ with the spring of stiffness $k_{h, 2}$ and so forth, then the governing equation for $j$-th particle displacement reads
\begin{equation}
  \label{eq:44}
  m \ddd{u_{h, j}}{t}
  =
  \sum_{m=1}^{+\infty}
  \left(
    k_{h, m} \left(u_{h, j+m} - u_{h, j} \right)
    -
    k_{h, m} \left(u_{h, j} - u_{h, j-m} \right)
  \right),
\end{equation}
which yields the system of equations
\begin{equation}
  \label{eq:45}
  \ddd{}{t}
  \begin{bmatrix}
    \vdots \\
    u_{h, -2} \\
    u_{h, -1} \\
    u_{h, 0} \\
    u_{h, 1} \\
    u_{h, 2} \\
    \vdots
  \end{bmatrix}
  =
  \begin{bmatrix}
    \ddots & \vdots & \vdots & \vdots & \vdots & \vdots & \reflectbox{$\ddots$} \\
    \cdots & -2 \sum_{j=1}^{+\infty} k_{h, j} & k_{h, 1} & k_{h, 2} & k_{h, 3} & k_{h, 4} & \cdots \\
    \cdots & k_{h, 1} & -2 \sum_{j=1}^{+\infty} k_{h, j} & k_{h, 1} & k_{h, 2} & k_{h, 3} & \cdots \\
    \cdots & k_{h, 2} & k_{h, 1} & -2 \sum_{j=1}^{+\infty} k_{h, j} & k_{h, 1} & k_{h, 2} & \cdots \\
    \cdots & k_{h, 3} & k_{h, 2} & k_{h, 1} & -2 \sum_{j=1}^{+\infty} k_{h, j} & k_{h, 1} & \cdots \\
    \cdots & k_{h, 4} & k_{h, 3} & k_{h, 2} & k_{h, 1} & -2 \sum_{j=1}^{+\infty} k_{h, j} & \cdots \\
    \reflectbox{$\ddots$} & \vdots & \vdots & \vdots & \vdots & \vdots & \ddots
  \end{bmatrix}
  \begin{bmatrix}
    \vdots \\
    u_{h, -2} \\
    u_{h, -1} \\
    u_{h, 0} \\
    u_{h, 1} \\
    u_{h, 2} \\
    \vdots
  \end{bmatrix}
  .
\end{equation}
Comparing~\eqref{eq:42} and~\eqref{eq:45} we thus see that the coefficients $\left\{ c_{h, i} \right\}_{i=1}^{+\infty}$ are related to the spring stiffnesses as
\begin{equation}
  \label{eq:46}
  c_{h, i}
  =
  \begin{cases}
    k_{h, i}, & i = 1, \dots, +\infty, \\
    -2 \sum_{j=1}^{+ \infty} k_{h, j}, & i=0.
  \end{cases}
\end{equation}
This places yet another structural restriction on the coefficients  $\left\{ c_{h, i} \right\}_{i=1}^{+\infty}$.

\subsection{Search for continuous analogue of lattice model---case study for the nearest neighbour interaction lattice model}
\label{sec:search-cont-anal}

For example, if we consider the nearest neighbour interaction as in~\eqref{eq:6}, we see that the coefficients $c_{h, j}$ are given by $c_{h, 0}  = 2 \frac{K}{\rho h^2}$ and $c_{h, \pm 1} = - \frac{K}{\rho h^2}$ and all remaining coefficients vanish. In this case~\eqref{eq:42} thus reduces to the system that involves the well-known banded matrix representing the centred second order finite differences discretisation of the second derivative operator,    
\begin{equation}
  \label{eq:47}
  \ddd{}{t}
  \begin{bmatrix}
    \vdots \\
    u_{h, -2} \\
    u_{h, -1} \\
    u_{h, 0} \\
    u_{h, 1} \\
    u_{h, 2} \\
    \vdots
  \end{bmatrix}
  -
  \frac{\ccontinuous^2}{h^2}
  \begin{bmatrix}
    \ddots & \ddots &  &  &  &  &  \\
    \ddots & -2 & 1 &  &  &  &  \\
           & 1 & -2 & 1 &  &  &  \\
           &  & 1 & -2 & 1 &  &  \\
           &  &  & 1 & -2 & 1 &  \\
           &  &  &  & 1 & -2 & \ddots \\
           &  &  &  &  & \ddots & \ddots
  \end{bmatrix}
  \begin{bmatrix}
    \vdots \\
    u_{h, -2} \\
    u_{h, -1} \\
    u_{h, 0} \\
    u_{h, 1} \\
    u_{h, 2} \\
    \vdots
  \end{bmatrix}
  =
  \begin{bmatrix}
    \vdots \\
    0 \\
    0 \\
    0 \\
    0 \\
    0 \\
    \vdots
  \end{bmatrix}
  ,
\end{equation}
where we have used the notation~\eqref{eq:9}. Let us for a moment focus on this simple system, and let us recall some standard manipulations. Assume that we have a real valued function $u(t, x)$, and let us naively reformulate the equation~\eqref{eq:6} for the function $u(t, x)$ as
\begin{equation}
  \label{eq:48}
  \ppd{}{t} u
  -
  \ccontinuous^2
  \convolution{\left( \frac{\diracdelta_{h} - 2 \diracdelta + \diracdelta_{-h}}{h^2} \right)}{u} = 0.    
\end{equation}
This makes sense provided that the function $u(t, x)$ is chosen in such a way that its sampling at the grid points $\left. u(t, x) \right|_{x = x_{h, j}}$ coincides with the values $\{u_{h, j} (t)\}_{j=-\infty}^{+\infty}$, that is
\begin{equation}
  \label{eq:49}
  \left. u(t, x) \right|_{x = x_{h, j}} = u_{h, j} (t).
\end{equation}
If~\eqref{eq:49} holds, then sampling of~\eqref{eq:48} at $x_{h, j}$ gives~\eqref{eq:6}. The application of \emph{continuous} Fourier transform to~\eqref{eq:48} then yields
\begin{equation}
  \label{eq:50}
  \left(
    \ppd{}{t}
    -
    \frac{\ccontinuous^2}{h^2}
    \left(
      \exponential{\iunit h \xi}
      -
      2
      +
      \exponential{-\iunit h \xi}
    \right)
  \right)
  \FourierTransform{u}
  =
  0
  ,
\end{equation}
where we have used identities~\eqref{eq:22}, in particular the convolution-to-multiplication property and the shift property that gives
\begin{equation}
  \label{eq:51}
  \FourierTransform{\convolution{\diracdelta_{\pm h}}{u}}
  =
  \FourierTransform{\diracdelta_{\pm h}} \FourierTransform{u}
  =
  \exponential{\mp \iunit h \xi} \FourierTransform{u}
  .
\end{equation}
Equation~\eqref{eq:50} in the Fourier domain can be further rewritten as
\begin{equation}
  \label{eq:52}
  \left(
    \ppd{}{t}
    +
    \ccontinuous^2
    \frac{
      2
      \left(
        1
        -
        \cos \left(h \xi \right)
      \right)
    }
    {
      h^2
    }
  \right)
  \FourierTransform{u}
  =
  0
  ,
\end{equation}
and one more application of Fourier transform---now with respect to the time variable---then yields the equation
\begin{equation}
  \label{eq:53}
  \left(
    -\omega^2
    +
    \ccontinuous^2
    \frac{
      2
      \left(
        1
        -
        \cos \left(h \xi \right)
      \right)
    }
    {
      h^2
    }
  \right)
  \fouriertransforms_{t \to \omega}\left[\FourierTransform{u} \right]
  =
  0
  .
\end{equation}
From~\eqref{eq:53} we get the well-known dispersion relation~\eqref{eq:11}, that is
$
\omega^2 = \ccontinuous^2
\frac{
  2
  \left(
    1
    -
    \cos \left(h \xi \right)
  \right)
}
{
  h^2
}
$,
which can be also rewritten as
\begin{equation}
  \label{eq:54}
  \omega^2
  =
  4
  \ccontinuous^2
  \sin^2 \left( \frac{h \xi}{2} \right)
  .
\end{equation}
(Note that the application of Fourier transform can be equally interpreted as substituting a wave-like \emph{ansatz} $u(t, x) = \widehat{u}(\omega, \xi) \exponential{\iunit \left(\xi x - \omega t \right)}$ into~\eqref{eq:48}.) On the other hand, the Fourier transformed version of the standard wave equation~\eqref{eq:8} reads
\begin{equation}
  \label{eq:55}
  \left(
    -
    \omega^2
    +
    \ccontinuous^2
    \xi^2
  \right)
  \fouriertransforms_{t \to \omega}\left[\FourierTransform{u} \right]
  =
  0
  ,
\end{equation}
and the dispersion relation is~\eqref{eq:12}, that is
\begin{equation}
  \label{eq:56}
  \omega^2
  =
  \ccontinuous^2
  \xi^2
  .
\end{equation}

As we have already noted, the dispersion relations for the lattice model~\eqref{eq:6} and the standard wave equation~\eqref{eq:8} are different. The continuum mechanics community addressing the \emph{continualisation problem} tried to resolve this issue by designing an alternative partial differential equation whose dispersion relation would be closer to the discrete dispersion relation~\eqref{eq:54}. The first straightforward idea in this direction is to use the formal formula for the shift operator, namely 
\begin{equation}
  \label{eq:57}
  u(x + h)
  =
  \exponential{h\pd{}{x}} u(x).
\end{equation}
(We note that this is just a formal way how to rewrite the Taylor expansion $ u(x + h)
=
\sum_{j=0}^{+ \infty}
\frac{1}{j!}
\pd{^ju}{x^j}
h^j$.)
With this formula we can reformulate the discrete system of equations~\eqref{eq:6} as
\begin{equation}
  \label{eq:58}
  \ppd{u}{t}
  -
  \ccontinuous^2
  \frac{
    \exponential{h\pd{}{x}}
    -
    2
    +
    \exponential{-h\pd{}{x}}
  }
  {
    h^2
  }
  u
  =
  0.
\end{equation}
(Compare with convolution based reformulation in~\eqref{eq:48}.) Now we can use only the first few terms in the Taylor expansion of the exponential, which yields
\begin{equation}
  \label{eq:59}
  \left(
    \ppd{}{t}
    -
    \ccontinuous^2
    \left(
      \ppd{}{x}
      +
      \frac{h^2}{12} \pd{^4}{x^4}
      +
      \cdots
    \right)
  \right)
  u
  =
  0
  ,
\end{equation}
which gives us ``enriched continua'' partial differential equations of type~\eqref{eq:14}. Yet another possibility is to rewrite the term with exponentials as
\begin{equation}
  \label{eq:60}
  \exponential{h\pd{}{x}}
  -
  2
  +
  \exponential{-h\pd{}{x}}
  =
  4
  \sinh^2 \left(\frac{h}{2} \pd{}{x} \right)
\end{equation}
and search for more sophisticated approximations beyond Taylor series. For example, Pad\'e approximation of $\sinh^2 x$ of order two reads $\sinh^2 x \approx \frac{x^2}{1 - \frac{x^2}{3}}$ which allows us to formally write 
\begin{equation}
  \label{eq:61}
  \sinh^2 \left(\frac{h}{2} \pd{}{x} \right)
  \approx
  \frac{ \left( \frac{h}{2} \pd{}{x} \right)^2}{ 1 - \frac{\left(\frac{h}{2} \pd{}{x} \right)^2}{3}},
\end{equation}
and using this formula in~\eqref{eq:58} formally yields~\eqref{eq:15}, that is
\begin{equation}
  \label{eq:62}
  \left(
    \left( 1 - \frac{h^2}{12} \ppd{}{x} \right)
    \ppd{}{t}
    -
    \ccontinuous^2 \ppd{}{x}
  \right)
  u
  =
  0
  .
\end{equation}
We note that all manipulations outline above might be done in the Fourier domain as well. In particular, the formal Taylor expansion
\begin{equation}
  \label{eq:63}
  \frac{
    \exponential{h\pd{}{x}}
    -
    2
    +
    \exponential{-h\pd{}{x}}
  }
  {
    h^2
  }
  =
  \ppd{}{x}
  +
  \frac{h^2}{12} \pd{^4}{x^4}
  +
  \cdots
\end{equation}
can be seen as the physical domain counterpart of the Fourier domain expansion
\begin{equation}
  \label{eq:64}
  \frac{
    2
    \left(
      1
      -
      \cos \left(h \xi \right)
    \right)
  }
  {
    h^2
  }
  =
  \xi^2
  -
  \frac{h^2}{12}
  \xi^4
  +
  \cdots
  ,
\end{equation}
see~\eqref{eq:52}. This observation also reveals that the partial differential equations~\eqref{eq:59} and~\eqref{eq:62} indeed lead to better approximations of the discrete dispersion relation~\eqref{eq:11}. In particular, partial differential equation~\eqref{eq:59} yields the dispersion relation
\begin{subequations}
  \label{eq:65}     
  \begin{align}
    \label{eq:66}
    \omega^2 &= \ccontinuous^2
               \left(
               \xi^2
               -
               \frac{h^2}{12}
               \xi^4
               +
               \cdots
               \right)
               ,\\
    \intertext{while partial differential equation~\eqref{eq:62} yields the dispersion relation}
    \label{eq:67}
    \omega^2 &=  \ccontinuous^2
               \frac{\xi^2}{1 + \frac{h^2}{12} \xi^2}.
  \end{align}
\end{subequations}

However, partial differential equations~\eqref{eq:59} and \eqref{eq:62} are still only approximations. If the ground truth is the system of ordinary differential equations~\eqref{eq:6}, then neither the standard wave equation~\eqref{eq:8} nor its improvements \eqref{eq:59} and \eqref{eq:62} \emph{exactly} correspond to the behaviour of the discrete lattice model, and we are no closer to the exact characterisation of the correspondence between the discrete and continuous problems~\eqref{eq:1} and \eqref{eq:3}.

\subsection{Reconstruction procedure in detail---from discrete grid values to a function of real variable via bandwidth limited interpolant}
\label{sec:reconstr-proc-from}
In order to identify the exact continuous counterpart of a lattice model, we first need to carefully discuss the reconstruction procedure, that is the way how to reconstruct a real valued function $u_h(t, x)$ out of the discrete grid values $\vec{u}_h = \{u_{h, j} (t)\}_{j=-\infty}^{+\infty}$. Clearly, the continuous counterpart of a lattice model crucially depends on a particular reconstruction procedure, as different reconstruction procedures lead to different continuous versions of the lattice model. The reconstruction problem is schematically shown in Figure~\eqref{fig:interpolant-construction}, and there exist many possibilities how to solve the reconstruction problem. However, concerning the infinite lattice problem, it turns out that one reconstruction procedure might be denoted as a canonical one due to its extreme convenience.

\begin{figure}
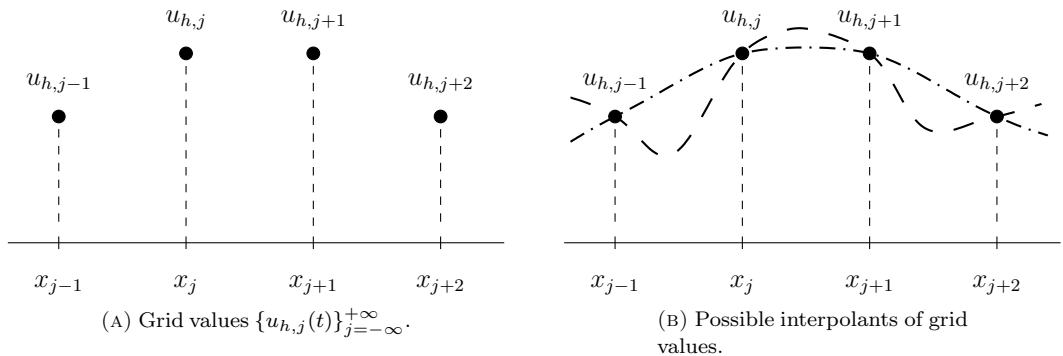

  \centering
  \subfloat[Grid values $\{u_{h, j} (t)\}_{j=-\infty}^{+\infty}$.]{\includegraphics[width=0.35\textwidth]{figures/bandwidth-limited-interpolant}}
  \qquad
  \subfloat[Possible interpolants of grid values.]{\includegraphics[width=0.35\textwidth]{figures/bandwidth-limited-interpolant-full}}
  \caption{Construction of a function $u_h(t, x)$ out of grid values.}
  \label{fig:interpolant-construction}
\end{figure}

\subsubsection{Semidiscrete Fourier transform and bandwidth limited interpolant}
\label{sec:bandw-limit-interp}
The convenient reconstruction procedure is based on the concept of \emph{bandwidth limited interpolant} as this reconstruction is \emph{closely related to the Fourier transform}, which is in turn critical for the dispersion relation. The first observation we can make regarding the wave motion is the following. If we have a wave with wavenumber $\xi + l \frac{2\pi}{h}$, where $l \in \Z$, then
\begin{equation}
  \label{eq:68}
  \left.
    \exponential{\iunit \left(\xi +   l \frac{2\pi}{h} \right) x}
  \right|_{x = x_{h, j}}
  =
  \exponential{\iunit \xi x_{h, j} + \iunit   l \frac{2\pi}{h}  j h}
  =
  \exponential{\iunit 2 lj \pi}
  \left.
    \exponential{\iunit \xi  x}
  \right|_{x = x_{h, j}}
  =
  \left.
    \exponential{\iunit \xi  x}
  \right|_{x = x_{h, j}}
\end{equation}
holds for all $j \in \Z$. This means that two waves with wavenumbers $\xi$ and $\xi + l \frac{2\pi}{h}$, where $l \in \Z$, are \emph{indistinguishable on the equispaced spatial grid} $\left\{x_{h, j} \right\}_{j = -\infty}^{+ \infty}$, albeit the two waves \emph{do differ on the full real line}. Consequently, we can restrict our wavenumbers $\xi$ to an interval of length $\frac{2 \pi}{h}$ as these wavenumbers are sufficient to represent wave with any wavenumber provided that we work with sampling on the equispaced grid  $\left\{x_{h, j} \right\}_{j = -\infty}^{+ \infty}$. This observation is in numerical analysis referred to as the \emph{aliasing phenomenon}, see, for example, \cite[Chapter 2]{trefethen.ln:spectral}, while the physics community the observation leads to the concept of \emph{Brillouin zone}, see, for example, \cite[Chapter II]{brillouin.l:wave*1} and~\cite[Section 2.2]{kunin.ia:elastic*1}. Now we can adjust our fully continuous Fourier transform apparatus accordingly.

Instead of continuous Fourier transform~\eqref{eq:fourier-transform} of a function $u$ we introduce the \emph{semidiscrete Fourier transform}, which works with the (infinite) vector of grid values $\vec{u}_h =  \{u_{h, j} \}_{j=-\infty}^{+\infty}$. The semidiscrete Fourier transform and its inverse are defined as follows.
\begin{definition}[Semidiscrete Fourier transform]
  \label{dfn:3}
  Let $\vec{u}_h =  \{u_{h, j} \}_{j=-\infty}^{+\infty}$ be a vector of grid values and let $v_h$ be a function such that the following integrals are well defined. We denote
  \begin{subequations}
    \label{eq:semidiscrete-fourier-transform}
    \begin{align}
      \label{eq:69}
      \FourierTransformSemidiscrete{\vec{u}_h} (\xi)
      &=_{\bydefinition}
        \left(
        h
        \sum_{j= - \infty}^{+ \infty}
        \tensor{\left(\vec{u}_h\right)}{_j}
        \exponential{- \iunit \xi x_{h, j}}
        \right)
        \chi_{\xi \in \left[ -\frac{\pi}{h}, \frac{\pi}{h} \right]},
      \\
      \label{eq:70}
      \tensor{\left( \InverseFourierTransformSemidiscrete{v_h} \right)}{_j}
      &=_{\bydefinition}
        \frac{1}{2\pi}
        \int_{\xi = - \frac{\pi}{h}}^{\frac{\pi}{h}}
        v_h(\xi)
        \exponential{\iunit \xi x_{h, j}}
        \,
        \diff \xi, \qquad j=-\infty, \dots, +\infty,
    \end{align}
  \end{subequations}
  where we use the notation $\tensor{\left(\vec{u}_h\right)}{_j} = u_{h, j}$ for the value at the $j$-th grid point and similarly for $\tensor{\left( \InverseFourierTransformSemidiscrete{v_h} \right)}{_j}$ for the $j$-th wavenumber coefficient. Furthermore, $\chi_{\xi \in \left[ -\frac{\pi}{h}, \frac{\pi}{h} \right]}$ denotes the characteristic function of the interval $\left[ -\frac{\pi}{h}, \frac{\pi}{h} \right]$. Function $\FourierTransformSemidiscrete{\vec{u}_h} (\xi)$ is called the semidiscrete Fourier transform of vector $\vec{u}_h$, while the vector $\InverseFourierTransformSemidiscrete{v_h}$ with the elements $\left\{ \tensor{\left( \InverseFourierTransformSemidiscrete{v_h} \right)}{_j} \right\}_{j=-\infty}^{+\infty}$ is called the inverse semidiscrete Fourier transform of function $v_h$.  
\end{definition}
 Clearly~\eqref{eq:69} is an analogue of continuous Fourier transform~\eqref{eq:19} wherein the integral is replaced by the midpoint rule, and wherein the wavenumbers are restricted to the corresponding interval. The inverse semidiscrete Fourier transform formula~\eqref{eq:70} is just a special version of the full inverse Fourier transform~\eqref{eq:20} restricted to the corresponding wavenumber interval. The fact that the formula~\eqref{eq:70} is indeed an inverse to~\eqref{eq:69} is easy to check by a simple computation, see, for example, \cite{bracewell.rn:fourier} and~\cite{trefethen.ln:spectral}.

The construction of bandwidth limited interpolant then works as shown in Figure~\eqref{fig:semidiscrete-fourier-transform}. We start with the grid values $\vec{u}_h$ and we do \emph{semidiscrete Fourier transform} according to definition~\eqref{eq:69}. The object we get is a function of $\xi$ defined for \emph{all}~$\xi \in \R$. The function is identically zero outside the interval $\left[ -\frac{\pi}{h}, \frac{\pi}{h} \right]$, but it is anyway defined for all $\xi \in \R$. Such a function in the Fourier domain is a legitimate input to the \emph{continuous} inverse Fourier transform. We thus use the continuous inverse Fourier transform~\eqref{eq:20}, and we get a function $u_h$ that is defined \emph{everywhere} in the spatial domain, that is for all~$x \in \R$---this function is the \emph{bandwidth limited interpolant of discrete grid values~$\vec{u}_h$}.

The whole construction of the bandwidth limited interpolant $u_h$ of grid values $\vec{u}_h$ thus reads
\begin{equation}
  \label{eq:bandwidth-limited-interpolant-semidiscrete}
  u_h = \InverseFourierTransform{\FourierTransformSemidiscrete{\vec{u}_h}}.
\end{equation}
We note that the construction of the bandwidth limited interpolant implies that
\begin{equation}
  \label{eq:71}
  \FourierTransform{u_h} = \FourierTransformSemidiscrete{\vec{u}_h}.
\end{equation}

\begin{figure}
  \centering
  \includegraphics[width=\textwidth]{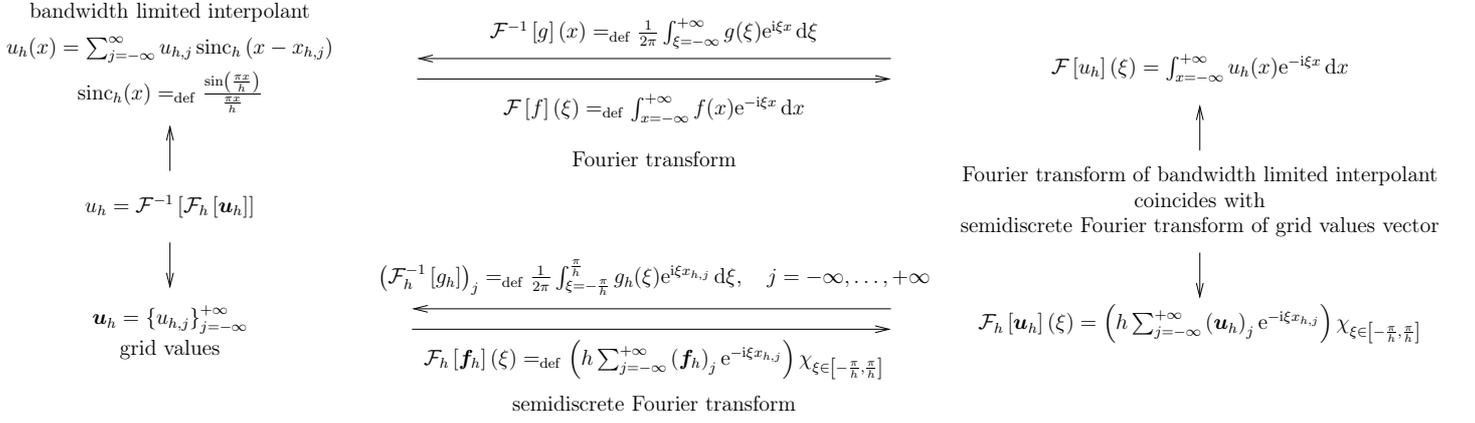}
  \caption{Semidiscrete Fourier transform and construction of bandwidth limited interpolant $u_h$ of grid values vector $\vec{u}_h$, equispaced grid $x_{h, j} = jh$, $j \in \Z$.}
  \label{fig:semidiscrete-fourier-transform}
\end{figure}

As indicated in Figure~\ref{fig:semidiscrete-fourier-transform}, the bandwidth limited interpolant can be constructed without an explicit reference to the Fourier transform as in~\eqref{eq:bandwidth-limited-interpolant-semidiscrete}. The bandwidth limited interpolant $u_h$ is given by the formula
\begin{subequations}
  \label{eq:bandwidth-limited-interpolant-semidiscrete-physical-space}
  \begin{align}
    \label{eq:72}
    u_h(x)
    &=
      \sum_{j=-\infty}^{+ \infty}\tensor{\left(\vec{u}_h\right)}{_j} \sinc_h \left( x - x_{h, j}\right)
      ,
    \\
    \label{eq:73}
    \sinc_h(x) &=_{\bydefinition} \frac{\sin \left( \frac{\pi x}{h} \right)}{\frac{\pi x}{h}},
  \end{align}
\end{subequations}
where the $\sinc$ function, $\sinc x =_{\bydefinition} \frac{\sin x}{x}$, is continuously defined at zero, that is $\left. \sinc x \right|_{x=0} = 1$. This is a classical result, see, for example, \cite[Section 1.10]{vichnevetsky.r.bowles.jb:fourier}, \cite[Chapter 2]{trefethen.ln:spectral}, \cite[Chapter 5]{boyd.jp:chebyshev} to name a few standard reference works. The explicit formula for the bandwidth limited interpolant~\eqref{eq:bandwidth-limited-interpolant-semidiscrete-physical-space} follows from the formula for the semidiscrete Fourier transform of $k$-th discrete Dirac distribution $\vec{\diracdelta}_h^k$, where the $k$-th discrete Dirac distribution $\vec{\diracdelta}_h^k$ is defined as
\begin{equation}
  \label{eq:74}
  \vec{\diracdelta}_h^k =_{\bydefinition}
  \begin{bmatrix}
    \vdots \\
    0 \\
    1 \\
    0 \\
    \vdots \\
  \end{bmatrix}
  ,
\end{equation}
where the nonzero element is placed at the $k$-th position. The semidiscrete Fourier transform~\eqref{eq:69} of $\vec{\diracdelta}_h^k$ yields
\begin{equation}
  \label{eq:75}
  \FourierTransformSemidiscrete{\vec{\diracdelta}_h^k} (\xi)
  =
  h
  \exponential{- \iunit \xi x_{h, k}}
  \chi_{\xi \in \left[ -\frac{\pi}{h}, \frac{\pi}{h} \right]},
\end{equation}
and the continuous inverse Fourier transform then gives
\begin{equation}
  \label{eq:76}
  \InverseFourierTransform{\FourierTransformSemidiscrete{\vec{\diracdelta}_h^k}}(x)
  =
  \InverseFourierTransform{
    h
    \exponential{- \iunit \xi x_{h, k}}
    \chi_{\xi \in \left[ -\frac{\pi}{h}, \frac{\pi}{h} \right]}
  }
  =
  \InverseFourierTransform{
    h
    \exponential{- \iunit \xi x_{h, k}}
    \FourierTransform{ \frac{1}{h} \sinc\left( \frac{\pi x}{h} \right) }
  }
  =
  \sinc_h \left( x - x_{h, k}\right)
  ,
\end{equation}
where we have used the Fourier transformation formula for the unit box and the scaling property of Fourier transform, which give us $\FourierTransform{ \frac{1}{h} \sinc\left( \frac{\pi x}{h} \right) } = \chi_{\xi \in \left[ -\frac{\pi}{h}, \frac{\pi}{h} \right]}$, and the shift property of Fourier transform.

Furthermore, assume that $g_h$ is a bandwidth limited function in the Fourier space meaning that its Fourier image~$\widehat{g_h} = \FourierTransform{g_h}$ contains only wavenumbers in the first Brillouin zone,
\begin{equation}
  \label{eq:77}
  \FourierTransform{g_h} = \FourierTransform{g_h} \chi_{\xi \in \left[ -\frac{\pi}{h}, \frac{\pi}{h} \right]}.
\end{equation}
The following equalities hold for the bandwidth limited function,
\begin{equation}
  \label{eq:78}
  \tensor{\left( \InverseFourierTransformSemidiscrete{\widehat{g_h}} \right)}{_j}
  =
  \left.
    \left( \InverseFourierTransform{\widehat{g_h}} \right)
  \right|_{x = x_{h, j}}
\end{equation}
and
\begin{equation}
  \label{eq:79}
  \convolution{\frac{1}{h} \sinc\left( \frac{\pi x}{h} \right)}{g_h} = g_h.
\end{equation}
The first equality follows from the definition of the continuous/discrete Fourier transform, while the second equality is a consequence of the convolution-to-multiplication property~\eqref{eq:25},
\begin{equation}
  \label{eq:80}
  \convolution{\frac{1}{h} \sinc\left( \frac{\pi x}{h} \right)}{g_h}
  =
  \InverseFourierTransform{\FourierTransform{\frac{1}{h} \sinc\left( \frac{\pi x}{h} \right)}{ \FourierTransform{g_h}}}
  =
  \InverseFourierTransform{
    \chi_{\xi \in \left[ -\frac{\pi}{h}, \frac{\pi}{h} \right]}
    \FourierTransform{g_h}
  }
  =
  \InverseFourierTransform{
    \FourierTransform{g_h}
  }
  =
  g_h.
\end{equation}
Note that the second equality tells us that the function $\frac{1}{h} \sinc\left( \frac{\pi x}{h} \right)$ is the identity for the convolution operator on the bandwidth limited functions. Calculation~\eqref{eq:80} also reveals that convolution of an arbitrary function $f$ with the function $\frac{1}{h} \sinc\left( \frac{\pi x}{h} \right)$ gives us the bandwidth limited version $f_h$ of $f$---in this sense the convolution with $\frac{1}{h} \sinc\left( \frac{\pi x}{h} \right)$ acts as a low pass filter.

We can summarise our findings regarding the bandwidth limited interpolant as a lemma.

\begin{lemma}[Bandwidth limited interpolant, infinite lattice]
  \label{lm:5}
  Let $\vec{u}_h =  \{u_{h, j} \}_{j=-\infty}^{+\infty}$ be a vector of grid values at grid points and let $\left\{x_{h, j}\right\}_{j=-\infty}^{+\infty}$, $x_{h, j} =_{\bydefinition} jh$, $h>0$, be the corresponding grid on the real line $\R$. Let $u_h$ denote the \emph{bandwidth limited interpolant} of grid values $\vec{u}_h$, that is
\begin{equation}
  \label{eq:409}
  u_h =_{\bydefinition} \InverseFourierTransform{\FourierTransformSemidiscrete{\vec{u}_h}}.
\end{equation}
Then the bandwidth limited interpolant $u_h$ is given by the formula
\begin{subequations}
  \label{eq:412}
  \begin{align}
    \label{eq:413}
    u_h(x)
    &=
      \sum_{j=-\infty}^{+ \infty}\tensor{\left(\vec{u}_h\right)}{_j} \sinc_h \left( x - x_{h, j}\right)
      ,
    \\
    \label{eq:419}
    \sinc_h(x) &=_{\bydefinition} \frac{\sin \left( \frac{\pi x}{h} \right)}{\frac{\pi x}{h}},
  \end{align}
\end{subequations}
and where the $\sinc$ function is defined as $\sinc x =_{\bydefinition} \frac{\sin x}{x}$ with the convention $\left. \sinc x \right|_{x=0} = 1$.
\end{lemma}

We can also show that the bandwidth limited function coincides with its bandwidth limited interpolant.
\begin{lemma}[Bandwidth limited function coincides with its bandwidth limited interpolant, infinite lattice]
  \label{lm:11}
  Let $U_h$ be a \emph{bandwidth limited function} in the sense of~\eqref{eq:77}, that is
  \begin{equation}
    \label{eq:448}
    \FourierTransform{U_h} = \FourierTransform{U_h} \chi_{\xi \in \left[ -\frac{\pi}{h}, \frac{\pi}{h} \right]}.
  \end{equation}
  Let $\vec{U}_h =  \{U_{h, j} \}_{j=-\infty}^{+\infty}$ be a vector of grid values $U_h$ at grid points $\vec{x}_{h} = \left\{x_{h, j}\right\}_{j=-\infty}^{+\infty}$, $x_{h, j} =_{\bydefinition} jh$, $h>0$, on the real line~$\R$. Let~$u_h$ denote the \emph{bandwidth limited interpolant} of grid values $\vec{U}_h$. Then
  \begin{equation}
    \label{eq:449}
    u_h = U_h
  \end{equation}
  everywhere on the real line $\R$.
\end{lemma}
\begin{proof}
The lemma follows from a simple manipulation. Since $U_h$ is a bandwidth limited function, we have
\begin{equation}
  \label{eq:451}
  U_h = \InverseFourierTransform{\FourierTransform{U_h} \chi_{\xi \in \left[ -\frac{\pi}{h}, \frac{\pi}{h} \right]}},
\end{equation}
which implies that the vector of grid values is given as
\begin{equation}
  \label{eq:452}
  \vec{U}_h
  =
  \left.
    \InverseFourierTransform{\FourierTransform{U_h} \chi_{\xi \in \left[ -\frac{\pi}{h}, \frac{\pi}{h} \right]}}
  \right|_{x = \vec{x}_{h}}.
\end{equation}
We however know that
\begin{equation}
  \label{eq:453}
  \left.
    \InverseFourierTransform{\FourierTransform{U_h} \chi_{\xi \in \left[ -\frac{\pi}{h}, \frac{\pi}{h} \right]}}
  \right|_{x = \vec{x}_{h}}
  =
  \InverseFourierTransformSemidiscrete{\FourierTransform{U_h}}, 
\end{equation}
which is a direct consequence of the definition of Fourier transform and the definition of semidiscrete Fourier transform, see also~\eqref{eq:78}. This equality implies that the grid values vector $\vec{U}_h$ can be obtained from $U_h$ as
\begin{equation}
  \label{eq:455}
  \vec{U}_h = \InverseFourierTransformSemidiscrete{\FourierTransform{U_h}}.
\end{equation}

On the other hand, the bandwidth limited interpolant $u_h$ is by definition constructed as
\begin{equation}
  \label{eq:450}
  u_h =_{\bydefinition} \InverseFourierTransform{\FourierTransformSemidiscrete{\vec{U}_h}}.
\end{equation}
Substituting \eqref{eq:455} into~\eqref{eq:450} to get
\begin{equation}
  \label{eq:454}
  u_h
  =
  \InverseFourierTransform{\FourierTransformSemidiscrete{\vec{U}_h}}
  =
  \InverseFourierTransform{\FourierTransformSemidiscrete{\InverseFourierTransformSemidiscrete{\FourierTransform{U_h}}}}
  =
  U_h,
\end{equation}
which was to prove.
\end{proof}

\subsubsection{Semidiscrete Fourier transform and discrete convolution}
\label{sec:semid-four-transf}
The semidiscrete Fourier transform inherits many properties of the continuous Fourier transform. In particular, if we define the discrete convolution of grid values vectors $\vec{f}_h$ and $\vec{g}_h$ as a vector whose $m$-th component is given by a discrete analogue of convolution~\eqref{eq:21}. In particular, we define
\begin{definition}[Discrete convolution, infinite lattice]
  \label{dfn:4}
  Let $\vec{f}_h =  \{f_{h, j} \}_{j=-\infty}^{+\infty}$ and $\vec{g}_h =  \{ g_{h, j} \}_{j=-\infty}^{+\infty}$ be vectors of grid values such as  the following sum is well defined,
  \begin{equation}
    \label{eq:81}
    \tensor{\left( \discreteconvolution{\vec{f}_h}{\vec{g}_h} \right)}{_m}
    =_{\bydefinition}
    h
    \sum_{j = -\infty}^{+ \infty}
    f_{h, m - j} g_{h, j}
    .
  \end{equation}
  (Note the factor $h$ in the definition~\eqref{eq:81}.) The vector $\discreteconvolution{\vec{f}_h}{\vec{g}_h}$ with components given by the formula~\eqref{eq:81} is called the discrete convolution of vectors $\vec{f}_h$ and $\vec{g}_h$. 
\end{definition}
If we further define the elementwise multiplication of vectors as
\begin{equation}
  \label{eq:82}
  \tensor{\left( \tensorschur{\vec{f}_h}{\vec{g}_h} \right)}{_m}
  =_{\bydefinition}
  f_{h, m} g_{h, m},
\end{equation}
then we recover the convolution-to-multiplication formulae~\eqref{eq:25} and~\eqref{eq:26} in the sense of the following lemma, which is straightforward to prove by direct substitution.
\begin{lemma}[Convolution-to-multiplication property, infinite lattice]
  \label{lm:6}
   Let $\vec{f}_h =  \{f_{h, j} \}_{j=-\infty}^{+\infty}$ and $\vec{g}_h =  \{ g_{h, j} \}_{j=-\infty}^{+\infty}$ be vectors of grid values such that the expressions on both sides of the following equations are well defined, then
  \begin{subequations}
    \begin{align}
      \label{eq:83}
      \FourierTransformSemidiscrete{\discreteconvolution{\vec{f}_h}{\vec{g}_h}}
      &=
        \FourierTransformSemidiscrete{\vec{f}_h}\FourierTransformSemidiscrete{\vec{g}_h},
      \\
      \label{eq:84}
      \FourierTransformSemidiscrete{\tensorschur{\vec{f}_h}{\vec{g}_h}}
    &=
      \discreteconvolution{\FourierTransformSemidiscrete{\vec{f}_h}}{\FourierTransformSemidiscrete{\vec{g}_h}}
      .
  \end{align}  
\end{subequations}
\end{lemma}

Furthermore, the properties of continuous/semidiscrete Fourier transform can be used to evaluate the bandwidth limited interpolant $u_h(y)$ at any point $y \in \R$ in the physical space \emph{without the need to sum the series representation}~\eqref{eq:bandwidth-limited-interpolant-semidiscrete-physical-space}.  In order to evaluate $u_h(y)$ we can proceed indirectly in the Fourier space
\begin{equation}
  \label{eq:85}
  u_h(y)
  =
  \left. u_h(x + y) \right|_{x=0}
  =
  \left. \left(  \InverseFourierTransform{\FourierTransform{u_h(x + y)}} \right) \right|_{x=0}
  =
  \left. \left(  \InverseFourierTransform{\exponential{\iunit \xi y} \FourierTransform{u_h} } \right) \right|_{x=0}
  =
  \left. \left(  \InverseFourierTransform{\exponential{\iunit \xi y} \FourierTransformSemidiscrete{\vec{u}_h} } \right) \right|_{x=0}
  ,
\end{equation}
where we have exploited the shift property of continuous Fourier transform and the fact that the continuous Fourier transform of the bandwidth limited interpolant $u_h$ coincides with the semidiscrete Fourier transform of the grid values vector $\vec{u}_h$.

\subsubsection{Semidiscrete Fourier transform diagonalises infinite circulant matrices/Laurent operators}
\label{sec:semid-four-transf-2}

As we have seen the continuous Fourier transform ``diagonalises'' the operator ``convolution with a given function $g$''. This property is inherited in the semidiscrete setting. The action of the Laurent operator $\tensorq{L}_h^{\text{infinite}}$ (infinite circulant matrix) on vector $\vec{f}_h$, see Section~\ref{sec:lattice-model}, is in fact the discrete convolution of (infinite) vector $\vec{c}_h$ which represents the middle row/column of $\tensorq{L}_h^{\text{infinite}}$,
\begin{equation}
  \label{eq:86}
  \vec{c}_h =_{\bydefinition}
  \begin{bmatrix}
    \vdots \\ c_{h, 2} \\ c_{h, 1} \\ c_{h, 0} \\ c_{h, 1} \\ c_{h, 2} \\  \vdots
  \end{bmatrix}
\end{equation}
with the (infinite) vector $\vec{f}_h$, that is
\begin{equation}
  \label{eq:87}
  \tensorq{L}_h^{\text{infinite}}\vec{f}_h = \discreteconvolution{\vec{c}_h}{\vec{f}_h}.
\end{equation}
Consequently, if we want to solve the eigenvalue problem
\begin{equation}
  \label{eq:88}
  \tensorq{L}_h^{\text{infinite}}\vec{f}_h = \lambda \vec{f}_h, 
\end{equation}
we can exploit the discrete convolution-to-multiplication rule almost in the same way as on the continuous level, see Section~\eqref{eq:fourier-transform}. We rewrite the left-hand side of~\eqref{eq:88} using the discrete convolution, we apply the semidiscrete Fourier transform, and we use the discrete convolution-to-multiplication rule~\eqref{eq:83}, which yields
\begin{equation}
  \label{eq:89}
  \FourierTransformSemidiscrete{\vec{c}_h}\FourierTransformSemidiscrete{\vec{f}_h} = \lambda \FourierTransformSemidiscrete{\vec{f}_h}.
\end{equation}
Following the procedure outlined is Section~\ref{sec:fourier-transform}, we rewrite the last equation as
\begin{equation}
  \label{eq:90}
  \left(\FourierTransformSemidiscrete{\vec{c}_h} - \lambda \right) \FourierTransformSemidiscrete{\vec{f}_h} = 0,
\end{equation}
which has to hold for all $\xi \in \left[ -\frac{\pi}{h}, \frac{\pi}{h} \right]$. (Note that here we deviate from the continuous setting.) Consequently, we identify the eigenfunction--eigenvalue pair as
\begin{subequations}
  \label{eq:91}
  \begin{align}
    \label{eq:92}
    \FourierTransformSemidiscrete{\vec{f}_h} &= \diracdelta_\eta (\xi) \chi_{\xi \in \left[ -\frac{\pi}{h}, \frac{\pi}{h} \right]}, \\
    \label{eq:93}
    \lambda &= \left. \FourierTransformSemidiscrete{\vec{c}_h} \right|_{\xi = \eta},
  \end{align}
\end{subequations}
where $\diracdelta_\eta (\xi)$ denotes the shifted Dirac function in the Fourier space, $\diracdelta_\eta (\xi)= \diracdelta(\xi - \eta)$. The spectrum thus contains infinitely many eigenvalues $\eta$ which are generated by the semidiscrete Fourier transform of the middle row/column vector~$\vec{c}_h$ of matrix~$\tensorq{L}_h^{\text{infinite}}$, that is by the vector $\FourierTransformSemidiscrete{\vec{c}_h}$. (Note that $\eta \in \left[ -\frac{\pi}{h}, \frac{\pi}{h} \right]$. This characterisation of spectra of Laurent operators is a standard one---see \cite[Corollary 1.3]{gohberg.i.goldberg.s.ea:basic}.) If we want to explicitly identify the eigenfunctions/eigenvectors~$\vec{f}_h$, then we must take the inverse semidiscrete Fourier transform in~\eqref{eq:92},
\begin{equation}
  \label{eq:94}
  \vec{f}_h = \InverseFourierTransformSemidiscrete{ \diracdelta_\eta (\xi) \chi_{\xi \in \left[ -\frac{\pi}{h}, \frac{\pi}{h} \right]}},
\end{equation}
which can be done using the identity~\eqref{eq:78} as follows
\begin{multline}
  \label{eq:95}
  \tensor{\left( \vec{f}_h \right)}{_j}
  =
  \tensor{\left( \InverseFourierTransformSemidiscrete{ \diracdelta_\eta (\xi) \chi_{\xi \in \left[ -\frac{\pi}{h}, \frac{\pi}{h} \right]}} \right)}{_j}
  =
  \left.
    \left(
      \InverseFourierTransform{ \diracdelta_\eta (\xi) \chi_{\xi \in \left[ -\frac{\pi}{h}, \frac{\pi}{h} \right]}}
    \right)
  \right|_{x = x_{h, j}}
  =
  \left.
    \left(
      \InverseFourierTransform{
        \frac{1}{2\pi} \FourierTransform{\exponential{\iunit \eta x}}
        \FourierTransform{ \frac{1}{h} \sinc\left( \frac{\pi x}{h} \right)}
      }
    \right)
  \right|_{x = x_{h, j}}
  \\
  =
  \left.
    \left(
      \convolution{
        \frac{1}{2\pi} \exponential{\iunit \eta x}
      }
      {
        \frac{1}{h} \sinc\left( \frac{\pi x}{h} \right)
      }
    \right)
  \right|_{x = x_{h, j}}
  =
  \left.
    \left(
      \frac{1}{2 \pi}
      \int_{y=-\infty}^{+\infty}
      \exponential{\iunit \eta \left(x -y\right)}
      \frac{1}{h} \sinc\left( \frac{\pi y}{h} \right)
      \,
      \diff y
    \right)
  \right|_{x = x_{h, j}}
  \\
  =
  \left.
    \frac{\exponential{\iunit \eta x}}{2 \pi}
  \right|_{x = x_{h, j}}
  \FourierTransform{\frac{1}{h} \sinc\left( \frac{\pi y}{h} \right)}(\eta)
  =
  \left.
    \frac{\exponential{\iunit \eta x}}{2 \pi}
  \right|_{x = x_{h, j}}
  \chi_{\eta \in \left[ -\frac{\pi}{h}, \frac{\pi}{h} \right]}.
\end{multline}
The same calculation shows that the corresponding bandwidth limited interpolant $f_h$ of $\vec{f}_h$ is given by the formula
\begin{equation}
  \label{eq:96}
  f_h
  =
  \convolution{
    \frac{1}{2\pi} \exponential{\iunit \eta x}
  }
  {
    \frac{1}{h} \sinc\left( \frac{\pi x}{h} \right)
  },
\end{equation}
where $\eta \in \left[ -\frac{\pi}{h}, \frac{\pi}{h} \right]$. Equality~\eqref{eq:95} shows that the eigenvectors for the Laurent operator (infinite circulant matrix) are obtained by \emph{sampling of the eigenfunctions of the corresponding continuous operator} ``convolution with a given function'', see Section~\ref{sec:fourier-transform}. (Provided that we work with eigenfunctions in the Brillouin zone.) Furthermore, if we restrict ourselves to eigenfunctions in the Brillouin zone, then the eigenvalues for the discrete problem with the generating vector $\vec{c}_h$,
\begin{equation}
  \label{eq:97}
  \discreteconvolution{\vec{c}_h}{\vec{f}_h} = \lambda \vec{f}_h,
\end{equation}
coincide with the eigenvalues for the corresponding continuous problem
\begin{equation}
  \label{eq:98}
  \convolution{g}{f} = \lambda f,
\end{equation}
where we set $g = \InverseFourierTransform{\FourierTransformSemidiscrete{\vec{c}_h}}$, see~\eqref{eq:38} and~\eqref{eq:93}.  

\subsubsection{Semidiscrete Fourier transform and differentiation}
\label{sec:semid-four-transf-1}

Concerning the differentiation of a function reconstructed from the grid values, there are again many possibilities how to introduce a suitable concept of differentiation. (One can for example think of finite differences type approximations and so forth.) Since we have decided to work with the bandwidth limited interpolant, we can introduce the discrete differentiation naturally via the differentiation of the bandwidth limited interpolant. We simply \emph{define} the symbol $\dd{^n}{x^n}$ acting on the vector of grid values $\vec{f}_h$ via an analogue of the differentiation formula~\eqref{eq:23}
\begin{equation}
  \label{eq:99}
  \FourierTransformSemidiscrete{\dd{^n}{x^n} \vec{f}_h}
  =_{\bydefinition}
  \left( \iunit \xi \right)^n \FourierTransformSemidiscrete{\vec{f}_h},
\end{equation}
or, more explicitly, as
\begin{equation}
  \label{eq:100}
  \dd{^n}{x^n} \vec{f}_h = \InverseFourierTransformSemidiscrete{\left( \iunit \xi \right)^n \FourierTransformSemidiscrete{\vec{f}_h}}. 
\end{equation}
For further reference we thus have the following definition.
\begin{definition}[Fourier transform based differentiation, infinite lattice]
  \label{dfn:7}
  Let $\vec{f}_h =  \{f_{h, j} \}_{j=-\infty}^{+\infty}$ be a vector of grid values. We set
  \begin{equation}
    \label{eq:422}
    \dd{^n}{x^n} \vec{f}_h =_{\bydefinition} \InverseFourierTransformSemidiscrete{\left( \iunit \xi \right)^n \FourierTransformSemidiscrete{\vec{f}_h}}.
  \end{equation}
\end{definition}
If we define the symbol $\dd{^n}{x^n}$ in this way, then it gives us an approximation of the derivative at the grid points, and, moreover, the discrete differentiation $\dd{^n}{x^n}$ inherits the properties of continuous differentiation with respect to the Fourier transform. Namely, if we are interested in the eigenvalue problem
\begin{equation}
  \label{eq:101}
  \dd{^n}{x^n} \vec{f}_h = \lambda \vec{f}_h,
\end{equation}
then we can solve it using the same procedure as in Section~\ref{sec:fourier-transform}, where we analyse the continuous eigenvalue problem~\eqref{eq:30}. We first apply the semidiscrete Fourier transform to~\eqref{eq:101} and we get the equation
\begin{equation}
  \label{eq:102}
  \left( \iunit \xi \right)^n \FourierTransformSemidiscrete{\vec{f}_h} = \lambda \FourierTransformSemidiscrete{\vec{f}_h},
\end{equation}
which should hold for all $\xi \in \left[ -\frac{\pi}{h}, \frac{\pi}{h} \right]$. Formal solution of this equation is
\begin{equation}
  \label{eq:103}
  \FourierTransformSemidiscrete{\vec{f}_h} = \diracdelta_\eta (\xi) \chi_{\xi \in \left[ -\frac{\pi}{h}, \frac{\pi}{h} \right]}, 
\end{equation}
with $\eta$ fixed at arbitrary wavenumber $\eta \in \left[ -\frac{\pi}{h}, \frac{\pi}{h} \right]$, and the corresponding eigenvalue is $\lambda = \left(\iunit \eta \right)^n$. The corresponding (infinite) eigenvector $\vec{f}_h$ is then obtained by the inverse semidiscrete Fourier transform, see~\eqref{eq:94} for explicit calculation. We see that even at the semidiscrete level the differentiation operator again shares the same eigenfunctions with the ``convolution with a given function'' operator, the operators again differ by eigenvalues only, see Section~\ref{sec:fourier-transform} for the same observation on the continuous level. Finally, we note that~\eqref{eq:103} and~\eqref{eq:94} imply that the \emph{eigenvector $\vec{f}_h$ elements are obtained by sampling of the continuous eigenfunction $f$ for the continuous differentiation operator}, see~\eqref{eq:34}. Furthermore, we see that such $\vec{f}_h$ is an eigenvector for an arbitrary order semidiscrete differential operator $\dd{^n}{x^n}$.

\subsection{Search for continuous analogue of lattice model---revisiting the nearest neighbour interaction lattice model}
\label{sec:search-cont-anal-1}

Having thoroughly discussed the reconstruction procedure, we can go back to the nearest neighbour interaction lattice model. This time we want to identify exact counterpart of the discrete model \emph{provided that the reconstruction is done using the bandwidth limited interpolant}.

We start with the $j$-th discrete equation~\eqref{eq:6}, and we multiply the $j$-th equation by $\exponential{- \iunit \xi x_{h, j}}$. This gives
\begin{equation}
  \label{eq:104}
  \ddd{}{t}\left( u_{h, j}\exponential{- \iunit \xi x_{h, j}} \right) - \ccontinuous^2 \frac{u_{h, j+1}\exponential{- \iunit \xi x_{h, j}} - 2 u_{h, j}\exponential{- \iunit \xi x_{h, j}} + u_{h, j-1}\exponential{- \iunit \xi x_{h, j}}}{h^2} = 0.
\end{equation}
Then we exploit the fact that
\begin{equation}
  \label{eq:105}
  \exponential{- \iunit \xi x_{h, j}}
  =
  \exponential{- \iunit \xi x_{h, j \pm 1} \pm \iunit \xi h}
  =
  \exponential{- \iunit \xi x_{h, j \pm 1}} \exponential{\pm \iunit \xi h},
\end{equation}
and we get
\begin{equation}
  \label{eq:106}
  \ddd{}{t}\left( u_{h, j}\exponential{- \iunit \xi x_{h, j}} \right) - \ccontinuous^2 \frac{ \left(u_{h, j+1}\exponential{- \iunit \xi x_{h, j+1}} \right) \exponential{\iunit \xi h} - 2 u_{h, j}\exponential{- \iunit \xi x_{h, j}} + \left(u_{h, j-1}\exponential{- \iunit \xi x_{h, j-1}}\right) \exponential{-\iunit \xi h}}{h^2} = 0.
\end{equation}
Now we sum all the equations with respect to $j$, which leads to
\begin{equation}
  \label{eq:107}
  \ddd{}{t}\FourierTransformSemidiscrete{\vec{u}_h}
  -
  \frac{\ccontinuous^2}{h^2}
  \left(\exponential{\iunit \xi h} - 2 +  \exponential{-\iunit \xi h} \right)
  \FourierTransformSemidiscrete{\vec{u}_h}
  =
  0
  .
\end{equation}
We are however using the bandwidth limited interpolant, hence we can rewrite the previous equation as
\begin{equation}
  \label{eq:108}
  \ppd{}{t}\FourierTransform{u_h}
  -
  \frac{\ccontinuous^2}{h^2}
  \left(\exponential{\iunit \xi h} - 2 +  \exponential{-\iunit \xi h} \right)
  \FourierTransform{u_h}
  =
  0
  .
\end{equation}
(Semidiscrete Fourier transform of grid values is identical to the continuous Fourier transform of the corresponding bandwidth limited interpolant, see Figure~\eqref{fig:semidiscrete-fourier-transform}.) Transition from~\eqref{eq:107} to \eqref{eq:108} is the key step in our analysis---it embodies the jump from the discrete to the continuous setting. Equation~\eqref{eq:108} can be manipulated using the standard machinery of \emph{continuous} Fourier transform. First we rewrite equation~\eqref{eq:108} as
\begin{equation}
  \label{eq:109}
  \ppd{}{t}
  \FourierTransform{u_h}
  +
  2
  \ccontinuous^2
  \frac{1 - \cos \left(\xi h\right)}{h^2}{\FourierTransform{u_h}}
  =
  0
  .
\end{equation}
Since we want to manipulate the equation to a form similar to the Fourier transformed standard wave equation, we want to see the Fourier image of the second derivative at the right position, hence we rewrite~\eqref{eq:109} as
\begin{equation}
  \label{eq:110}
  \ppd{}{t}
  \FourierTransform{u_h}
  -
  2
  \ccontinuous^2
  \frac{1 - \cos \left(\xi h\right)}{\xi^2 h^2}\left( -\xi^2 \FourierTransform{u_h} \right)
  =
  0
  ,
\end{equation}
which can be further manipulated into the form
\begin{equation}
  \label{eq:111}
  \ppd{}{t}
  \FourierTransform{u_h}
  -
  \ccontinuous^2
  \left(\frac{\sin \left( \frac{\xi h}{2} \right)}{\frac{\xi h}{2}} \right)^2 \FourierTransform{\ppd{u_h}{x}}
  =
  0.
\end{equation}
The factor in the second term can be rewritten as Fourier transform of a known function. It holds
\begin{equation}
  \label{eq:112}
  \FourierTransform{\UnitTriangle \left( \frac{x}{h} \right)} = h \sinc^2 \left( \frac{h \xi}{2} \right),
\end{equation}
where $\UnitTriangle$ is the unit triangle function defined as
\begin{equation}
  \label{eq:113}
  \UnitTriangle(x)
  =_\bydefinition
  \begin{cases}
    0,& x \in (-\infty, -1), \\
    x+1, & x \in [-1, 0], \\
    -x+1, & x \in (0, 1], \\
    0, & x \in (1, +\infty),
  \end{cases}
\end{equation}
see also Figure~\ref{fig:unit-triangle}.
\begin{figure}[t]
  \centering
  \includegraphics[width=0.25\textwidth]{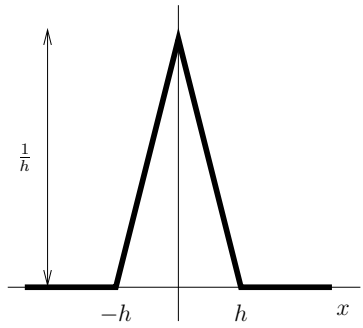}
  \caption{Function $\frac{1}{h}\UnitTriangle \left( \frac{x}{h} \right)$.}
  \label{fig:unit-triangle}
\end{figure}
Consequently, we can rewrite~\eqref{eq:111} as
\begin{equation}
  \label{eq:114}
  \ppd{}{t}
  \FourierTransform{u_h}
  -
  \ccontinuous^2
  \left( \frac{1}{h}\FourierTransform{\UnitTriangle \left( \frac{x}{h} \right)} \right)
  \FourierTransform{\ppd{u_h}{x}}
  =
  0
  .
\end{equation}
Using again the convolution theorem we get
\begin{equation}
  \label{eq:115}
  \ppd{}{t}
  \FourierTransform{u_h}
  -
  \ccontinuous^2
  \FourierTransform{
    \convolution{
      \left( \frac{1}{h}\UnitTriangle \left( \frac{x}{h} \right) \right)
    }
    {
      \ppd{u_h}{x}
    }
  }
  =
  0,
\end{equation}
and upon the transformation back to the physical domain we finally arrive at
\begin{equation}
  \label{eq:116}
  \ppd{u_h}{t}
  -
  \ccontinuous^2
  \convolution{
    \left( \frac{1}{h}\UnitTriangle \left( \frac{x}{h} \right) \right)
  }
  {
    \ppd{u_h}{x}
  }
  =
  0
  .
\end{equation}
We can thus claim that if grid values $\vec{u}_h(t) = \left\{ u_{h, j} (t)\right\}_{j=-\infty}^{+\infty}$ solve the system of ordinary differential equations~\eqref{eq:6}, then the corresponding bandwidth limited interpolant $u_h(x, t)$ solves~\eqref{eq:116}. Note that the sequence $\frac{1}{h}\UnitTriangle \left( \frac{x}{h} \right)$ is a $\diracdelta$-sequence, that is
\begin{equation}
  \label{eq:117}
  \frac{1}{h}\UnitTriangle \left( \frac{x}{h} \right) \stackrel{h \to 0+}{\longrightarrow} \diracdelta(x)
\end{equation}
in the sense of distributions, hence in the limit $h \to 0+$ we indeed recover the standard wave equation.

Furthermore, if $u$ is \emph{any} function---not necessarily a bandwidth limited interpolant of some grid values---that solves
\begin{equation}
  \label{eq:118}
  \ppd{u}{t}
  -
  \ccontinuous^2
  \convolution{
    \left( \frac{1}{h}\UnitTriangle \left( \frac{x}{h} \right) \right)
  }
  {
    \ppd{u}{x}
  }
  =
  0
  ,
\end{equation}
then the integration by parts in the convolution term in~\eqref{eq:118} yields
\begin{equation}
  \label{eq:119}
  \ppd{u}{t}
  -
  \ccontinuous^2
  \convolution{
    \left(
      \ppd{}{x}
      \left( \frac{1}{h}\UnitTriangle \left( \frac{x}{h} \right) \right)
    \right)
  }
  {
    u
  }
  =
  0.
\end{equation}
The (distributional) derivative of the unit triangle function reads $\ppd{}{x} \left( \frac{1}{h}\UnitTriangle \left( \frac{x}{h} \right) \right) = \frac{1}{h^2} \left( \diracdelta(x+h) - 2\diracdelta(x) +  \diracdelta(x+h) \right)$, hence~\eqref{eq:119} reduces to
\begin{equation}
  \label{eq:120}
  \ppd{u}{t}(x,t)
  -
  \frac{\ccontinuous^2}{h^2}
  \left(
    u(x+h)
    -
    2u(x)
    +
    u(x-h)
  \right)
  =
  0.
\end{equation}
Consequently, we see that \emph{any} solution $u$ to~\eqref{eq:118} sampled to the grid $\left\{x_{h, j} \right\}_{j = -\infty}^{+ \infty}$ solves the discrete system~\eqref{eq:6}. (So far we do not know whether $u$ coincides with the bandwidth limited interpolant of these values, but we might expect this. We show this in the proof of Theorem~\ref{thr:2} which deals with a more general problem.) We thus have the following \emph{equivalence between the discrete system of ordinary differential equations and the corresponding continuous partial differential equation}.

\begin{theorem}[Equivalence between a discrete system of ordinary differential equations for grid values and the corresponding partial differential equation---nearest neighbour interaction, infinite lattice]
  \label{thr:1}
  Let $h>0$, and let $\left\{x_{h, j}\right\}_{j=-\infty}^{+\infty}$, $x_{h, j} =_{\bydefinition} jh$, be the corresponding grid on the real line $\R$. Let~$\vec{u}_h(t) = \left\{ u_{h, j} (t)\right\}_{j=-\infty}^{+\infty}$ be a collection of grid values on the grid $\left\{x_{h, j} \right\}_{j=-\infty}^{+\infty}$, and let~$u_h(x,t)$ be the corresponding bandwidth limited interpolant of~$\vec{u}_h(t)$, that is
  \begin{subequations}
    \label{eq:121}
    \begin{equation}
      \label{eq:122}
      u_h(x, t)
      =
      \sum_{j=-\infty}^{+ \infty}  u_{h, j} (t) \sinc_h \left( x - x_{h, j}\right)
      ,
    \end{equation}
    where
    \begin{equation}
      \label{eq:123}
      \sinc_h(x) =_{\bydefinition} \frac{\sin \left( \frac{\pi x}{h} \right)}{\frac{\pi x}{h}}.
    \end{equation}
  \end{subequations}
  Let the grid values $\vec{u}_h(t)  = \left\{ u_{h, j} (t)\right\}_{j=-\infty}^{+\infty}$ solve the initial value problem for the system of ordinary differential equations
  \begin{subequations}
    \label{eq:124}
    \begin{align}
      \label{eq:125}
      \ddd{u_{h, j}}{t} - \ccontinuous^2 \frac{u_{h, j+1} - 2 u_{h, j} + u_{h, j-1}}{h^2} &= 0, \\
      \label{eq:126}
      \left. u_{h, j} \right|_{t=0} &= u_{h, j}^0,\\
      \label{eq:127}
      \left. \dd{u_{h, j}}{t} \right|_{t=0} &= v_{h, j}^0.
    \end{align}
  \end{subequations}
  Let $u$ be a function that solves, for $x \in \R$ on the real line, the initial value problem 
  \begin{subequations}
    \label{eq:128}
    \begin{align}
      \label{eq:129}
      \ppd{u(x, t)}{t}
      -
      \ccontinuous^2
      \convolution{
      \left( \frac{1}{h}\UnitTriangle \left( \frac{x}{h} \right) \right)
      }
      {
      \ppd{u(x, t)}{x}
      }
      &=
        0
        ,
      \\
      \label{eq:130}
      \left. u \right|_{t=0} &= u_h^0,\\
      \label{eq:131}
      \left. \dd{u}{t} \right|_{t=0} &= v_h^0,
    \end{align}
  \end{subequations}
  where $\UnitTriangle$ denotes the unit triangle function~\eqref{eq:113}, that is
  \begin{equation}
    \label{eq:132}
    \UnitTriangle(x)
    =_\bydefinition
    \begin{cases}
      0,& x \in (-\infty, -1), \\
      x+1, & x \in [-1, 0], \\
      -x+1, & x \in (0, 1], \\
      0, & x \in (1, +\infty),
    \end{cases}
  \end{equation}
  and where the initial data $u_h^0$ and $v_h^0$ in~\eqref{eq:130} and \eqref{eq:131} are the bandwidth limited interpolants of the initial data~\eqref{eq:126} and \eqref{eq:127}. Then the function $u$ is the solution to the continuous problem~\eqref{eq:128} if and only if $u = u_h$, where $u_h$ is the bandwidth limited interpolant $u_h$ of grid values $\vec{u}_h(t)  = \left\{ u_{h, j} (t)\right\}_{j=-\infty}^{+\infty}$ that solve the discrete problem~\eqref{eq:124}.
\end{theorem}

\subsection{Correspondence between discrete and continuous models}
\label{sec:corr-betw-discr}
The findings regarding the nearest neighbour interaction are straightforward to generalise to the multiple-neighbours setting~\eqref{eq:41}, wherein the interaction between the particles is not limited to nearest neighbour interaction---all particles can interact with all neighbours. Namely, we prove the following theorem.

\begin{theorem}[Equivalence between discrete system of ordinary differential equations for grid values and the corresponding partial differential equation---general interaction, infinite lattice]
  \label{thr:2}
  Let $h>0$, and let $\left\{x_{h, j}\right\}_{j=-\infty}^{+\infty}$, $x_{h, j} =_{\bydefinition} jh$, be the corresponding grid on the real line $\R$. Let~$\vec{u}_h(t) = \left\{ u_{h, j} (t)\right\}_{j=-\infty}^{+\infty}$ be a collection of grid values on the grid $\left\{x_{h, j} \right\}_{j=-\infty}^{+\infty}$, and let~$u_h(x,t)$ be the corresponding bandwidth limited interpolant of~$\vec{u}_h(t)$, that is
  \begin{subequations}
    \label{eq:133}
    \begin{equation}
      \label{eq:134}
      u_h(x, t)
      =
      \sum_{j=-\infty}^{+ \infty}  u_{h, j} (t) \sinc_h \left( x - x_{h, j}\right)
      ,
    \end{equation}
    where
    \begin{equation}
      \label{eq:135}
      \sinc_h(x) =_{\bydefinition} \frac{\sin \left( \frac{\pi x}{h} \right)}{\frac{\pi x}{h}}.
    \end{equation}
  \end{subequations}
  Let the grid values $\vec{u}_h(t)  = \left\{ u_{h, j} (t)\right\}_{j=-\infty}^{+\infty}$ solve the initial value problem for the system of ordinary differential equations
  \begin{subequations}
    \label{eq:136}
    \begin{align}
      \ddd{u_{h,j}}{t}
      -
      \sum_{m=-\infty}^{+ \infty} c_{h, j-m}u_{h, m}
      &=
        0,
      \\
      \label{eq:137}
      \left. u_{h, j} \right|_{t=0} &= u_{h, j}^0,\\
      \label{eq:138}
      \left. \dd{u_{h, j}}{t} \right|_{t=0} &= v_{h, j}^0,
    \end{align}
  \end{subequations}
  with the property $c_{h, i} = c_{h, -i}$, that is the system
  \begin{equation}
    \label{eq:139}
    \ddd{}{t}
    \begin{bmatrix}
      \vdots \\
      u_{h, -2} \\
      u_{h, -1} \\
      u_{h, 0} \\
      u_{h, 1} \\
      u_{h, 2} \\
      \vdots
    \end{bmatrix}
    -
    \begin{bmatrix}
      \ddots & \vdots & \vdots & \vdots & \vdots & \vdots & \reflectbox{$\ddots$} \\
      \cdots & c_{h, 0} & c_{h, 1} & c_{h, 2} & c_{h, 3} & c_{h, 4} & \cdots \\
      \cdots & c_{h, 1} & c_{h, 0} & c_{h, 1} & c_{h, 2} & c_{h, 3} & \cdots \\
      \cdots & c_{h, 2} & c_{h, 1} & c_{h, 0} & c_{h, 1} & c_{h, 2} & \cdots \\
      \cdots & c_{h, 3} & c_{h, 2} & c_{h, 1} & c_{h, 0} & c_{h, 1} & \cdots \\
      \cdots & c_{h, 4} & c_{h, 3} & c_{h, 2} & c_{h, 1} & c_{h, 0} & \cdots \\
      \reflectbox{$\ddots$} & \vdots & \vdots & \vdots & \vdots & \vdots & \ddots
    \end{bmatrix}
    \begin{bmatrix}
      \vdots \\
      u_{h, -2} \\
      u_{h, -1} \\
      u_{h, 0} \\
      u_{h, 1} \\
      u_{h, 2} \\
      \vdots
    \end{bmatrix}
    =
    \begin{bmatrix}
      \vdots \\
      0 \\
      0 \\
      0 \\
      0 \\
      0 \\
      \vdots
    \end{bmatrix}
    .
  \end{equation}
  Let $u$ be a function that solves, for $x \in \R$ on the real line, the initial value problem 
  \begin{subequations}
    \label{eq:140}
    \begin{align}
      \label{eq:141}
      \ppd{u}{t}
      +
      \convolution{
      \left(
      \frac{1}{h}
      \InverseFourierTransform{\frac{\FourierTransformSemidiscrete{\vec{c}_h}}{\xi^2}}
      \right)
      }
      {
      \ppd{u}{x}
      }
      &=
        0
        ,
      \\
      \label{eq:142}
      \left. u \right|_{t=0} &= u_h^0,\\
      \label{eq:143}
      \left. \dd{u}{t} \right|_{t=0} &= v_h^0,
    \end{align}
  \end{subequations}
  where
  \begin{equation}
    \label{eq:144}
    \vec{c}_h =_{\bydefinition}
    \begin{bmatrix}
      \vdots \\ c_{h, 2} \\ c_{h, 1} \\ c_{h, 0} \\ c_{h, 1} \\ c_{h, 2} \\  \vdots
    \end{bmatrix}
  \end{equation}
  denotes the vector of coefficients in~\eqref{eq:136} and \eqref{eq:139} respectively, and $\FourierTransformSemidiscrete{\cdot}$ and $\InverseFourierTransform{\cdot}$ denote the \emph{semidiscrete} Fourier transform~\eqref{eq:69} and the \emph{continuous} inverse Fourier transform~\eqref{eq:20},  and where the initial data $u_h^0$ and $v_h^0$ in~\eqref{eq:142} and \eqref{eq:143} are bandwidth limited interpolants of the initial data~\eqref{eq:137} and \eqref{eq:138}. Then the function $u$ is a solution to the continuous problem~\eqref{eq:140} if and only if $u = u_h$, where $u_h$ is the bandwidth limited interpolant $u_h$ of grid values $\vec{u}_h(t)  = \left\{ u_{h, j} (t)\right\}_{j=-\infty}^{+\infty}$ that solve the discrete problem~\eqref{eq:136}.
\end{theorem}

We note that~\eqref{eq:141} can be also rewritten as
\begin{equation}
  \label{eq:145}
  \ppd{u_h}{t}
  +
  \convolution{
    \left(
      \frac{1}{h}
      \InverseFourierTransform{\frac{\FourierTransform{c_h}}{\xi^2}}
    \right)
  }
  {
    \ppd{u_h}{x}
  }
  =
  0
  ,
\end{equation}
where $c_h$ is the bandwidth limited interpolant of grid values $\vec{c}_h$. This observation follows from the properties of bandwidth limited interpolant, see~\eqref{eq:71}. As we shall see from the proof, the \emph{theorem indeed holds only if the initial conditions are well prepared}. The initial conditions $u_h^0$ and $v_h^0$ for the continuous case must be bandwidth limited functions constructed as bandwidth limited interpolants of the discrete values $\left\{ u_{h, j}^0\right\}_{j=-\infty}^{+\infty}$ and $\left\{ v_{h, j}^0\right\}_{j=-\infty}^{+\infty}$.

\begin{proof}
  We proceed as in Section~\eqref{sec:search-cont-anal-1}. First we note that~\eqref{eq:136} can be rewritten using the discrete convolution~\eqref{eq:81} operator as
  \begin{equation}
    \label{eq:146}
    \ddd{}{t} \vec{u}_h - \frac{1}{h} \discreteconvolution{\vec{c}_h}{\vec{u}_h} = \vec{0}.
  \end{equation}
  We take the discrete Fourier transform of~\eqref{eq:146}, and we use the discrete convolution-to-multiplication property~\eqref{eq:83}, and we get
  \begin{equation}
    \label{eq:147}
    \ddd{}{t} \FourierTransformSemidiscrete{\vec{u}_h} - \frac{1}{h} \FourierTransformSemidiscrete{\vec{c}_h} \FourierTransformSemidiscrete{\vec{u}_h} = 0.
  \end{equation}
  Now we exploit the fact that the semidiscrete Fourier transform of grid values vector is identical to the continuous Fourier transform of the corresponding bandwidth limited interpolant, see~\eqref{eq:71}, and we rewrite~\eqref{eq:147} as
  \begin{equation}
    \label{eq:148}
    \ppd{}{t} \FourierTransform{u_h} - \frac{1}{h} \FourierTransformSemidiscrete{\vec{c}_h} \FourierTransform{u_h} = 0,
  \end{equation}
  which bring us from the discrete setting to the continuous setting. We rearrange the terms to see the Fourier image of the second derivative,
  \begin{equation}
    \label{eq:149}
    \ppd{}{t} \FourierTransform{u_h} + \frac{1}{h} \frac{\FourierTransformSemidiscrete{\vec{c}_h}}{\xi^2} \left(- \xi^2 \FourierTransform{ u_h} \right) = 0,
  \end{equation}
  which upon taking the inverse Fourier transform yields
  \begin{equation}
    \label{eq:150}
    \ppd{u_h}{t}
    +
    \convolution{
      \left(
        \frac{1}{h}
        \InverseFourierTransform{\frac{\FourierTransformSemidiscrete{\vec{c}_h}}{\xi^2}}
      \right)
    }
    {
      \ppd{u_h}{x}
    }
    =
    0
    ,
  \end{equation}
  where we have also used the differentiation property~\eqref{eq:23}.
  
  On the other hand, assume that $u$ is a function that solves the equation
  \begin{equation}
    \label{eq:151}
    \ppd{u}{t}
    +
    \convolution{
      \left(
        \frac{1}{h}
        \InverseFourierTransform{\frac{\FourierTransformSemidiscrete{\vec{c}_h}}{\xi^2}}
      \right)
    }
    {
      \ppd{u}{x}
    }
    =
    0
  \end{equation}
  with the initial conditions~\eqref{eq:142} and \eqref{eq:143}, that is
  \begin{subequations}
    \label{eq:152}
    \begin{align}
      \label{eq:153}
      \left. u \right|_{t=0} &= u_h^0,\\
      \label{eq:154}
      \left. \dd{u}{t} \right|_{t=0} &= v_h^0.
    \end{align}
  \end{subequations}
  Note that we do not assume \emph{a priori} that $u$ is the bandwidth limited interpolant of some grid values or a bandwidth limited function. We first observe that the definition of continuous/semidiscrete Fourier transform, see~\eqref{eq:fourier-transform} and~\eqref{eq:semidiscrete-fourier-transform}, imply that
  \begin{equation}
    \label{eq:155}
    \FourierTransformSemidiscrete{\vec{c}_h}
    =
    h \FourierTransform{\sum_{n = -\infty}^{+\infty} c_{h, n} \diracdelta(x - x_{h, n})}
    \chi_{\xi \in \left[ -\frac{\pi}{h}, \frac{\pi}{h} \right]}
    .
  \end{equation}
  Using the convolution-to-multiplication formula~\eqref{eq:25} and the differentiation formula~\eqref{eq:22} we see that the continuous Fourier transform of~\eqref{eq:151} reads
  \begin{equation}
    \label{eq:156}
    \ppd{}{t}
    \FourierTransform{u}
    -
    \frac{1}{h}
    \FourierTransformSemidiscrete{\vec{c}_h}
    \FourierTransform{u}
    =
    0
    .
  \end{equation}
  This equation must hold for all $\xi \in \R$. In virtue of~\eqref{eq:155} we however see that the second term on the left-hand side of~\eqref{eq:156} vanishes for all $\xi \in \R \setminus \left[ -\frac{\pi}{h}, \frac{\pi}{h} \right]$, which implies that
  \begin{equation}
    \label{eq:157}
    \ppd{}{t}
    \FourierTransform{u}
    \chi_{\xi \in \R \setminus \left[ -\frac{\pi}{h}, \frac{\pi}{h} \right]}
    =
    0.
  \end{equation}
  Since the functions $u_h^0$ and $v_h^0$ in the \emph{initial} conditions~\eqref{eq:153} and \eqref{eq:154} are bandwidth limited functions, we see that the function $u$ starting from these initial conditions must be a bandwidth limited function at $t=0$. Equation~\eqref{eq:157} then implies that $u$ must remain a bandwidth limited function during the evolution, that is for all $t>0$.

  Having shown that the solution to \eqref{eq:151} must be a bandwidth limited function, we can continue in analysis of~\eqref{eq:151}. We now denote $u$ as $u_h$ in order to indicate that we are in fact working with a bandwidth limited function. The integration by parts in the convolution term in \eqref{eq:151} then yields
  \begin{multline}
    \label{eq:158}
    \convolution{
      \left(
        \frac{1}{h}
        \InverseFourierTransform{\frac{\FourierTransformSemidiscrete{\vec{c}_h}}{\xi^2}}
      \right)
    }
    {
      \ppd{u_h}{x}
    }
    =
    \convolution{
      \left(
        \frac{1}{h}
        \ppd{}{x}
        \InverseFourierTransform{\frac{\FourierTransformSemidiscrete{\vec{c}_h}}{\xi^2}}
      \right)
    }
    {
      u_h
    }
    =
    \convolution{
      \left(
        \frac{1}{h}
        \InverseFourierTransform{- \xi^2 \frac{\FourierTransformSemidiscrete{\vec{c}_h}}{\xi^2}}
      \right)
    }
    {
      u_h
    }
    =
    -
    \convolution{
      \left(
        \frac{1}{h}
        \InverseFourierTransform{ \FourierTransformSemidiscrete{\vec{c}_h}}
      \right)
    }
    {
      u_h
    }
    \\
    =
    -
    \convolution{
      \left(
        \InverseFourierTransform{
          \FourierTransform{\sum_{n = -\infty}^{+\infty} c_{h, n} \diracdelta(x - x_{h, n})}
          \chi_{\xi \in \left[ -\frac{\pi}{h}, \frac{\pi}{h} \right]}
        }
      \right)
    }
    {
      u_h
    }
    =
    -
    \convolution{
      \left(
        \InverseFourierTransform{
          \FourierTransform{\sum_{n = -\infty}^{+\infty} c_{h, n} \diracdelta(x - x_{h, n})}
          \FourierTransform{\frac{1}{h} \sinc\left( \frac{\pi x}{h} \right)}
        }
      \right)
    }
    {
      u_h
    }
    \\
    =
    \convolution{\left( \convolution{\sum_{n = -\infty}^{+\infty} c_{h, n} \diracdelta(x - x_{h, n})}{\frac{1}{h} \sinc\left( \frac{\pi x}{h} \right)} \right)}{u_h}
    =
    \convolution{\left(\sum_{n = -\infty}^{+\infty} c_{h, n} \diracdelta(x - x_{h, n}) \right)}{ \left(\convolution{\frac{1}{h} \sinc\left( \frac{\pi x}{h} \right)}{u_h}\right)}
    \\
    =
    \convolution{\left(\sum_{n = -\infty}^{+\infty} c_{h, n} \diracdelta(x - x_{h, n}) \right)}{u_h}
    =
    -
    \sum_{n = -\infty}^{+\infty} c_{h, n} u_h(x - x_{h, n})
    ,
  \end{multline}
  where we have used the fact that~$\frac{1}{h} \sinc\left( \frac{\pi x}{h} \right)$ is the identity on the space of bandwidth limited functions, see~\eqref{eq:79}. Consequently, equation \eqref{eq:151} reduces to
  \begin{equation}
    \label{eq:159}
    \ppd{u_h}{t}(x,t)
    -
    \sum_{n = -\infty}^{+\infty} c_{h, n} u_h(x - x_{h, n}, t)
    =
    0,
  \end{equation}
  which upon sampling at $x_{h, j}$ yields the system of equations
  \begin{equation}
    \label{eq:160}
    \ppd{u_{h, j}}{t}
    -
    \sum_{n = -\infty}^{+\infty} c_{h, n} u_{h, j - n}
    =
    0.   
  \end{equation}
  (Recall that on the equispaced grid we have $x_{h, j} - x_{h, n} = x_{h, j-n}$ and that the bandwidth limited function $u_h$ is in one-to-one correspondence with its grid values at the grid $\left\{ x_{h,j}\right\}_{j=-\infty}^{+\infty}$.) System~\eqref{eq:160} is the same as system~\eqref{eq:136}, it suffices to relabel the summation index.
\end{proof}

Theorem~\ref{thr:2} concludes our search for exact correspondence between the lattice and continuous models. First we observe that Theorem~\eqref{thr:1} is just a special case of Theorem~\ref{thr:2}. Indeed, if we choose
\begin{equation}
  \label{eq:161}
  \vec{c}_h =
  \begin{bmatrix}
    \vdots \\ c_{h, 2} \\ c_{h, 1} \\ c_{h, 0} \\ c_{h, 1} \\ c_{h, 2} \\  \vdots
  \end{bmatrix}
  =_{\bydefinition}
  \frac{\ccontinuous^2}{h^2}
  \begin{bmatrix}
    \vdots \\ 0 \\ 1 \\ -2 \\ 1 \\ 0 \\  \vdots
  \end{bmatrix},
\end{equation}
that is if we consider only the nearest neighbour interactions, then the system of differential equations~\eqref{eq:136} reduces to the well-known form~\eqref{eq:47}, and quick calculation reveals that
\begin{multline}
  \label{eq:162}
  \frac{1}{h}
  \InverseFourierTransform{\frac{\FourierTransformSemidiscrete{\vec{c}_h}}{\xi^2}}
  =
  \ccontinuous^2
  \InverseFourierTransform{ \frac{\exponential{\iunit \xi h} - 2 + \exponential{-\iunit \xi h}}{h^2 \xi^2} \chi_{\xi \in \left[ -\frac{\pi}{h}, \frac{\pi}{h} \right]}}
  \\
  =
  -
  \ccontinuous^2
  \InverseFourierTransform{
    \FourierTransform{\frac{1}{h} \UnitTriangle \left( \frac{x}{h} \right)}
    \FourierTransform{\frac{1}{h} \sinc\left( \frac{\pi x}{h} \right)}
  }
  =
  -
  \ccontinuous^2 \convolution{\frac{1}{h} \UnitTriangle \left( \frac{x}{h} \right)}{\frac{1}{h} \sinc\left( \frac{\pi x}{h} \right)}.
\end{multline}
Consequently, the corresponding continuous problem~\eqref{eq:151} reads
\begin{equation}
  \label{eq:163}
  \ppd{u}{t}
  -
  \ccontinuous^2
  \convolution{
    \left(
      \frac{1}{h} \convolution{\UnitTriangle \left( \frac{x}{h} \right)}{\frac{1}{h} \sinc\left( \frac{\pi x}{h} \right)}
    \right)
  }
  {
    \ppd{u}{x}
  }
  =
  0
  .
\end{equation}
In virtue of the associativity of the convolution and the identity property for bandwidth limited functions~\eqref{eq:79} we then see that~\eqref{eq:163} reduces to~\eqref{eq:129}. (Recall that from Theorem~\eqref{thr:2} we know that $u$ is necessarily a bandwidth limited function.) Higher order finite differences type approximations of the second derivative operator can be analysed in a similar manner.

An important corollary of Theorem~\ref{thr:2} deals with the standard wave equation. If we are interested in the solution of
\begin{equation}
  \label{eq:164}
  \ppd{u}{t} - \ccontinuous^2 \ppd{u}{x} = 0,
\end{equation}
then the discretisation scheme~\eqref{eq:136} that gives the bandwidth limited interpolant $u_h$ that \emph{exactly} solves the standard wave equation~\eqref{eq:164} \emph{everywhere on the real line} is obtained by the choice of coefficients $\vec{c}_h$ such that
\begin{equation}
  \label{eq:165}
  \left(
    \frac{1}{h}
    \InverseFourierTransform{\frac{\FourierTransformSemidiscrete{\vec{c}_h}}{\xi^2}}
  \right)
  =
  -
  \ccontinuous^2
  \diracdelta.
\end{equation}
Indeed, if we prepare the coefficients in this way, then~\eqref{eq:140} reduces to
\begin{equation}
  \label{eq:166}
  \ppd{u_h}{t} - \ccontinuous^2 \convolution{\diracdelta}{\ppd{u_h}{x}} = 0,
\end{equation}
which is the same as~\eqref{eq:164}. From~\eqref{eq:165} we see that the coefficients $\vec{c}_h$ can be calculated as
\begin{equation}
  \label{eq:167}
  \vec{c}_h
  =
  -
  \ccontinuous^2 h \InverseFourierTransformSemidiscrete{\xi^2}
  .
\end{equation}
For the $j$-th element of the coefficients vector we thus have
\begin{equation}
  \label{eq:168}
  c_{h, j}
  =
  -
  \frac{\ccontinuous^2 h}{2\pi}
  \int_{\xi = - \frac{\pi}{h}}^{\frac{\pi}{h}}
  \xi^2
  \exponential{\iunit \xi x_{h, j}}
  \,
  \diff \xi
  =
  -
  \frac{\ccontinuous^2 h}{2\pi}
  \int_{\xi = - \frac{\pi}{h}}^{\frac{\pi}{h}}
  \xi^2
  \exponential{\iunit j \xi h}
  \,
  \diff \xi
  =
  \begin{cases}
    -
    \frac{\ccontinuous^2 h}{2\pi}
    \left(
    (-1)^j
    \frac{4 \pi}{h^3 j^2}
    \right)
    =
    (-1)^{j+1}
    \ccontinuous^2
    \frac{2}{h^2j^2},  &\qquad j \not = 0,\\
    - \ccontinuous^2 \frac{\pi^2}{3 h^2}, &\qquad j = 0,
  \end{cases}
\end{equation}
and we get the following corollary of Theorem~\eqref{thr:2}.

\begin{corollary}[Exact lattice model for the standard wave equation]
  \label{crl:1}
  Let $h>0$, and let $\left\{x_{h, j}\right\}_{j=-\infty}^{+\infty}$, $x_{h, j} =_{\bydefinition} jh$, be the corresponding grid on the real line $\R$. Let~$\vec{u}_h(t) = \left\{ u_{h, j} (t)\right\}_{j=-\infty}^{+\infty}$ be a collection of grid values on the grid $\left\{x_{h, j} \right\}_{j=-\infty}^{+\infty}$, and let~$u_h(x,t)$ be the corresponding bandwidth limited interpolant of~$\vec{u}_h(t)$. 
  The grid values $\vec{u}_h(t)  = \left\{ u_{h, j} (t)\right\}_{j=-\infty}^{+\infty}$ solve the system of ordinary differential equations
  \begin{equation}
    \label{eq:169}
    \ddd{u_{h,j}}{t}
    -
    \sum_{m=-\infty}^{+ \infty} c_{h, j-m}u_{h, m}
    =
    0
  \end{equation}
  with
  \begin{equation}
    \label{eq:170}
    c_{h, j}
    =
    \ccontinuous^2
    \begin{cases}
      (-1)^{j+1}
      \frac{2}{h^2j^2}
      , & \qquad j \not = 0, \\
      - \frac{\pi^2}{3 h^2}, &\qquad j = 0,
    \end{cases}
  \end{equation}
  if and only if the bandwidth limited interpolant $u_h$ solves
  \begin{equation}
    \label{eq:171}
    \ppd{u_h}{t} - \ccontinuous^2 \ppd{u_h}{x} = 0
  \end{equation}
  provided that the initial conditions for the continuous problem~\eqref{eq:171} are bandwidth limited interpolants of initial conditions for the discrete problem~\eqref{eq:169}, see Theorem~\ref{thr:2}. 
\end{corollary}
The coefficients~\eqref{eq:170} are in fact the coefficients in the (infinite) second order spectral differentiation matrices known in the numerical analysis, see, for example, \cite[Chapter 2]{trefethen.ln:spectral}. We note that the structural condition $c_{h, 0} = - 2 \sum_{j=1}^{+ \infty} c_{h, j} $ that follows from~\eqref{eq:46} is satisfied for the choice~\eqref{eq:170}. This is not surprising. The matrix representing the convolution in~\eqref{eq:169} is a differentiation matrix corresponding to the second derivative and the derivative of a constant is equal to zero---we expect this to be preserved on the discrete level as well, which means that the sum of matrix row elements should be zero.

\section{Periodic lattice}
\label{sec:periodic-lattice}

The lattice (chain) model now describes $N$ equal mass particles in the spatial domain $[0, 2\pi]$ that are in the equilibrium positioned at equispaced grid
$\{x_{h, j}\}_{j=1}^{N}$,
\begin{subequations}
  \label{eq:172}
  \begin{equation}
    \label{eq:173}
    x_{h, j} = jh,
  \end{equation}
  where
  \begin{equation}
    \label{eq:174}
    h = \frac{2 \pi}{N},
  \end{equation}
\end{subequations}
see~Figure~\ref{fig:periodic-lattice}. Since the particular formulae might slightly differ for $N$ odd and even, we for the sake of simplicity assume that $N$ is even. (This allows us to follow, for example, the presentation in~\cite[Chapter 3]{trefethen.ln:spectral}.) The displacements of the individual particles are denoted as
\begin{equation}
  \label{eq:175}
  \{u_{h, j} (t)\}_{j=1}^{N}.
\end{equation}
The whole lattice is assumed periodic in space with the period $2 \pi$; we formally identify the particle at position $x_N = 2 \pi$ (the right chain end) with the particle at $x_0 =_\bydefinition 0$ (the left chain end) and so on, and we can formally write $\{u_{h, j} (t)\}_{j=-\infty}^{\infty}$ with
\begin{equation}
  \label{eq:176}
  u_{h, j+kN} =_{\bydefinition} u_{h, j}
\end{equation}
with $k \in \Z$ and $j = 1, \dots, N$. (Note that the point $x_0$ is not a part of the grid in the sense that we are not interested in the displacement at this point because the displacement at this point is due to periodicity identical to the displacement at point~$x_N$.) This manipulation yields a periodic extension of the finite length lattice to the infinite length lattice. Consequently, the analysis of the periodic lattice setting might build on already obtained results for the infinite lattice.

\begin{figure}
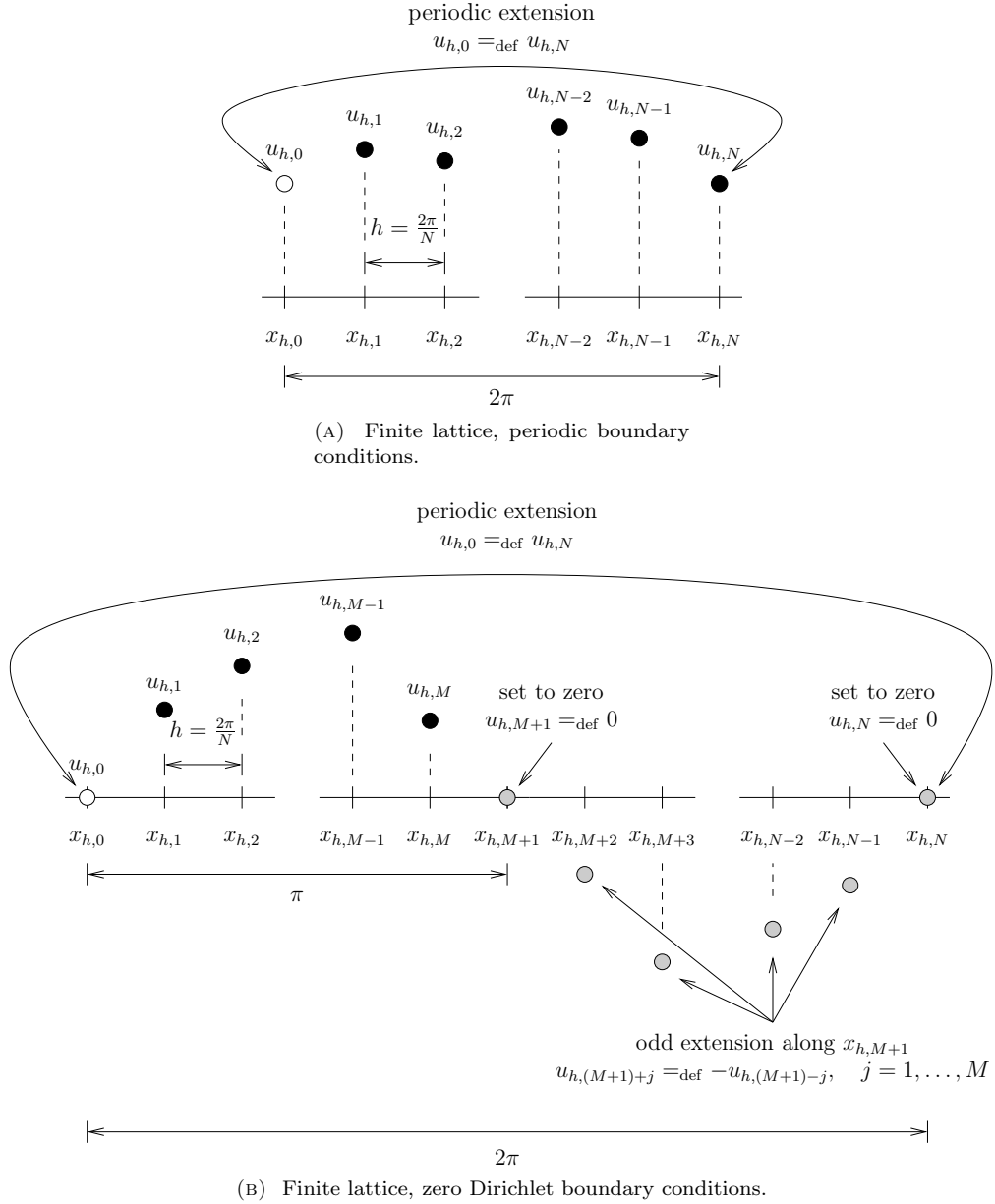

  \centering
  \subfloat[\label{fig:periodic-lattice} Finite lattice, periodic boundary conditions.]{\includegraphics[width=0.4\textwidth]{figures/periodic-lattice}}
  \\
  \subfloat[\label{fig:dirichlet-lattice} Finite lattice, zero Dirichlet boundary conditions.]{\includegraphics[width=0.7\textwidth]{figures/dirichlet-lattice}}
  \caption{Lattices.}
\end{figure}

\subsection{Periodic lattice model of interacting particles}
\label{sec:periodic-model}
The governing equations for the particles displacements $\{u_{h, j} (t)\}_{j=1}^{N}$ now read
\begin{subequations}
  \label{eq:177}  
  \begin{equation}
    \label{eq:178}
    \ddd{u_{h,j}}{t}
    -
    \sum_{m=1}^{N} \tilde{c}_{h, j-m}u_{h, m}
    =
    0
    ,
  \end{equation}
  where the coefficients $\left\{ \tilde{c}_{h, i} \right\}_{i=0}^{N-1}$ are given as
  \begin{equation}
    \label{eq:179}
    \tilde{c}_{h, i}
    =
    \begin{cases}
      c_{h, i}, & i=0, \dots, \frac{N}{2}, \\
      c_{h, N-i}, & i= \frac{N}{2}+1, \dots, N-1,
    \end{cases}
  \end{equation}
  which means that the coefficients vector is populated as
  \begin{equation}
    \label{eq:180}
    \widetilde{\vec{c}}_h
    =_{\bydefinition}
    \begin{bmatrix}
      c_{h, 0} \\ c_{h, 1} \\ \vdots \\ c_{h, \frac{N}{2} - 2} \\ c_{h, \frac{N}{2} - 1} \\ c_{h, \frac{N}{2}} \\ c_{h, \frac{N}{2} -1} \\ \vdots \\ c_{h, 2} \\ c_{h, 1} 
    \end{bmatrix}
    .
  \end{equation}
  The coefficients $\left\{ c _{h, i} \right\}_{i=0}^{\frac{N}{2}}$ are given numbers characterising interparticle interactions, wherein we always set
  \begin{equation}
    \label{eq:181}
    c_{h, \frac{N}{2}}=_{\bydefinition}0,
  \end{equation}
  to avoid self-interactions. The sum in~\eqref{eq:178} is interpreted using the $N$-periodicity assumption for the coefficients~$\left\{ \tilde{c}_{h, i} \right\}_{i=0}^{N-1}$ meaning that whenever the index~$j$ in~$\tilde{c}_{h, j}$ is negative, then we use
  \begin{equation}
    \label{eq:182}
    \tilde{c}_{h, j} =_{\bydefinition} \tilde{c}_{h, j + N}.
  \end{equation}
\end{subequations}
For example, if $N=6$, then $\tilde{c}_{h, -4} =_{\bydefinition} \tilde{c}_{h, 2}$ and so forth. Unlike in the previous case of infinite chain, the system~\eqref{eq:177} is now represented by a finite matrix,
\begin{equation}
  \label{eq:183}
  \ddd{}{t}
  \begin{bmatrix}
    u_{h, 1} \\
    u_{h, 2} \\
    \vdots \\
    u_{h, \frac{N}{2} - 1} \\
    u_{h, \frac{N}{2}} \\
    u_{h, \frac{N}{2} + 1} \\
    \vdots \\
    u_{h, N-1} \\
    u_{h, N}
  \end{bmatrix}
  -
  \begin{bmatrix}
    c_{h, 0} & c_{h, 1} & \cdots &  c_{h, \frac{N}{2} - 1} &  c_{h, \frac{N}{2}} &  c_{h, \frac{N}{2} - 1} & \cdots & c_{h, 2} & c_{h, 1} \\
    c_{h, 1} & c_{h, 0} & \cdots &  c_{h, \frac{N}{2} - 2} &  c_{h, \frac{N}{2} - 1} &  c_{h, \frac{N}{2} } & \cdots & c_{h, 3} & c_{h, 2} \\
    \vdots & \vdots & \ddots & \vdots & \vdots & \vdots & \reflectbox{$\ddots$} & \vdots & \vdots \\
    c_{h, \frac{N}{2} - 2} & c_{h, \frac{N}{2} - 1} & \cdots & c_{h, 0}  &  c_{h, 1} & c_{h, 2} & \cdots & c_{h, \frac{N}{2}} & c_{h, \frac{N}{2} - 1} \\

    c_{h, \frac{N}{2} - 1} & c_{h, \frac{N}{2} - 2} & \cdots & c_{h, 1}  &  c_{h, 0} & c_{h, 1} & \cdots & c_{h, \frac{N}{2} - 1} & c_{h, \frac{N}{2}} \\
    c_{h, \frac{N}{2}} & c_{h, \frac{N}{2} - 1} & \cdots & c_{h, 2}  &  c_{h, 1} & c_{h, 0} & \cdots & c_{h, \frac{N}{2} - 2} & c_{h, \frac{N}{2} - 1} \\
    \vdots & \vdots & \reflectbox{$\ddots$} & \vdots & \vdots & \vdots & \ddots & \vdots & \vdots \\
    c_{h, 2} & c_{h, 3} & \cdots &  c_{h, \frac{N}{2}-1} &  c_{h, \frac{N}{2}-2} &  c_{h, \frac{N}{2} - 1} & \cdots & c_{h, 0} & c_{h, 1} \\
    c_{h, 1} & c_{h, 2} & \cdots &  c_{h, \frac{N}{2}} &  c_{h, \frac{N}{2}-1} &  c_{h, \frac{N}{2} - 2} & \cdots & c_{h, 1} & c_{h, 0} \\
  \end{bmatrix}
  \begin{bmatrix}
    u_{h, 1} \\
    u_{h, 2} \\
    \vdots \\
    u_{h, \frac{N}{2} - 1} \\
    u_{h, \frac{N}{2}} \\
    u_{h, \frac{N}{2} + 1} \\
    \vdots \\
    u_{h, N-1} \\
    u_{h, N}
  \end{bmatrix}
  =
  0,
\end{equation}
which we also write in the matrix--vector form as
\begin{equation}
  \label{eq:184}
  \ddd{\vec{u}_h}{t} - \tensorq{L}_h^{\text{periodic}} \vec{u}_h = \vec{0}
\end{equation}
with the obvious identification of vector $\vec{u}_h$ and matrix $\tensorq{L}_h^{\text{periodic}}$. (The matrix form~\eqref{eq:183} of~\eqref{eq:178}, see the middle row of~\eqref{eq:183}, clearly shows that the coefficients $\left\{ c _{h, i} \right\}_{i=0}^{\frac{N}{2}}$ indeed characterise interactions of $i$-th neighbours.) For example, in the case of lattice formed by six particles with the nearest neighbour interaction we set $N=6$, and the coefficients are given as
\begin{equation}
  \label{eq:185}
  c_{h, i} =_{\bydefinition}
  \frac{\ccontinuous^2}{h^2}
  \begin{cases}
    -2,&  i=0, \\
    1,& i=1, \\
    0,& i=2, 3,
  \end{cases}
\end{equation}
which yields the standard matrix
\begin{equation}
  \label{eq:186}
  \ddd{}{t}
  \begin{bmatrix}
    u_{h, 1} \\
    u_{h, 2} \\
    u_{h, 3} \\
    u_{h, 4} \\
    u_{h, 5} \\
    u_{h, 6} \\
  \end{bmatrix}
  -
  \frac{\ccontinuous^2}{h^2}
  \begin{bmatrix}
    -2 & 1 &  &  &  & 1 \\
    1 & -2 & 1 &  &  &  \\
       & 1 & -2 & 1 &  &  \\
       &  & 1 & -2 & 1 &  \\
       &  &  & 1 & -2 & 1 \\
    1 &  &  &  & 1 & -2 \\
  \end{bmatrix}
  \begin{bmatrix}
    u_{h, 1} \\
    u_{h, 2} \\
    u_{h, 3} \\
    u_{h, 4} \\
    u_{h, 5} \\
    u_{h, 6} \\
  \end{bmatrix}
  =
  \begin{bmatrix}
    0 \\
    0 \\
    0 \\
    0 \\
    0 \\
    0
  \end{bmatrix}
  .
\end{equation}
We note that this matrix also arises in the discretisation of the second derivative operator $\ddd{}{x}$ (with periodic boundary condition) by the centred second order finite difference scheme. Finally, we observe that matrix of interest in~\eqref{eq:183} is now a finite dimensional \emph{circulant matrix}---compare with the infinite lattice setting, see~\eqref{eq:42}, where we have been dealing with an infinite circulant matrix/Laurent operator. 

\subsection{Reconstruction procedure---from discrete grid values to a periodic function of real variable}
\label{sec:reconstr-proc-from-2}
In order to identify the exact continuous counterpart of a lattice model, we again first need to carefully discuss the reconstruction procedure, that is the way how to reconstruct a real valued function $u_h(t, x)$ out of the collection of discrete grid values $\vec{u}_h = \{u_{h, j} (t)\}_{j=1}^{N}$. Following the previous section we stick to an analogue of bandwidth limited interpolant, but due to periodicity we restrict ourselves to \emph{discrete} wavenumbers, and, furthermore, we consider only \emph{finite number} of wavenumbers. Indeed, if we want the \emph{wavelike} \emph{ansatz} $\exponential{\iunit \xi x}$ to be a $2\pi$ periodic function, which means that we want $\exponential{\iunit \xi x} = \exponential{\iunit \xi (x + 2 \pi)}$, then we get the condition~$\xi = k$, $k \in \Z$. Furthermore, if we want to stay in the first Brillouin zone $k \in \left[ -\frac{\pi}{h}, \frac{\pi}{h} \right]$, then we must set
\begin{equation}
  \label{eq:187}
  k = -\frac{N}{2} + 1, \dots, \frac{N}{2},
\end{equation}
see the formula for the interparticle distance~\eqref{eq:174}.

\subsubsection{Discrete Fourier transform and bandwidth limited interpolant}
\label{sec:discr-four-transf}
In the periodic setting we thus make one step further in the discretisation of formulae for the Fourier transform, and we introduce the \emph{discrete Fourier transform} and its inverse as follows.
\begin{definition}[Discrete Fourier transform]
\label{dfn:5}
\begin{subequations}
  \label{eq:discrete-fourier-transform}
  Let $\vec{u}_h =  \{u_{h, j} \}_{j=1}^{N}$ be a collection of grid values at the grid $\{x_{h, j}\}_{j=1}^{N}$ in the interval $[0, 2\pi]$, see~\eqref{eq:172}, and let $\vec{v}_h = \{v_{h, j} \}_{k = -\frac{N}{2} + 1}^{\frac{N}{2}}$ be a collection of values in the wavenumber space. We denote
  \begin{align}
    \label{eq:188}
    \tensor{\left( \FourierTransformDiscrete{\vec{u}_h} \right)}{_k}
    &=_{\bydefinition}
      h
      \sum_{j= 1}^{N}
      \tensor{\left(\vec{u}_h\right)}{_j}
      \exponential{- \iunit k x_{h, j}}
      ,
      \qquad
      k = -\frac{N}{2} + 1, \dots, \frac{N}{2},
    \\
    \label{eq:189}
    \tensor{\left( \InverseFourierTransformDiscrete{\vec{v}_h} \right)}{_j}
    &=_{\bydefinition}
      \frac{1}{2\pi}
      \sum_{k = -\frac{N}{2} + 1}^{\frac{N}{2}}
      \tensor{\left(\vec{v}_h\right)}{_k}
      \exponential{\iunit k x_{h, j}}
      ,
      \qquad
      j = 1, \dots, N.
  \end{align}
  and we call the collection $\FourierTransformDiscrete{\vec{u}_h}$ the dicrete Fourier transform of $\vec{u}_h$, while $\InverseFourierTransformDiscrete{\vec{v}_h}$ is called the inverse discrete Fourier transform of $\vec{v}_h$.
\end{subequations}
\end{definition}
(Compare with~\eqref{eq:fourier-transform} and~\eqref{eq:semidiscrete-fourier-transform}.) We emphasise that we use the notation $\tensor{\left(\vec{u}_h\right)}{_j} = u_{h, j}$ for the value at the $j$-th grid point, in particular the index $j$ is not necessarily an index corresponding to an arrangement of grid values is a column vector. Similarly, $\tensor{\left( \InverseFourierTransformDiscrete{v_h} \right)}{_j}$ denotes the $j$-th wavenumber coefficient, not necessarily an index corresponding to an arrangement of wavenumber coefficients in a column vector. This is further emphasised by the fact that we refer to~$\vec{u}_h$ as a \emph{collection} of grid values and not as a \emph{vector} of grid values.  We note that both transforms work with discrete collections of length~$N$, and that we \emph{include} the factor $h$ in the definition of discrete Fourier transform. In computer codes for discrete Fourier transform the factor~$h$ is typically not included, and, furthermore, the computer codes assume a particular arrangement of the collection $\vec{u}_h$ in a corresponding column vector.

If we want to obtain (an analogue of) the bandwidth limited interpolant~$u_h$ of the grid values~$\vec{u}_h$, then we can follow the procedure known from the semidiscrete case, see Figure~\eqref{fig:semidiscrete-fourier-transform}---we first take the forward discrete Fourier transform, and then we apply the inverse continuous Fourier transform. The \emph{forward discrete} Fourier transform of grid values $\vec{u}_h$ yields
\begin{equation}
  \label{eq:190}
  \widehat{\vec{u}}_h =_{\bydefinition} \FourierTransformDiscrete{\vec{u}_h}.
\end{equation}
In order to take the \emph{inverse continuous} Fourier transform of $\widehat{\vec{u}}_h$, we need to interpret this collection of values as a function in the Fourier space, that is as a function of variable $\xi$. A naive guess might be the following---we could have set
\begin{equation}
  \label{eq:191}
  \widehat{\vec{u}}_h (\xi) \stackrel{?}{=}_{\bydefinition} \sum_{k=-\frac{N}{2}+1}^{\frac{N}{2}} \tensor{\left( \widehat{\vec{u}}_h \right)}{_k} \diracdelta_k (\xi),
\end{equation}
which would, upon the application of inverse continuous Fourier transform, yield
\begin{equation}
  \label{eq:192}
  u_h \stackrel{?}{=}
  \InverseFourierTransform{ \sum_{k=-\frac{N}{2}+1}^{\frac{N}{2}} \tensor{\left( \widehat{\vec{u}}_h \right)}{_k} \diracdelta_k (\xi)}
  =
  \frac{1}{2\pi}
  \sum_{k = -\frac{N}{2} + 1}^{\frac{N}{2}}
  \tensor{\left(\widehat{\vec{u}}_h\right)}{_k}
  \exponential{\iunit k x}
  ,
\end{equation}
where we have used formulae~\eqref{eq:22}. Clearly, formula~\eqref{eq:192} can lead to complex functions, which we would like to avoid.

In order to avoid this possibility, we need to make a subtle change in~\eqref{eq:191} to fix the asymmetry in the wavenumber domain. (The wavenumber takes values $k= - \frac{N}{2} + 1, \dots, \frac{N}{2}$ and this collection is not symmetric with respect to the origin. We would like to see the wavenumber $-\frac{N}{2}$ as well.) We thus change~\eqref{eq:191} to its symmetrised version
\begin{equation}
  \label{eq:193}
  \widehat{\vec{u}}_h (\xi) =_{\bydefinition}
  \frac{1}{2}\tensor{\left( \widehat{\vec{u}}_h \right)}{_{\frac{N}{2}}} \diracdelta_{-\frac{N}{2}} (\xi)
  +
  \sum_{k=-\frac{N}{2}+1}^{\frac{N}{2}-1} \tensor{\left( \widehat{\vec{u}}_h \right)}{_k} \diracdelta_{k}(\xi)
  +
  \frac{1}{2}\tensor{\left( \widehat{\vec{u}}_h \right)}{_{\frac{N}{2}}} \diracdelta_{\frac{N}{2}} (\xi),
\end{equation}
that is we split the last element in the sum~\eqref{eq:191} as
\begin{multline}
  \label{eq:194}
  \left.
    \tensor{\left( \widehat{\vec{u}}_h \right)}{_k} \diracdelta \left( \xi - k \right)
  \right|_{k = \frac{N}{2}}
  =
  \left.
    \left(
      h
      \sum_{j= 1}^{N}
      \tensor{\left(\vec{u}_h\right)}{_j}
      \exponential{- \iunit k x_{h, j}}
    \right)
    \diracdelta \left( \xi - k \right)
  \right|_{k = \frac{N}{2}}
  =
  \frac{1}{2}
  \left(
    h
    \sum_{j= 1}^{N}
    \tensor{\left(\vec{u}_h\right)}{_j}
    \exponential{- \iunit \frac{N}{2} x_{h, j}}
  \right)
  +
  \frac{1}{2}
  \left(
    h
    \sum_{j= 1}^{N}
    \tensor{\left(\vec{u}_h\right)}{_j}
    \exponential{- \iunit \frac{N}{2} x_{h, j}}
  \right)
  \\
  =
  \frac{1}{2}
  \left(
    h
    \sum_{j= 1}^{N}
    \tensor{\left(\vec{u}_h\right)}{_j}
    \exponential{\iunit \frac{N}{2} x_{h, j}}
  \right)
  +
  \frac{1}{2}
  \left(
    h
    \sum_{j= 1}^{N}
    \tensor{\left(\vec{u}_h\right)}{_j}
    \exponential{- \iunit \frac{N}{2} x_{h, j}}
  \right)
  =
  \frac{1}{2}
  \left.
    \tensor{\left( \widehat{\vec{u}}_h \right)}{_k}
    \diracdelta \left( \xi - k \right)
  \right|_{k = -\frac{N}{2}}
  +
  \frac{1}{2}
  \left.
    \tensor{\left( \widehat{\vec{u}}_h \right)}{_k}
    \diracdelta \left( \xi - k \right)
  \right|_{k = \frac{N}{2}}
  ,
\end{multline}
where we have used the fact that
\begin{equation}
  \label{eq:195}
  \exponential{- \iunit \frac{N}{2}  x_{h,j}} = \exponential{- \iunit \pi j} = \exponential{\iunit \pi j} = \exponential{\iunit \frac{N}{2}  x_{h,j}}.
\end{equation}
Now take the inverse continuous Fourier transform of $\widehat{\vec{u}}_h (\xi)$ defined in~\eqref{eq:193}, and we get the sought formula for the bandwidth limited interpolant in the periodic case,
\begin{subequations}
  \label{eq:bandwidth-limited-intrepolant-periodic}
  \begin{equation}
    \label{eq:196}
    u_h =_{\bydefinition}
    \frac{1}{2\pi}
    \left(
      \frac{1}{2}
      \tensor{\left(\widehat{\vec{u}}_h\right)}{_{-\frac{N}{2}}}
      \exponential{- \iunit \frac{N}{2}  x}
      +
      \sum_{k = -\frac{N}{2} + 1}^{\frac{N}{2} - 1}
      \tensor{\left(\widehat{\vec{u}}_h\right)}{_k}
      \exponential{\iunit k x}
      +
      \frac{1}{2}
      \tensor{\left(\widehat{\vec{u}}_h\right)}{_{\frac{N}{2}}}
      \exponential{\iunit \frac{N}{2}  x}
    \right)
    ,
    \qquad
    x \in [0, 2 \pi],
  \end{equation}
  where we define the originally missing element $\tensor{\left(\widehat{\vec{u}}_h\right)}{_{-\frac{N}{2}}}$ using the $\frac{N}{2}$ element of the original collection $\widehat{\vec{u}}_h$, that is
  \begin{equation}
    \label{eq:197}
    \tensor{\left(\widehat{\vec{u}}_h\right)}{_{-\frac{N}{2}}}
    =_{\bydefinition}
    \tensor{\left(\widehat{\vec{u}}_h\right)}{_{\frac{N}{2}}}
    ,
  \end{equation}
  and where the collection $\widehat{\vec{u}}_h$ is obtained by the discrete Fourier transform of the grid values $\vec{u}_h$, that is
  \begin{equation}
    \label{eq:198}
    \widehat{\vec{u}}_h =_{\bydefinition} \FourierTransformDiscrete{\vec{u}_h}.
  \end{equation}
\end{subequations}

We note that if evaluate~\eqref{eq:196} at arbitrary point $x = x_{h, j}$, where $j \in \left\{1, \dots, N \right\}$, then~\eqref{eq:196} with the convention \eqref{eq:197} yields exactly the $j$-th element of inverse discrete Fourier transform collection $\vec{u}_h$, see~\eqref{eq:189}. Consequently, the bandwidth limited interpolant formula~\eqref{eq:196} is indeed a decent extension of the inverse discrete Fourier transform~\eqref{eq:189} for arbitrary point $x$.  (Note however that the analogue of the bandwidth limited interpolant \emph{is not obtained} by blindly replacing $x_{h, j}$ in~\eqref{eq:189} by continuous variable $x$.) The symmetrisation trick is essential only if we want to work with the bandwidth limited interpolants, it does not require us to extend the concept of discrete Fourier transform to include one more wavenumber, the discrete Fourier transform machinery works as usual.

We see that~\eqref{eq:bandwidth-limited-intrepolant-periodic} in fact defines a $2\pi$ periodic function on the whole real line, and this function is effectively a periodic extension of the interpolant on the original interval $[0, 2\pi]$. The \emph{continuous} Fourier transform applied to the (periodic extension of) bandwidth limited interpolant $u_h$ behaves as expected,
\begin{multline}
  \label{eq:199}
  \FourierTransform{u_h}
  =
  \FourierTransform{
    \frac{1}{2\pi}
    \left(
      \frac{1}{2}
      \tensor{\left(\widehat{\vec{u}}_h\right)}{_{-\frac{N}{2}}}
      \exponential{- \iunit \frac{N}{2}  x}
      +
      \sum_{k = -\frac{N}{2} + 1}^{\frac{N}{2} - 1}
      \tensor{\left(\widehat{\vec{u}}_h\right)}{_k}
      \exponential{\iunit k x}
      +
      \frac{1}{2}
      \tensor{\left(\widehat{\vec{u}}_h\right)}{_{\frac{N}{2}}}
      \exponential{\iunit \frac{N}{2}  x}
    \right)
  }
  \\
  =
  \frac{1}{2}
  \tensor{\left(\widehat{\vec{u}}_h\right)}{_{-\frac{N}{2}}}
  \diracdelta_{-\frac{N}{2}}
  +
  \sum_{k = -\frac{N}{2} + 1}^{\frac{N}{2} - 1}
  \tensor{\left(\widehat{\vec{u}}_h\right)}{_k}
  \diracdelta_{k}
  +
  \frac{1}{2}
  \tensor{\left(\widehat{\vec{u}}_h\right)}{_{\frac{N}{2}}}
  \diracdelta_{\frac{N}{2}}
  =
  \sideset{}{^\prime}\sum_{k = -\frac{N}{2}}^{\frac{N}{2}}
  \tensor{\left(\widehat{\vec{u}}_h\right)}{_k}
  \diracdelta_{k}
  =
  \sideset{}{^\prime}\sum_{k = -\frac{N}{2}}^{\frac{N}{2}}
  \tensor{\left( \FourierTransformDiscrete{\vec{u}_h} \right)}{_k}
  \diracdelta_{k}
  ,
\end{multline}
with the convention that
\begin{equation}
  \label{eq:200}
  \sideset{}{^\prime}\sum_{k = -\frac{N}{2}}^{\frac{N}{2}}
  \tensor{\left(\widehat{\vec{u}}_h\right)}{_k}
  \diracdelta_{k}
  =_{\bydefinition}
  \frac{1}{2}
  \tensor{\left(\widehat{\vec{u}}_h\right)}{_{-\frac{N}{2}}}
  \diracdelta_{-\frac{N}{2}}
  +
  \sum_{k = -\frac{N}{2} + 1}^{\frac{N}{2} - 1}
  \tensor{\left(\widehat{\vec{u}}_h\right)}{_k}
  \diracdelta_{k}
  +
  \frac{1}{2}
  \tensor{\left(\widehat{\vec{u}}_h\right)}{_{\frac{N}{2}}}
  \diracdelta_{\frac{N}{2}}
  .
\end{equation}
The primed sum symbol $\sideset{}{^\prime}\sum_{k = -\frac{N}{2}}^{\frac{N}{2}}$ means that the terms with the lowest/highest wavenumbers $-\frac{N}{2}$ and $\frac{N}{2}$ are premultiplied by the factor $\frac{1}{2}$, and the element $\tensor{\left(\widehat{\vec{u}}_h\right)}{_{-\frac{N}{2}}}$  that is missing in the Fourier space collection $\widehat{\vec{u}}_h$ is defined as in~\eqref{eq:197}. We note that~\eqref{eq:199} can be also rewritten as
\begin{equation}
  \label{eq:201}
  u_h = \InverseFourierTransform{
    \sideset{}{^\prime}\sum_{k = -\frac{N}{2}}^{\frac{N}{2}}
    \tensor{\left( \FourierTransformDiscrete{\vec{u}_h} \right)}{_k}
    \diracdelta_{k}
  },
\end{equation}
which is an analogue of~\eqref{eq:bandwidth-limited-interpolant-semidiscrete}. Using the primed sum notation we can start completing the diagram showing the construction of the bandwidth limited interpolant, see Figure~\eqref{fig:discrete-fourier-transform}.

An analogue of the explicit formula~\eqref{eq:bandwidth-limited-interpolant-semidiscrete-physical-space} for the bandwidth limited interpolant is obtained by the same manipulation as in Section~\ref{sec:bandw-limit-interp}. In particular, we define the $m$-th discrete Dirac distribution $\vec{\diracdelta}_h^{h, m}$ on the space of grid values as
\begin{equation}
  \label{eq:202}
  \vec{\diracdelta}_h^{h, m}
  =_{\bydefinition}
  \begin{bmatrix}
    0 \\
    \vdots \\
    0 \\
    1 \\
    0 \\
    \vdots \\
    0
  \end{bmatrix}
  ,  
\end{equation}
where the nonzero element is placed at the $m$-th position in this vector of length $N$. (Meaning that the only nonzero element is the grid value at point $x_{h, m}$.) The forward discrete Fourier transform~\eqref{eq:188} of $\vec{\diracdelta}_h^{h, m}$ then yields
\begin{equation}
  \label{eq:203}
  \tensor{
    \left(
      \FourierTransformDiscrete{\vec{\diracdelta}_h^{h, m}}
    \right)
  }{_k}
  =
  h
  \exponential{- \iunit k x_{h, m}}
  .
\end{equation}
Subsequent application of the continuous inverse Fourier transform then gives the bandwidth limited interpolant of the Dirac function placed at the $m$-th grid point $x_{h, m}$,
\begin{multline}
  \label{eq:204}
  \diracdelta_h^{h, m}
  =
  \InverseFourierTransform{
    \FourierTransformDiscrete{\vec{\diracdelta}_h^{h, m}} (\xi)
  }
  =
  \InverseFourierTransform{
    \sideset{}{^\prime}\sum_{k = -\frac{N}{2}}^{\frac{N}{2}}
    \tensor{\left( \FourierTransformDiscrete{\vec{\diracdelta}_h^{h, m}} \right)}{_k}
    \diracdelta_{k}
  }
  \\
  =
  \InverseFourierTransform{
    \frac{1}{2}
    \tensor{
      \left(
        \FourierTransformDiscrete{\vec{\diracdelta}_h^{h, m}}
      \right)
    }{_{-\frac{N}{2}}}
    \diracdelta_{-\frac{N}{2}}
    +
    \sum_{k = -\frac{N}{2} + 1}^{\frac{N}{2} - 1}
    \tensor{
      \left(
        \FourierTransformDiscrete{\vec{\diracdelta}_h^{h, m}}
      \right)
    }{_k}
    \diracdelta_{k}
    +
    \frac{1}{2}
    \tensor{
      \left(
        \FourierTransformDiscrete{\vec{\diracdelta}_h^{h, m}}
      \right)
    }{_{\frac{N}{2}}}
    \diracdelta_{\frac{N}{2}}
  }
  \\
  =
  \frac{h}{2\pi}
  \left(
    \frac{1}{2}
    \exponential{-\iunit \frac{N}{2} \left( x - x_{h, m}\right)}
    +
    \sum_{k = -\frac{N}{2} + 1}^{\frac{N}{2} - 1}
    \exponential{\iunit k \left( x - x_{h, m}\right)}
    +
    \frac{1}{2}
    \exponential{\iunit \frac{N}{2} \left( x - x_{h, m}\right)}
  \right)
  \\
  =
  \frac{h}{2\pi}
  \cos \left( \frac{x - x_{h, m}}{2} \right)
  \frac{\sin \left( \frac{N}{2} \left( x - x_{h, m}\right) \right)}{\sin \left( \frac{x - x_{h, m}}{2}\right)}
  ,
\end{multline}
where the last equality follows from the standard manipulation based on the formula for the sum of power series. We can thus conclude that the bandwidth limited interpolant $u_h$ of the grid values $\vec{u}_h$ is given as
\begin{subequations}
  \label{eq:bandwidth-limited-interpolant-periodic-physical-space}
  \begin{align}
    \label{eq:205}
    u_h(x) &= \sum_{j=1}^{N}\tensor{\left(\vec{u}_h\right)}{_j} \Sdiracdelta_h^h \left( x - x_{h, j}\right), \\
    \label{eq:206}
    \Sdiracdelta_h^h \left(x \right)
           &=_{\bydefinition}
             \frac{h}{2\pi}
             \cos \left( \frac{x}{2} \right)
             \frac{\sin \left( \frac{\pi}{h} x \right)}{\sin \left( \frac{x}{2} \right)}
             .
  \end{align}
\end{subequations}  
This is the well known periodic domain analogue to the infinite domain formula~\eqref{eq:bandwidth-limited-interpolant-semidiscrete-physical-space}, see, for example, \cite[Chapter 3]{trefethen.ln:spectral} and \cite[Chapter 5]{boyd.jp:chebyshev}.

The construction of bandwidth limited interpolant is schematically summarised in Figure~\ref{fig:discrete-fourier-transform}. Note that in the right column in Figure~\ref{fig:discrete-fourier-transform} we slightly abuse the notation---we formally write the discrete Fourier transform as a function of (continuous) wavenumber $\xi$ and we pick the right (discrete) wavenumbers via the multiplication with the Dirac distribution. We can formulate our findings as a Lemma.

\begin{lemma}[Bandwidth limited interpolant, periodic lattice]
  \label{lm:7}
  Let $\vec{u}_h =  \{u_{h, j} \}_{j=1}^{N}$ be a collection of grid values at the grid $\{x_{h, j}\}_{j=1}^{N}$ in the interval $[0, 2\pi]$, see~\eqref{eq:172}. Let $u_h$ be the bandwidth limited interpolant of $\vec{u}_h$, that is a $2\pi$ periodic function~$u_h$ constructed as
\begin{equation}
  \label{eq:427}
  u_h =_{\bydefinition}
  \InverseFourierTransform{
    \sideset{}{^\prime}\sum_{k = -\frac{N}{2}}^{\frac{N}{2}}
    \tensor{\left( \FourierTransformDiscrete{\vec{u}_h} \right)}{_k}
    \diracdelta_{k}
  },
\end{equation}
where the primed sum symbol denotes the following operation
\begin{equation}
  \label{eq:439}
  \sideset{}{^\prime}\sum_{k = -\frac{N}{2}}^{\frac{N}{2}} a_k =_{\bydefinition} \frac{1}{2}a_{-\frac{N}{2}} +  \sideset{}{^\prime}\sum_{k = -\frac{N}{2}+1}^{\frac{N}{2}-1} a_k + \frac{1}{2}a_{\frac{N}{2}}.
\end{equation}
Then the bandwidth limited interpolant $u_h$ is given by the formula
\begin{subequations}
  \label{eq:436}
  \begin{align}
    \label{eq:437}
    u_h(x) &= \sum_{j=1}^{N}\tensor{\left(\vec{u}_h\right)}{_j} \Sdiracdelta_h^h \left( x - x_{h, j}\right), \\
    \label{eq:438}
    \Sdiracdelta_h^h \left(x \right)
           &=_{\bydefinition}
             \frac{h}{2\pi}
             \cos \left( \frac{x}{2} \right)
             \frac{\sin \left( \frac{\pi}{h} x \right)}{\sin \left( \frac{x}{2} \right)}
             .
  \end{align}
\end{subequations}  
\end{lemma}

\begin{figure}
  \centering
  \includegraphics[width=\textwidth]{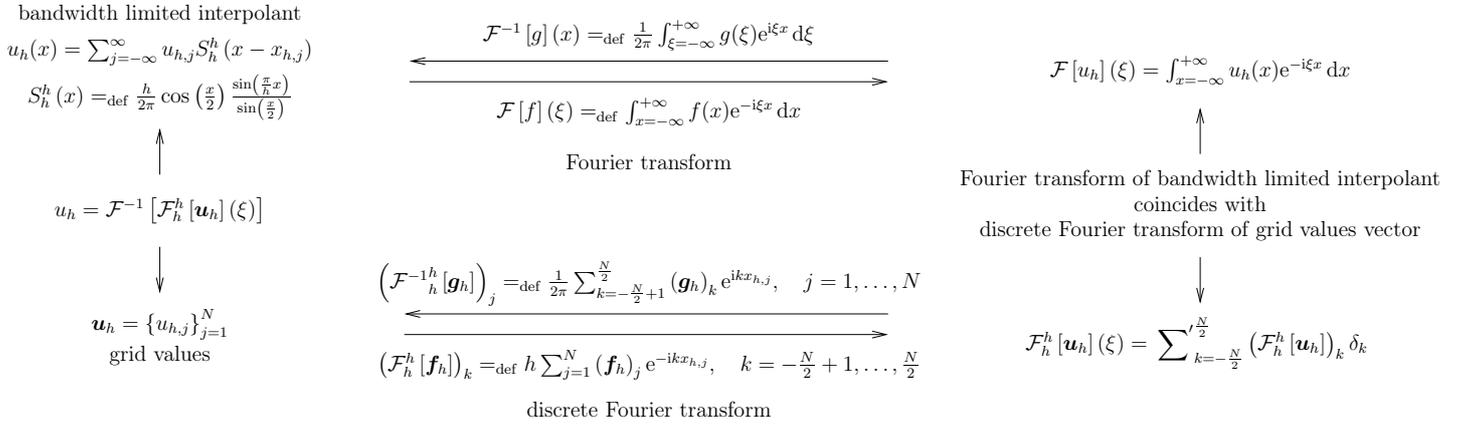}
  \caption{Discrete Fourier transform and construction of bandwidth limited interpolant $u_h$ of grid values vector $\vec{u}_h$, equispaced grid $x_{h, j} = jh$, $j = 1, \dots, N$, $h = \frac{2\pi}{N}$.}
  \label{fig:discrete-fourier-transform}
\end{figure}

\subsubsection{Discrete Fourier transform and evaluation of bandwidth limited interpolant at arbitrary point}
\label{sec:discr-four-transf-4}
The properties of discrete Fourier transform can be again used to evaluate the bandwidth limited interpolant $u_h(y)$ at any point $y \in [0, 2 \pi]$ in the physical space \emph{without the need to sum the series representation}~\eqref{eq:bandwidth-limited-interpolant-periodic-physical-space}. In order to evaluate~$u_h(y)$ we can indeed proceed indirectly in the Fourier space. By the definition of bandwidth limited interpolant, see~\eqref{eq:196}, we have
\begin{multline}
  \label{eq:207}
  u_h (y)=
  \frac{1}{2\pi}
  \left(
    \frac{1}{2}
    \tensor{\left(\widehat{\vec{u}}_h\right)}{_{-\frac{N}{2}}}
    \exponential{- \iunit \frac{N}{2}  y}
    +
    \sum_{k = -\frac{N}{2} + 1}^{\frac{N}{2} - 1}
    \tensor{\left(\widehat{\vec{u}}_h\right)}{_k}
    \exponential{\iunit k y}
    +
    \frac{1}{2}
    \tensor{\left(\widehat{\vec{u}}_h\right)}{_{\frac{N}{2}}}
    \exponential{\iunit \frac{N}{2}  y}
  \right)
  \\
  =
  \frac{1}{2\pi}
  \left(
    \frac{1}{2}
    \tensor{\left(\widehat{\vec{u}}_h\right)}{_{-\frac{N}{2}}}
    \exponential{- \iunit \frac{N}{2}  \left(y - x_{h, j}\right)}\exponential{- \iunit \frac{N}{2}  x_{h, j}}
    +
    \sum_{k = -\frac{N}{2} + 1}^{\frac{N}{2} - 1}
    \tensor{\left(\widehat{\vec{u}}_h\right)}{_k}
    \exponential{\iunit k \left(y - x_{h, j}\right)}\exponential{\iunit k  x_{h, j}}
    +
    \frac{1}{2}
    \tensor{\left(\widehat{\vec{u}}_h\right)}{_{\frac{N}{2}}}
    \exponential{\iunit \frac{N}{2}  \left(y - x_{h, j}\right)}\exponential{\iunit \frac{N}{2}  x_{h, j}}
  \right)
  \\
  =
  \frac{1}{2\pi}
  \left(
    \sum_{k = -\frac{N}{2} + 1}^{\frac{N}{2} - 1}
    \left[
      \tensor{\left(\widehat{\vec{u}}_h\right)}{_k}
      \exponential{\iunit k \left(y - x_{h, j}\right)}
    \right]
    \exponential{\iunit k  x_{h, j}}
    +
    \left[
      \frac{1}{2}
      \tensor{\left(\widehat{\vec{u}}_h\right)}{_{\frac{N}{2}}}
      \left(
        \exponential{- \iunit \frac{N}{2}  \left(y - x_{h, j}\right)} + \exponential{\iunit \frac{N}{2}  \left(y - x_{h, j}\right)}
      \right)
    \right]
    \exponential{\iunit \frac{N}{2}  x_{h, j}}
  \right)
  ,
\end{multline}
where $x_{h, j}$ is an arbitrary grid point. However, the last formula in~\eqref{eq:207} is in fact the $j$-th component of the inverse discrete Fourier transform of the collection $\vec{v}_h$ defined as
\begin{subequations}
  \label{eq:208}
  \begin{align}
    \label{eq:209}
    \tensor{\left( \vec{v}_h \right)}{_k} &=_{\bydefinition}
                                            \tensor{\left(\widehat{\vec{u}}_h\right)}{_k}
                                            \exponential{\iunit k \left(y - x_{h, j}\right)}, \qquad k = -\frac{N}{2} + 1, \dots, \frac{N}{2} -1, \\
    \tensor{\left( \vec{v}_h \right)}{_{\frac{N}{2}}} &=_{\bydefinition}     \frac{1}{2}
                                                        \tensor{\left(\widehat{\vec{u}}_h\right)}{_{\frac{N}{2}}}
                                                        \left(
                                                        \exponential{- \iunit \frac{N}{2}  \left(y - x_{h, j}\right)} + \exponential{\iunit \frac{N}{2}  \left(y - x_{h, j}\right)}
                                                        \right).
  \end{align}
\end{subequations}
In order to evaluate the bandwidth limited interpolant at an arbitrary point $y$ we can thus take the discrete Fourier transform of grid values, which gives us collection $\widehat{\vec{u}}_h$, form collection $\vec{v}_h$ via multiplication of the elements of $\widehat{\vec{u}}_h$ collection by given numbers, see~\eqref{eq:208}, and take the inverse discrete Fourier transform. The $j$-th component of such obtained vector gives us the value $u_h(y)$. This procedure is the counterpart of~\eqref{eq:85}, but now it is not just a curiosity, but a practical tool. The bandwidth limited interpolant~\eqref{eq:bandwidth-limited-interpolant-periodic-physical-space} can be evaluated at arbitrary point by taking two Fourier transforms and some multiplication, which is particularly convenient computationally provided that the discrete Fourier transform is done using a fast Fourier transform algorithm.

\subsubsection{Discrete Fourier transform and discrete periodic convolution}
\label{sec:discr-four-transf-1}
Now we proceed with the definition of discrete periodic convolution of collections $\vec{f}_h$ and $\vec{g}_h$ of length $M$. This is the discrete periodic domain analogue of discrete convolution formulae~\eqref{eq:21} (continuous case) and~\eqref{eq:81} (infinite lattice).
\begin{definition}[Discrete convolution, periodic lattice]
  \label{dfn:6}
  Let $\vec{f}_h =  \{f_{h, j} \}_{j=1}^{M}$ and $\vec{g}_h =  \{ g_{h, j} \}_{j=1}^{M}$ be collections of length $M$. We denote
  \begin{equation}
    \label{eq:210}
    \tensor{\left( \discreteperiodicconvolution{\vec{f}_h}{\vec{g}_h} \right)}{_j}
    =_{\bydefinition}
    h
    \sum_{m = 1}^{M}
    f_{h, j - m} g_{h, m}
    ,
  \end{equation}
  where $h =_{\bydefinition} \frac{2 \pi}{M}$. (Note the factor $h$ in the definition~\eqref{eq:210}.) The vector $\discreteconvolution{\vec{f}_h}{\vec{g}_h}$ of length $M$ with components $j=1, \dots, M$ given by~\eqref{eq:210} is called the \emph{discrete periodic convolution} of collections $\vec{f}_h$ and $\vec{g}_h$.
\end{definition}
  
The collections $\vec{f}_h$ and $\vec{g}_h$ can be interpreted as vectors, where the vector components are numbered from one to $M$ and the vectors could be vectors both in the physical domain or in the Fourier domain. (Note also the factor $h$ in the definition of convolution~\eqref{eq:210}.) The sum in~\eqref{eq:210} goes over all elements of the given vector, and the indices out of the range $1, \dots, M$ are obtained by the periodic extension, that is
\begin{equation}
  \label{eq:211}
  f_{h, -l} = _{\bydefinition} f_{h, -l + M}, \qquad l = 0, \dots, M - 1.
\end{equation}
If we further define the elementwise multiplication of vectors of length $M$ as a vector of length $M$ with the $m$-th element given by
\begin{equation}
  \label{eq:212}
  \tensor{\left( \tensorschur{\vec{f}_h}{\vec{g}_h} \right)}{_m}
  =
  f_{h, m} g_{h, m},
\end{equation}
then we recover the convolution formulae~\eqref{eq:25} and~\eqref{eq:26} (continuous case) and~\eqref{eq:83} and~\eqref{eq:84} (semidiscrete case) in the sense that the discrete Fourier transform of discrete periodic convolution is again product of discrete Fourier transforms and so forth. We have the following lemma.

\begin{lemma}[Convolution-to-multiplication property, periodic lattice]
  \label{lm:8}
  Let $\vec{f}_h =  \{f_{h, j} \}_{j=1}^{M}$ and $\vec{g}_h =  \{ g_{h, j} \}_{j=1}^{M}$ be given collections of grid values. Then
  \begin{subequations}
    \begin{align}
      \label{eq:213}
      \FourierTransformDiscrete{\discreteperiodicconvolution{\vec{f}_h}{\vec{g}_h}}
      &=
        \tensorschur{
        \FourierTransformDiscrete{\vec{f}_h}
        }
        {
        \FourierTransformDiscrete{\vec{g}_h}
        },
      \\
      \label{eq:214}
      \FourierTransformDiscrete{\tensorschur{\vec{f}_h}{\vec{g}_h}}
      &=
        \discreteperiodicconvolution{\FourierTransformDiscrete{\vec{f}_h}}{\FourierTransformDiscrete{\vec{g}_h}}
        .
    \end{align}  
  \end{subequations}
\end{lemma}
\begin{proof}
 The convolution-to-multiplication formulae are straightforward to prove by direct substitution. For example, concerning~\eqref{eq:213} we proceed as follows
\begin{multline}
  \label{eq:215}
  \tensor{
    \left(
      \FourierTransformDiscrete{\discreteperiodicconvolution{\vec{f}_h}{\vec{g}_h}}
    \right)
  }{_k}
  =
  \tensor{
    \left(
      \FourierTransformDiscrete{
        h
        \sum_{m = 1}^{M}
        f_{h, j - m} g_{h, m}
      }
    \right)
  }{_k}
  =
  h
  \sum_{j= 1}^{M}
  \left(
    h
    \sum_{m = 1}^{M}
    f_{h, j - m} g_{h, m}
  \right)
  \exponential{- \iunit k x_{h, j}}
  \\
  =
  h^2
  \sum_{j= 1}^{M}
  \sum_{m = 1}^{M}
  f_{h, j - m}
  \exponential{- \iunit k \left(x_{h, j} - x_{h, m}\right)}
  g_{h, m}
  \exponential{- \iunit k x_{h, m}}
  =
  h
  \sum_{m = 1}^{M}
  \left(
    h
    \sum_{j= 1}^{M}
    f_{h, j - m}
    \exponential{- \iunit k x_{h, j-m}}
  \right)
  g_{h, m}
  \exponential{- \iunit k x_{h, m}}
  \\
  =
  h
  \sum_{m = 1}^{M}
  \tensor{\left(\FourierTransformDiscrete{\vec{f}_h} \right)}{_k}
  g_{h, m}
  \exponential{- \iunit k x_{h, m}}
  =
  \tensor{\left(\FourierTransformDiscrete{\vec{f}_h} \right)}{_k}
  \left(
    h
    \sum_{m = 1}^{M}
    g_{h, m}
    \exponential{- \iunit k x_{h, m}}
  \right)
  \\
  =
  \tensor{\left(\FourierTransformDiscrete{\vec{f}_h} \right)}{_k}
  \tensor{\left(\FourierTransformDiscrete{\vec{g}_h} \right)}{_k}
  =
  \tensor{
    \left(
      \tensorschur{\FourierTransformDiscrete{\vec{f}_h}}{\FourierTransformDiscrete{\vec{f}_h}}
    \right)
  }{_k}
  ,
\end{multline}
where we have exploited the formula for the grid points~\eqref{eq:173} and the invariance of the discrete Fourier transform with respect to the shift.
\end{proof}

\subsubsection{Discrete Fourier transform diagonalises arbitrary circulant matrix}
\label{sec:discr-four-transf-3}
Since the discrete Fourier transform and the discrete periodic convolution inherit the properties of their continuous counterparts, we can expect that we recover, even at the discrete level, the continuous level results based on the convolution-to-multiplication formula. In particular, we show that \emph{the discrete Fourier transform can be used to solve eigenvalue problem for the discrete periodic convolution operator.} (For the continuous version of the same see the discussion following equation~\eqref{eq:35}.) Indeed, assume that $\vec{a}_h$ is a given collection $\vec{a}_h =_{\bydefinition} \left\{a_i\right\}_{i=1}^N$,
\begin{equation}
  \label{eq:216}
  \vec{a}_h
  =
  \begin{bmatrix}
    a_1 \\
    a_2 \\
    \vdots \\
    a_{N-1} \\
    a_{N}
  \end{bmatrix}
  ,
\end{equation}
and that we want to solve the eigenvalue problem
\begin{equation}
  \label{eq:217}
  \frac{1}{h}
  \discreteperiodicconvolution{\vec{a}_h}{\vec{g}_h} = \lambda \vec{g}_h
\end{equation}
for the eigenvalue $\lambda$ and the eigenvector $\vec{g}_h$. (Recall in the convolution sum we use the convention that the indices out of range $1, \dots, N$ are defined by the periodic extension, that is $a_0 = a_N$, $a_{-1} = a_{N-1}$ and so forth, see~\eqref{eq:211}.) In the matrix from the eigenvalue problem~\eqref{eq:217} reads
\begin{equation}
  \label{eq:218}
  \underbrace{
    \begin{bmatrix}
      a_N & a_{N-1} & a_{N-2} & \dots & a_1 \\
      a_1 & a_N & a_{N-1} & \dots & a_2 \\
      a_2 & a_1 & a_N & \dots & a_3 \\
      \vdots & \vdots & \vdots& \ddots & \vdots \\
      a_{N-1} & a_{N-2} & a_{N-3} & \dots & a_N
    \end{bmatrix}
  }_{\tensorq{A}_h = _{\bydefinition}}
  \begin{bmatrix}
    g_1 \\
    g_2 \\
    g_3 \\
    \vdots \\
    g_N
  \end{bmatrix}
  =
  \lambda
  \begin{bmatrix}
    g_1 \\
    g_2 \\
    g_3 \\
    \vdots \\
    g_N
  \end{bmatrix}
  ,
\end{equation}
wherein the matrix of interest $\tensorq{A}_h$ is a circulant matrix. Taking the discrete Fourier transform of~\eqref{eq:217} yields
\begin{equation}
  \label{eq:219}
  \frac{1}{h}
  \tensorschur{
    \FourierTransformDiscrete{\vec{a}_h}
  }
  {
    \FourierTransformDiscrete{\vec{g}_h}
  }
  =
  \lambda
  \FourierTransformDiscrete{\vec{g}_h}
  .
\end{equation}
This system of equations has $N$ solutions (eigenvector--eigenfunction pairs) $\vec{g}_{h, m}$, $\lambda_{h, m}$ of the form
\begin{subequations}
  \label{eq:220}
  \begin{align}
    \label{eq:221}
    \FourierTransformDiscrete{\vec{g}_{h,m}} &=   \vec{\diracdelta}_h^{h, m}, \\
    \label{eq:222}
    \lambda_{h, m} &= \frac{1}{h} \tensor{\left(\FourierTransformDiscrete{\vec{a}_h} \right)}{_m}.
  \end{align}
\end{subequations}
(Recall that the discrete Dirac distribution in the physical/Fourier space, $\vec{\diracdelta}_h^{h, m}$, is defined as a collection of zeros except of the $m$-th element that is equal to one, see also~\eqref{eq:202}.) The eigenvalues of the circulant matrix $\tensorq{A}_h$ are thus given as components of the discrete Fourier transform of the generating collection $\vec{a}_h$, while the eigenvectors are obtained as inverse discrete Fourier transform of discrete Dirac distributions in the Fourier space,
\begin{equation}
  \label{eq:223}
  \vec{g}_{h,m} = \InverseFourierTransformDiscrete{\vec{\diracdelta}_h^{h, m}}.
\end{equation}

Eigenvalue/eigenvector characterisation~\eqref{eq:220} is a well known fact usually phrased as ``columns of Fourier transform matrix are eigenvectors of circulant matrices'' or ``Fourier matrix diagonalises circulant matrices'', see, for example, \cite[Section 3.2]{davis.pj:circulant}. Indeed, the ``matrix version'' of the diagonalisation property is straightforward to show using the following manipulation. Let $\FourierTransformDiscreteMatrix$ denote the matrix representing the application of discrete Fourier transform of collection $\vec{f}_h$,  $\FourierTransformDiscrete{\vec{f}_h}$, that is $\FourierTransformDiscreteMatrix \vec{f}_h = \FourierTransformDiscrete{\vec{f}_h}$, and let $\vec{g}$ be an eigenvector associated to the eigenvalue $\lambda$. The application of discrete Fourier transform to the eigenvalue problem~\eqref{eq:217} and its equivalent matrix reformulation~\eqref{eq:218}, that is
\begin{subequations}
  \label{eq:224}
  \begin{align}
    \label{eq:225}
    \tensorq{A}_h \vec{g}_h &= \lambda \vec{g}_h, \\
    \label{eq:226}
    \frac{1}{h}
    \discreteperiodicconvolution{\vec{a}_h}{\vec{g}_h} &= \lambda \vec{g}_h,
  \end{align}
\end{subequations}
yields
\begin{subequations}
  \label{eq:227}
  \begin{align}
    \label{eq:228}
    \FourierTransformDiscreteMatrix \tensorq{A}_h \vec{g}_h &= \lambda \left( \FourierTransformDiscreteMatrix \vec{g}_h \right), \\
    \label{eq:229}
    \frac{1}{h}
    \tensorschur{\left( \FourierTransformDiscreteMatrix \vec{a}_h  \right)}{\left( \FourierTransformDiscreteMatrix \vec{g}_h  \right) } &=  \lambda \left( \FourierTransformDiscreteMatrix \vec{g}_h \right),
  \end{align}
\end{subequations}
which is the same as
\begin{subequations}
  \label{eq:230}
  \begin{align}
    \label{eq:231}
    \left( \FourierTransformDiscreteMatrix \tensorq{A}_h \inverse{ \left( \FourierTransformDiscreteMatrix \right)} \right) \FourierTransformDiscreteMatrix \vec{g}_h &= \lambda \left( \FourierTransformDiscreteMatrix \vec{g}_h \right), \\
    \label{eq:232}
    \frac{1}{h}
    \left( \diag \left( \FourierTransformDiscreteMatrix \vec{a}_h \right)  \right) \FourierTransformDiscreteMatrix \vec{g}_h  &=  \lambda \left( \FourierTransformDiscreteMatrix \vec{g}_h \right),
  \end{align}
\end{subequations}
where $\diag \vec{w}$ is the operator that takes vector $\vec{w}$ and returns a square diagonal matrix with the elements of vector $\vec{w}$ on the main diagonal, and where $\inverse{\left(\FourierTransformDiscreteMatrix\right)}$ denotes the inverse matrix to the matrix $\FourierTransformDiscreteMatrix$ representing the application of discrete Fourier transform. From~\eqref{eq:230} we see that
\begin{subequations}
  \label{eq:233}
  \begin{align}
    \label{eq:234}
    \FourierTransformDiscreteMatrix \tensorq{A}_h \inverse{ \left( \FourierTransformDiscreteMatrix \right) } &= \frac{1}{h} \diag \left( \FourierTransformDiscreteMatrix \vec{a}_h \right),\\
    \intertext{as well as the matrix counterpart of~\eqref{eq:220}}
    \label{eq:235}
    \FourierTransformDiscreteMatrix \vec{g}_{h, m} &= \vec{\diracdelta}_h^{h, m}, \\
    \label{eq:236}
    \lambda_{h, m} &= \frac{1}{h} \tensor{\left( \FourierTransformDiscreteMatrix \vec{a}_h \right)}{_m}.
  \end{align}
\end{subequations}
The diagonalisation property~\eqref{eq:234} of circulant matrices is usually proved by a direct computation, but here we see it as a straightforward consequence of simple properties of \emph{continuous} Fourier transform, see the discussion following~\eqref{eq:35}, that are mirrored in the discrete setting. For further reference, we can summarise our findings as a lemma.

\begin{lemma}[Eigenvalues and eigenvectors of circulant matrices]
  \label{lm:9}
  Let $\tensorq{A}_h$ be a circulant matrix generated by the vector $\vec{a}_h$,
  \begin{equation}
    \label{eq:421}
      \vec{a}_h
  =
  \begin{bmatrix}
    a_1 \\
    a_2 \\
    \vdots \\
    a_{N-1} \\
    a_{N}
  \end{bmatrix}
  ,
\end{equation}
that is
  \begin{equation}
    \label{eq:411}
      \tensorq{A}_h
      =
      \begin{bmatrix}
        a_N & a_{N-1} & a_{N-2} & \dots & a_1 \\
        a_1 & a_N & a_{N-1} & \dots & a_2 \\
        a_2 & a_1 & a_N & \dots & a_3 \\
        \vdots & \vdots & \vdots& \ddots & \vdots \\
        a_{N-1} & a_{N-2} & a_{N-3} & \dots & a_N
      \end{bmatrix}
      .
    \end{equation}
    Then the matrix $\tensorq{A}_h$ is diagonalised by the matrix representing the discrete Fourier transform~\eqref{eq:discrete-fourier-transform} and the eigenvalues of $\tensorq{A}_h$ are obtained by the discrete Fourier transform of the vector $\vec{a}_h$. In particular it holds
    \begin{equation}
      \label{eq:423}
       \FourierTransformDiscreteMatrix \tensorq{A}_h \inverse{ \left( \FourierTransformDiscreteMatrix \right) } = \frac{1}{h} \diag \left( \FourierTransformDiscreteMatrix \vec{a}_h \right).
    \end{equation}
\end{lemma}
The matrix $\FourierTransformDiscreteMatrix$ in Lemma~\eqref{lm:9} is the matrix that represents the discrete Fourier transform~\eqref{eq:discrete-fourier-transform}. In particular, in interpreting the discrete Fourier transform as an action of a matrix to a column vector we have to respect a particular arrangement of vector elements in the given column vector. This is the point where computer implementations of Fourier transform might differ from the convention used in~\eqref{eq:discrete-fourier-transform} and therefore care must be taken during implementation.

We note that \emph{all circulants share the same eigenvectors}---the formula~\eqref{eq:223} for the eigenvectors is independent of the generating vector $\vec{a}_h$. The circulant matrices however differ in eigenvalues. Finally, if the circulant matrix has an additional structure as in~\eqref{eq:183} wherein the ``middle'' row is ``symmetric'' with respect to $c_{h, 0}$, then the discrete Fourier transform of the first column is guaranteed to produce real eigenvalues, and we are in fact dealing with a variant of discrete cosine transform of values $c_{h, 0}, \dots, c_{h, \frac{N}{2}}$. In particular, we have the following.

\begin{lemma}[Eigenvalues of circulant matrices $\tensorq{C}_{\left(2M + 2\right) \times \left(2M + 2\right)}$]
  \label{lm:1}
  Let $M \in \N$ and let $N =_{\bydefinition} 2M + 2$. Let $\tensorq{C}_{\left(2M + 2\right) \times \left(2M + 2\right)}$ be the $\left(2M + 2\right) \times \left(2M + 2\right)$ circulant matrix generated by the collection
  \begin{equation}
    \vec{c}_h =_{\bydefinition}
    \begin{bmatrix}
      c_{h, 1} \\ \vdots \\ c_{h, \frac{N}{2} - 2} \\ c_{h, \frac{N}{2} - 1} \\ c_{h, \frac{N}{2}} \\ c_{h, \frac{N}{2} -1} \\ \vdots \\ c_{h, 2} \\ c_{h, 1} \\ c_{h, 0}
    \end{bmatrix},
  \end{equation}
  as
  \begin{equation}
    \label{eq:237}
    \tensorq{C}_{\left(2M + 2\right) \times \left(2M + 2\right)}
    =_{\bydefinition}
    \begin{bmatrix}
      c_{h, 0} & c_{h, 1} & \cdots &  c_{h, \frac{N}{2} - 1} &  c_{h, \frac{N}{2}} &  c_{h, \frac{N}{2} - 1} & \cdots & c_{h, 2} & c_{h, 1} \\
      c_{h, 1} & c_{h, 0} & \cdots &  c_{h, \frac{N}{2} - 2} &  c_{h, \frac{N}{2} - 1} &  c_{h, \frac{N}{2} } & \cdots & c_{h, 3} & c_{h, 2} \\
      \vdots & \vdots & \ddots & \vdots & \vdots & \vdots & \reflectbox{$\ddots$} & \vdots & \vdots \\
      c_{h, \frac{N}{2} - 2} & c_{h, \frac{N}{2} - 1} & \cdots & c_{h, 0}  &  c_{h, 1} & c_{h, 2} & \cdots & c_{h, \frac{N}{2}} & c_{h, \frac{N}{2} - 1} \\

      c_{h, \frac{N}{2} - 1} & c_{h, \frac{N}{2} - 2} & \cdots & c_{h, 1}  &  c_{h, 0} & c_{h, 1} & \cdots & c_{h, \frac{N}{2} - 1} & c_{h, \frac{N}{2}} \\
      c_{h, \frac{N}{2}} & c_{h, \frac{N}{2} - 1} & \cdots & c_{h, 2}  &  c_{h, 1} & c_{h, 0} & \cdots & c_{h, \frac{N}{2} - 2} & c_{h, \frac{N}{2} - 1} \\
      \vdots & \vdots & \reflectbox{$\ddots$} & \vdots & \vdots & \vdots & \ddots & \vdots & \vdots \\
      c_{h, 2} & c_{h, 3} & \cdots &  c_{h, \frac{N}{2}-1} &  c_{h, \frac{N}{2}-2} &  c_{h, \frac{N}{2} - 1} & \cdots & c_{h, 0} & c_{h, 1} \\
      c_{h, 1} & c_{h, 2} & \cdots &  c_{h, \frac{N}{2}} &  c_{h, \frac{N}{2}-1} &  c_{h, \frac{N}{2} - 2} & \cdots & c_{h, 1} & c_{h, 0} \\
    \end{bmatrix},
  \end{equation}
  and let $c_{h, \frac{N}{2}} =_{\bydefinition} 0$. Then the eigenvalues of~$\tensorq{C}_{\left(2M + 2\right) \times \left(2M + 2\right)}$ are real, and they are given by the formula
  \begin{equation}
    \label{eq:238}
    \lambda_{h, m}
    =
    c_{h, 0}
    +
    2
    \sum_{j= 1}^{M}
    c_{h, j}
    \cos \left(j mh \right)
    ,
    \qquad
    m = -\frac{N}{2} + 1, \dots, \frac{N}{2},
  \end{equation}
  where $h = \frac{\pi}{M+1}$.

  Eigenvalue formula~\eqref{eq:238} implies that the circulant matrix~$\tensorq{C}_{\left(2M + 2\right) \times \left(2M + 2\right)}$ with the structure \eqref{eq:237} has $M$ eigenvalues~$\lambda_{h, m}$, $m = 1, \dots, M$, of algebraic and geometric multiplicity two, and two eigenvalues of algebraic and geometric multiplicity one, namely the eigenvalues $\lambda_{h, 0}$ and $\lambda_{h, M+1}$. The diagonal form~$\tensorq{D}_{\tensorq{C}_{\left(2M + 2\right) \times \left(2M + 2\right)}}$ of the circulant matrix~$\tensorq{C}_{\left(2M + 2\right) \times \left(2M + 2\right)}$, in a properly chosen basis, thus reads
  \begin{equation}
    \label{eq:239}
    \tensorq{D}_{\tensorq{C}_{\left(2M + 2\right) \times \left(2M + 2\right)}}
    =
    \begin{bmatrix}
      \lambda_{h, 1} & & & & & & & \\
                     & \ddots & & & & & & \\ 
                     & & \lambda_{h, M} & & & & & \\
                     & & & \lambda_{h, M+1} & & & & \\
                     & & & & \lambda_{h, M} & & & \\
                     & & & & & \ddots & & \\
                     & & & & & & \lambda_{h, 1} & \\
                     & & & & & & & \lambda_{h, 0} 
    \end{bmatrix}
    .
  \end{equation}
\end{lemma}

\begin{proof}
  We know, see~\eqref{eq:220}, that the eigenvalues of circulant matrices are given by the discrete Fourier transform of the collection $\vec{c}_h$ representing the circulant matrix. In our case we thus have
  \begin{equation}
    \label{eq:240}
    \lambda_{h, m}
    =
    \frac{1}{h}
    \tensor{\left( \FourierTransformDiscrete{\vec{c}_h} \right)}{_m}
    =
    \sum_{j= 1}^{N}
    \tensor{\left(\vec{c}_h\right)}{_j}
    \exponential{- \iunit m x_{h, j}}
    ,
    \qquad
    m = -\frac{N}{2} + 1, \dots, \frac{N}{2},
  \end{equation}
  which yields
  \begin{multline}
    \label{eq:241}
    \lambda_{h, m}
    =
    \sum_{j= 1}^{\frac{N}{2} - 1}
    c_{h, j}
    \exponential{- \iunit m x_{h, j}}
    +
    c_{h, \frac{N}{2}}
    \exponential{- \iunit m x_{h, \frac{N}{2}}}
    +
    \sum_{j= 1}^{\frac{N}{2}-1}
    c_{h, j}
    \exponential{- \iunit m x_{h, N-j}}
    +
    c_{h, 0}
    \exponential{- \iunit m x_{h, N}}
    \\
    =
    c_{h, 0}
    \exponential{- \iunit m x_{h, N}}
    +
    \sum_{j= 1}^{\frac{N}{2} - 1}
    c_{h, j}
    \left(
      \exponential{- \iunit m x_{h, j}}
      +
      \exponential{- \iunit m x_{h, N-j}}
    \right)
    \\
    =
    c_{h, 0}
    \exponential{- \iunit m x_{h, N}}
    +
    \sum_{j= 1}^{\frac{N}{2} - 1}
    c_{h, j}
    \left(
      \exponential{- \iunit m x_{h, j}}
      +
      \exponential{ \iunit m x_{h, j}}
    \right)
    =
    c_{h, 0}
    +
    2
    \sum_{j= 1}^{\frac{N}{2} - 1}
    c_{h, j}
    \cos \left(j mh \right)
    ,
  \end{multline}
  where we have used $x_{h, j} = jh = j \frac{\pi}{M+1} = j \frac{2 \pi}{N}$. (We recall that $N = 2M +2$ and that $c_{h, \frac{N}{2}} =_{\bydefinition} 0$.) This is the sought formula~\eqref{eq:238}. The remaining part of Lemma~\eqref{lm:1} follows straightforwardly from~\eqref{eq:238}. 
\end{proof}

We note that the formula for the eigenvalues~\eqref{eq:238} is in fact a variant of discrete cosine transform of the vector $\vec{c}_h$.

\subsubsection{Discrete Fourier transform and differentiation}
\label{sec:discr-four-transf-2}
Concerning the differentiation of a function reconstructed from the grid values $\vec{u}_h = \{u_{h, j} (t)\}_{j=1}^{N}$, we again define the differentiation via the differentiation of the corresponding bandwidth limited interpolant $u_h$. The grid values of the derivative are thus defined as
\begin{equation}
  \label{eq:242}
  \tensor{\left( \dd{^n}{x^n} \vec{u}_h \right)}{_j} =_{\bydefinition} \left. \dd{^n}{x^n} u_h \right|_{x = x_{h, j}},
\end{equation}
which in virtue of~\eqref{eq:bandwidth-limited-intrepolant-periodic} leads to
\begin{multline}
  \label{eq:243}
  \tensor{\left( \dd{^n}{x^n} \vec{u}_h \right)}{_j}
  =
  \left.
    \frac{1}{2\pi}
    \left(
      \frac{1}{2}
      \left(
        - \iunit \frac{N}{2}
      \right)^n
      \tensor{\left(\widehat{\vec{u}}_h\right)}{_{-\frac{N}{2}}}
      \exponential{- \iunit \frac{N}{2}  x}
      +
      \sum_{k = -\frac{N}{2} + 1}^{\frac{N}{2} - 1}
      \left(\iunit k \right)^n
      \tensor{\left(\widehat{\vec{u}}_h\right)}{_k}
      \exponential{\iunit k x}
      +
      \frac{1}{2}
      \left(
        \iunit \frac{N}{2}
      \right)^n
      \tensor{\left(\widehat{\vec{u}}_h\right)}{_{\frac{N}{2}}}
      \exponential{\iunit \frac{N}{2}  x}
    \right)
  \right|_{x = x_{h, j}}
  \\
  =
  \frac{1}{2 \pi}
  \left(
    \sum_{k = -\frac{N}{2} + 1}^{\frac{N}{2} - 1}
    \left(\iunit k \right)^n
    \tensor{\left(\widehat{\vec{u}}_h\right)}{_k}
    \exponential{\iunit k x_{h, j}}
    +
    \left(
      \iunit \frac{N}{2}
    \right)^n
    \tensor{\left(\widehat{\vec{u}}_h\right)}{_{\frac{N}{2}}}
    \exponential{\iunit \frac{N}{2} x_{h, j}}
    \frac{1}{2}
    \Bigl(
    (-1)^n
    \exponential{-\iunit N x_{h, j}}
    +
    1
    \Bigr)
  \right)
  \\
  =
  \begin{cases}
    \frac{1}{2 \pi}
    \sum_{k = -\frac{N}{2} + 1}^{\frac{N}{2} - 1}
    \left(\iunit k \right)^n
    \tensor{\left(\widehat{\vec{u}}_h\right)}{_k}
    \exponential{\iunit k x_{h, j}}
    , & \text{$n$ odd}
    \\
    \frac{1}{2 \pi}
    \sum_{k = -\frac{N}{2} + 1}^{\frac{N}{2}}
    \left(\iunit k \right)^n
    \tensor{\left(\widehat{\vec{u}}_h\right)}{_k}
    \exponential{\iunit k x_{h, j}}
    , & \text{$n$ even}
        ,
  \end{cases}
\end{multline}
where we have used the formula for the grid points~\eqref{eq:172}. (See also~\cite[Chapter 3]{trefethen.ln:spectral} for an alternative reasoning regarding the difference between the odd and even powers of $n$.) A brief inspection of the last formula thus \emph{almost} allows to reproduce the derivative formula known from the semidiscrete case~\eqref{eq:100}. In particular, we can write
\begin{subequations}
  \label{eq:244}
  \begin{equation}
    \label{eq:245}
    \dd{^n}{x^n} \vec{u}_h = \InverseFourierTransformDiscrete{\left( \tensorschur{\iunit \vec{k} \right)^n}{\FourierTransformDiscrete{\vec{u}_h}}},
  \end{equation}
  where the elements of the symbol $\left(\iunit \vec{k} \right)^n$, $m=-\frac{N}{2}+1, \dots, \frac{N}{2}$, read
  \begin{equation}
    \label{eq:246}
    \tensor{
      \left(
        \left(\iunit \vec{k} \right)^n
      \right)
    }{_m}
    =_{\bydefinition}
    \left(\iunit m\right)^n,
  \end{equation}
  and wherein we for $n$ odd set
  \begin{equation}
    \label{eq:247}
    \tensor{
      \left(
        \left(\iunit \vec{k} \right)^n
      \right)
    }{_{\frac{N}{2}}}
    =_{\bydefinition}
    0
    .
  \end{equation}
\end{subequations}
We can summarise our findings in the following definition. 

\begin{definition}[Fourier transform based differentiation, periodic lattice]
  \label{dfn:8}
  Let $\vec{f}_h =  \{f_{h, j} \}_{j=1}^{N}$ be a vector of grid values at the grid $\{x_{h, j}\}_{j=1}^{N}$ in the interval $[0, 2\pi]$, see~\eqref{eq:172}. We set
  \begin{equation}
    \label{eq:440}
    \dd{^n}{x^n} \vec{f}_h =_{\bydefinition} \InverseFourierTransformDiscrete{\left( \tensorschur{\iunit \vec{k} \right)^n}{\FourierTransformDiscrete{\vec{f}_h}}},
  \end{equation}
  where the elements of the Fourier space collection $\left(\iunit \vec{k} \right)^n$, $m=-\frac{N}{2}+1, \dots, \frac{N}{2}$, read
  \begin{equation}
    \label{eq:441}
    \tensor{
      \left(
        \left(\iunit \vec{k} \right)^n
      \right)
    }{_m}
    =_{\bydefinition}
    \left(\iunit m\right)^n,
  \end{equation}
  and wherein for $n$ odd we set
  $
    \tensor{
      \left(
        \left(\iunit \vec{k} \right)^n
      \right)
    }{_{\frac{N}{2}}}
    =_{\bydefinition}
    0
    $.
\end{definition}

\subsection{Correspondence between discrete and continuous models}
\label{sec:corr-betw-discr-1}

Now we are in a position to prove an analogue of Theorem~\eqref{thr:2} on the equivalence between discrete and continuous models.

\begin{theorem}[Equivalence between discrete system of ordinary differential equations for grid values and the corresponding partial differential equation---general interaction, periodic lattice]
  \label{thr:3}
  Let $N \in \N$ be an even number, and let $\left\{x_{h, j}\right\}_{j=1}^{N}$, where $x_{h, j} =_{\bydefinition} jh$ with $h =_{\bydefinition} \frac{2 \pi}{N}$, be the corresponding grid on the interval $[0, 2\pi]$. Let~$\vec{u}_h(t) = \left\{ u_{h, j} (t)\right\}_{j=1}^{N}$ be a vector of grid values on the grid $\left\{x_{h, j} \right\}_{j=1}^{N}$, and let~$u_h(x,t)$ be the corresponding bandwidth limited interpolant of~$\vec{u}_h(t)$, that is
  \begin{subequations}
    \label{eq:248}
    \begin{equation}
      \label{eq:249}
      u_h(x, t)
      =
      \sum_{j=1}^{N}u_{h, j} \Sdiracdelta_h^h \left( x - x_{h, j}\right)
      ,
    \end{equation}
    where
    \begin{equation}
      \label{eq:250}
      \Sdiracdelta_h^h \left(x \right)
      =_{\bydefinition}
      \frac{h}{2\pi}
      \cos \left( \frac{x}{2} \right)
      \frac{\sin \left( \frac{\pi}{h} x \right)}{\sin \left( \frac{x}{2} \right)}.
    \end{equation}
  \end{subequations}
  Let the grid values $\vec{u}_h(t)  = \left\{ u_{h, j} (t)\right\}_{j=1}^{N}$ solve the initial value problem for the system of ordinary differential equations
  \begin{subequations}
    \label{eq:251}
    \begin{align}
      \label{eq:252}
      \ddd{u_{h,j}}{t}
      -
      \sum_{m=1}^{N} \tilde{c}_{h, j-m}u_{h, m}
      &=
        0,
      \\
      \label{eq:253}
      \left. u_{h, j} \right|_{t=0} &= u_{h, j}^0,\\
      \label{eq:254}
      \left. \dd{u_{h, j}}{t} \right|_{t=0} &= v_{h, j}^0,
    \end{align}
  \end{subequations}
  with the coefficients $\left\{ \tilde{c}_{h, i} \right\}_{i=0}^{N-1}$ generated out of coefficients $\left\{ c _{h, i} \right\}_{i=0}^{\frac{N}{2}}$ as in~\eqref{eq:177}, meaning that the matrix form of~\eqref{eq:251} reads
  \begin{equation}
    \label{eq:255}
    \ddd{}{t}
    \begin{bmatrix}
      u_{h, 1} \\
      u_{h, 2} \\
      \vdots \\
      u_{h, \frac{N}{2} - 1} \\
      u_{h, \frac{N}{2}} \\
      u_{h, \frac{N}{2} + 1} \\
      \vdots \\
      u_{h, N-1} \\
      u_{h, N}
    \end{bmatrix}
    -
    \begin{bmatrix}
      c_{h, 0} & c_{h, 1} & \cdots &  c_{h, \frac{N}{2} - 1} &  c_{h, \frac{N}{2}} &  c_{h, \frac{N}{2} - 1} & \cdots & c_{h, 2} & c_{h, 1} \\
      c_{h, 1} & c_{h, 0} & \cdots &  c_{h, \frac{N}{2} - 2} &  c_{h, \frac{N}{2} - 1} &  c_{h, \frac{N}{2} } & \cdots & c_{h, 3} & c_{h, 2} \\
      \vdots & \vdots & \ddots & \vdots & \vdots & \vdots & \reflectbox{$\ddots$} & \vdots & \vdots \\
      c_{h, \frac{N}{2} - 2} & c_{h, \frac{N}{2} - 1} & \cdots & c_{h, 0}  &  c_{h, 1} & c_{h, 2} & \cdots & c_{h, \frac{N}{2}} & c_{h, \frac{N}{2} - 1} \\

      c_{h, \frac{N}{2} - 1} & c_{h, \frac{N}{2} - 2} & \cdots & c_{h, 1}  &  c_{h, 0} & c_{h, 1} & \cdots & c_{h, \frac{N}{2} - 1} & c_{h, \frac{N}{2}} \\
      c_{h, \frac{N}{2}} & c_{h, \frac{N}{2} - 1} & \cdots & c_{h, 2}  &  c_{h, 1} & c_{h, 0} & \cdots & c_{h, \frac{N}{2} - 2} & c_{h, \frac{N}{2} - 1} \\
      \vdots & \vdots & \reflectbox{$\ddots$} & \vdots & \vdots & \vdots & \ddots & \vdots & \vdots \\
      c_{h, 2} & c_{h, 3} & \cdots &  c_{h, \frac{N}{2}-1} &  c_{h, \frac{N}{2}-2} &  c_{h, \frac{N}{2} - 1} & \cdots & c_{h, 0} & c_{h, 1} \\
      c_{h, 1} & c_{h, 2} & \cdots &  c_{h, \frac{N}{2}} &  c_{h, \frac{N}{2}-1} &  c_{h, \frac{N}{2} - 2} & \cdots & c_{h, 1} & c_{h, 0} \\
    \end{bmatrix}
    \begin{bmatrix}
      u_{h, 1} \\
      u_{h, 2} \\
      \vdots \\
      u_{h, \frac{N}{2} - 1} \\
      u_{h, \frac{N}{2}} \\
      u_{h, \frac{N}{2} + 1} \\
      \vdots \\
      u_{h, N-1} \\
      u_{h, N}
    \end{bmatrix}
    =
    \begin{bmatrix}
      0 \\
      0 \\
      \vdots \\
      0 \\
      0 \\
      0 \\
      \vdots \\
      0 \\
      0
    \end{bmatrix}
    .
  \end{equation}
  Let $u$ solve, for $x \in \R$ on the real line, the initial value problem 
  \begin{subequations}
    \label{eq:256}
    \begin{align}
      \label{eq:257}
      \ppd{u}{t}
      +
      \convolution{
      \left(
      \frac{1}{h}
      \InverseFourierTransform{\frac{\FourierTransformSemidiscrete{\vec{C}_h}}{\xi^2}}
      \right)
      }
      {
      \ppd{u}{x}
      }
      &=
        0
        ,
      \\
      \label{eq:258}
      \left. u \right|_{t=0} &= u_h^0,\\
      \label{eq:259}
      \left. \dd{u}{t} \right|_{t=0} &= v_h^0,
    \end{align}
  \end{subequations}
  where the symbol $\vec{C}_h$ denotes the infinite collection~$\left\{ C_{h, i}\right\}_{i=-\infty}^{+\infty}$ constructed from the first column  $\widetilde{\vec{c}}_h$ of matrix in~\eqref{eq:255},
  \begin{equation}
    \label{eq:260}
    \widetilde{\vec{c}}_h =_{\bydefinition}
    \begin{bmatrix}
      c_{h, 0} \\ c_{h, 1} \\ \vdots \\ c_{h, \frac{N}{2} - 2} \\ c_{h, \frac{N}{2} - 1} \\ c_{h, \frac{N}{2}} \\ c_{h, \frac{N}{2} -1} \\ \vdots \\ c_{h, 2} \\ c_{h, 1}
    \end{bmatrix},
  \end{equation}
  via the procedure defined in~\eqref{eq:263}, and let the initial data $u_h^0$ and $v_h^0$ in~\eqref{eq:258} and \eqref{eq:259} be the bandwidth limited interpolants of the initial data~\eqref{eq:253} and \eqref{eq:254}. The symbols $\FourierTransformSemidiscrete{\cdot}$ and $\InverseFourierTransform{\cdot}$ denote the \emph{semidiscrete} Fourier transform~\eqref{eq:69} and the \emph{continuous} inverse Fourier transform~\eqref{eq:20}.  Then the function $u$ is the solution to the continuous problem~\eqref{eq:256} if and only if $u = u_h$, where $u_h$ is the bandwidth limited interpolant $u_h$ of grid values $\vec{u}_h(t)  = \left\{ u_{h, j} (t)\right\}_{j=1}^{N}$ that solve the discrete problem~\eqref{eq:251}.
\end{theorem}

\begin{proof}
  We proceed as in proof of Theorem~\ref{thr:2}, but we must take into account the different construction of the interpolant and the discrete nature of objects in Fourier space. First we note that~\eqref{eq:252} can be rewritten using the discrete periodic convolution operator, see~\eqref{eq:210}, in the form
  \begin{equation}
    \label{eq:261}
    \ddd{\vec{u}_h}{t} - \frac{1}{h} \discreteperiodicconvolution{\widetilde{\vec{c}}_h}{\vec{u}_h} = \vec{0}.
  \end{equation}
  We take the discrete Fourier transform of~\eqref{eq:261}, we use the convolution-to-multiplication property of discrete Fourier transform, see~\eqref{eq:213}, and we get system of equations
  \begin{equation}
    \label{eq:262}
    \ddd{}{t}
    \FourierTransformDiscrete{\vec{u}_h}
    -
    \frac{1}{h} \tensorschur{\FourierTransformDiscrete{\widetilde{\vec{c}}_h}}{\FourierTransformDiscrete{\vec{u}_h}} = \vec{0}
  \end{equation}
  for $k = - \frac{N}{2}+1, \dots, \frac{N}{2}$.  Let us now formally extend the properly aligned collection $\left\{ \widetilde{c}_{h, m} \right\}_{m=0}^{N-1}$ by leading and trailing zeros. In particular, we define the collection~$\left\{ C_{h, i}\right\}_{i=-\infty}^{+\infty}$ as
  \begin{equation}
    \label{eq:263}
    \vec{C}_h =
    \begin{bmatrix}
      \vdots \\ C_{h, -2} \\ C_{h, -1} \\ C_{h, 0} \\ C_{h, 1} \\ C_{h, 2} \\  \vdots \\ C_{h, \frac{N}{2} - 1} \\ C_{h, \frac{N}{2}} \\ C_{h, \frac{N}{2} + 1} \\ C_{h, \frac{N}{2} + 2} \\ \vdots \\ C_{h, N-1} \\ C_{h, N} \\ C_{h, N+1} \\ C_{h, N+2} \\ \vdots    
    \end{bmatrix}
    =_{\bydefinition}
    \begin{bmatrix}
       \vdots \\ 0 \\ 0 \\ 0 \\ c_{h, 0} \\ c_{h, 1} \\ \vdots \\ c_{h, \frac{N}{2} - 2} \\ c_{h, \frac{N}{2} - 1} \\ c_{h, \frac{N}{2}} \\ c_{h, \frac{N}{2} -1}\\ \vdots \\ c_{h, 2} \\ c_{h, 1} \\ 0 \\ 0 \\ \vdots
    \end{bmatrix},
  \end{equation}
  that is
  \begin{equation}
    \label{eq:264}
    \tensor{\left( \vec{C}_h \right)}{_m}
    =
    \bydefinition
    \begin{cases}
      \widetilde{c}_{h, m-1}, & m = 1, \dots, N, \\
      0, & \text{otherwise}.
    \end{cases}
  \end{equation}
  The semidiscrete Fourier transform of $\vec{C}_h$, see~\eqref{eq:69}, reads
  \begin{equation}
    \label{eq:265}
    \FourierTransformSemidiscrete{\vec{C}_h} (\xi)
    =
    \left(
      h
      \sum_{j= - \infty}^{+ \infty}
      \tensor{\left(\vec{C}_h\right)}{_j}
      \exponential{- \iunit \xi x_{h, j}}
    \right)
    \chi_{\xi \in \left[ -\frac{\pi}{h}, \frac{\pi}{h} \right]},
    =
    \left(
      h
      \sum_{j= 1}^{N}
      \tensor{\left(\widetilde{\vec{c}}_h\right)}{_j}
      \exponential{- \iunit \xi x_{h, j}}
    \right)
    \chi_{\xi \in \left[ -\frac{\pi}{h}, \frac{\pi}{h} \right]}
    ,
  \end{equation}
  while for the discrete Fourier transform of $\widetilde{\vec{c}}_h$, see~\eqref{eq:188}, we get
  \begin{equation}
    \label{eq:266}
    \tensor{\left( \FourierTransformDiscrete{\widetilde{\vec{c}}_h} \right)}{_k}
    =_{\bydefinition}
    h
    \sum_{j= 1}^{N}
    \tensor{\left(\widetilde{\vec{c}}_h\right)}{_j}
    \exponential{- \iunit k x_{h, j}}
    , \qquad k = - \frac{N}{2} + 1, \dots, \frac{N}{2}.
  \end{equation}
  We can thus conclude that
\begin{equation}
  \label{eq:267}
  \left.
    \FourierTransformSemidiscrete{\vec{C}_h}
  \right|_{\xi = k} =  \tensor{\left( \FourierTransformDiscrete{\widetilde{\vec{c}}_h} \right)}{_k},
\end{equation}
holds for all wavenumbers $k = - \frac{N}{2} +1, \dots , \frac{N}{2}$. Furthermore, we note that if we formally define $ \tensor{\left( \FourierTransformDiscrete{\widetilde{\vec{c}}_h} \right)}{_{-\frac{N}{2}}} =_{\bydefinition}  \tensor{\left( \FourierTransformDiscrete{\widetilde{\vec{c}}_h} \right)}{_{\frac{N}{2}}}$, which is necessary for the symmetrisation trick, then~\eqref{eq:267} holds also for $k = - \frac{N}{2}$. Equality~\eqref{eq:266} is essential for the transition form the discrete setting to the continuous setting, which is our next step.

We go back to equation~\eqref{eq:262} and we rewrite it as a system of equations for wavenumbers $k$ ranging from $-\frac{N}{2}$ to $\frac{N}{2}$,
\begin{equation}
  \label{eq:268}
  \ddd{}{t}
  \tensor{
    \left(
      \FourierTransformDiscrete{\vec{u}_h}
    \right)
  }{_k}
  -
  \frac{1}{h}
  \left.
    \left( \FourierTransformSemidiscrete{\vec{C}_h} \right)
  \right|_{\xi = k}
  \tensor{
    \left(
      \FourierTransformDiscrete{\vec{u}_h}
    \right)
  }{_k}
  =
  0.
\end{equation}
Multiplication by $\diracdelta(\xi -k)$ then yields
\begin{equation}
  \label{eq:269}
  \ppd{}{t}
  \tensor{
    \left(
      \FourierTransformDiscrete{\vec{u}_h}
    \right)
  }{_k}
  \diracdelta_k
  -
  \frac{1}{h}
  \FourierTransformSemidiscrete{\vec{C}_h}
  \tensor{
    \left(
      \FourierTransformDiscrete{\vec{u}_h}
    \right)
  }{_k}
  \diracdelta_k
  =
  0,
\end{equation}
which upon taking the primed sum with respect to all wavenumbers $k$, see the convention introduced in~\eqref{eq:200}, gives us the single equation for the continuous Fourier transform of the bandwidth limited interpolant $u_h$,
\begin{equation}
  \label{eq:270}
  \ppd{}{t}
  \FourierTransform{u_h}
  -
  \frac{1}{h}
  \FourierTransformSemidiscrete{\vec{C}_h}
  \FourierTransform{u_h}
  =
  0
  .
\end{equation}
This brings us to the fully continuous setting, and we are in fact at the same position as in the proof of Theorem~\eqref{thr:2}, equation~\eqref{eq:148}. Following the same steps as in the proof of~Theorem~\eqref{thr:2} we thus end up with
   \begin{equation}
     \label{eq:271}
     \ppd{u_h}{t}
     +
     \convolution{
       \left(
         \frac{1}{h}
         \InverseFourierTransform{\frac{\FourierTransformSemidiscrete{\vec{C}_h}}{\xi^2}}
       \right)
     }
     {
       \ppd{u_h}{x}
     }
     =
     0,
   \end{equation}
   which was to prove.

   On the other hand, assume that $u$ is a function that solves the equation
   \begin{equation}
     \label{eq:272}
     \ppd{u}{t}
     +
     \convolution{
       \left(
         \frac{1}{h}
         \InverseFourierTransform{\frac{\FourierTransformSemidiscrete{\vec{C}_h}}{\xi^2}}
       \right)
     }
     {
       \ppd{u}{x}
     }
     =
     0
   \end{equation}
   with the initial conditions~\eqref{eq:258} and \eqref{eq:259}. Subtracting the evolution equations for $u(t, x)$ and $u(t, x+2\pi)$ we get an evolution equation for the difference $u(t, x) - u(t, x+2\pi)$. The initial conditions for the difference are \emph{zero} initial conditions since the original initial conditions~\eqref{eq:258} and \eqref{eq:259} are $2 \pi$ periodic functions. (They are bandwidth limited interpolants with integer wavenumbers.) Consequently, the difference $u(t, x) - u(t, x+2\pi)$ is equal to zero thorough the evolution, meaning that the function $u$ that solves~\eqref{eq:272} is a $2\pi$ periodic function. Since $u$ is a $2\pi$ periodic function it must have a Fourier series expansion of type $u = \sum_{l = - \infty}^{+ \infty} a_m \exponential{\iunit l x}$, which means that its continuous Fourier transform contains only \emph{discrete} wavenumbers $l = -\infty, \dots, + \infty$. Next we observe that the definition of continuous/semidiscrete Fourier transform, see~\eqref{eq:fourier-transform} and~\eqref{eq:semidiscrete-fourier-transform}, imply that
  \begin{equation}
    \label{eq:273}
    \FourierTransformSemidiscrete{\vec{C}_h}
    =
    h \FourierTransform{\sum_{n = 1}^{N} C_{h, n} \diracdelta(x - x_{h, n})}
    \chi_{\xi \in \left[ -\frac{\pi}{h}, \frac{\pi}{h} \right]}
    .
  \end{equation}
  (Compare with the same step in Theorem~\ref{thr:2}, formula~\eqref{eq:155}. In the current setting we work with the finite sum $\sum_{j=1}^N$ instead of the infinite sum $\sum_{n=-\infty}^{+ \infty}$.) The rest of the proof follows the proof of Theorem~\eqref{thr:2}. Since we have~\eqref{eq:273}, we conclude that $u$ has wavenumbers restricted to the interval $\left[ - \frac{\pi}{h}, \frac{\pi}{h} \right]$ which in virtue of~\eqref{eq:174} translates to $\left[ - \frac{N}{2}, \frac{N}{2} \right]$. The image of $u$ in the Fourier space thus contains only discrete wavenumbers $l = -\frac{N}{2}, \dots, \frac{N}{2}$, which means that the solution $u$ is in fact a bandwidth limited interpolant of a periodic function in the sense of definition~\eqref{eq:bandwidth-limited-intrepolant-periodic}. We now denote $u$ as $u_h$ in order to indicate that we are in fact working with a bandwidth limited function. Using a calculation analogous to~\eqref{eq:158} we find that
  \begin{equation}
    \label{eq:274}
        \convolution{
      \left(
        \frac{1}{h}
        \InverseFourierTransform{\frac{\FourierTransformSemidiscrete{\vec{C}_h}}{\xi^2}}
      \right)
    }
    {
      \ppd{u_h}{x}
    }
    =
    -
    \sum_{n = 1}^{N} C_{h, n} u_h(x - x_{h, n})
    .
  \end{equation}
Consequently, equation~\eqref{eq:272} reduces to
  \begin{equation}
    \label{eq:275}
    \ppd{u_h}{t}(x,t)
    -
    \sum_{n = 1}^{N} C_{h, n} u_h(x - x_{h, n}, t)
    =
    0,
  \end{equation}
  which upon sampling at $x_{h, j}$ and due to the periodicity of $u_h$ yields the system of equations
  \begin{equation}
    \label{eq:276}
    \ppd{u_{h, j}}{t}
    -
    \sum_{n = 1}^{N} \widetilde{c}_{h, n} u_{h, j - n}
    =
    0.   
  \end{equation}
  (Recall that on the equispaced grid we have $x_{h, j} - x_{h, n} = x_{h, j-n}$ and that the bandwidth limited function periodic function~$u_h$ is in one-to-one correspondence with its grid values at the grid $\left\{ x_{h,j}\right\}_{j=1}^{N}$. Furthermore, the coefficients $C_{h, m}$ are related to coefficients $\widetilde{c}_{h, m}$ via~\eqref{eq:263}.) System~\eqref{eq:276} is the same as system~\eqref{eq:252}, it suffices to relabel the summation index and use the periodicity.
  
 \end{proof}

\section{Finite lattice model with fixed ends/zero Dirichlet boundary conditions}
\label{sec:finite-lattice-with}

The lattice (chain) model now describes $M$ equal mass particles in the spatial domain $[0, \pi]$ that are in the equilibrium positioned at equispaced grid $\{x_{h, j}\}_{j=1}^{M}$,
\begin{subequations}
  \label{eq:277}
\begin{equation}
  \label{eq:278}
  x_{h, j} = jh,
\end{equation}
where
\begin{equation}
  \label{eq:279}
  h = \frac{\pi}{M+1},
\end{equation}
\end{subequations}
see~Figure~\ref{fig:dirichlet-lattice}. The (unknown) longitudinal displacements of the individual particles are denoted as
\begin{equation}
  \label{eq:280}
   \{u_{h, j} (t)\}_{j=1}^{M}.
\end{equation}
The lattice (chain) is fixed at the left end, $x_{h, 0} =_{\bydefinition} 0 $, and at the right end, $x_{h, M+1} =_\bydefinition \pi$, meaning that the displacement of the leftmost and the rightmost \emph{virtual} particles is known,
\begin{subequations}
  \label{eq:281}
  \begin{align}
    \label{eq:282}
    u_{h, 0} (t) =_{\bydefinition} 0, \\    
    \label{eq:283}
    u_{h, M+1} (t) =_{\bydefinition} 0.
  \end{align}
\end{subequations}
Later on we shall work with an odd extension of the lattice (chain) along the right endpoint $x_{h, M+1}$ followed by the periodic extension into the whole real line, see Figure~\ref{fig:dirichlet-lattice}. This gives us a periodic lattice (chain) with the fundamental period $2\pi$ and $N = 2(M + 1)$ particles, which is the setting similar to that in Section~\eqref{sec:periodic-lattice}. Consequently, the analysis of the fixed ends finite lattice setting might build on already obtained results for the periodic lattice.

\subsection{Lattice model of nearest neighbour interacting particles with fixed ends/zero Dirichlet boundary conditions}
\label{sec:latt-model-inter}

Since we are now dealing with a finite length lattice with fixed ends the concept of multiple-neighbours interaction becomes a peculiar one. If a particle is somewhere in the middle of the lattice it would typically have enough neighbours to the left and to the right. This is however not true if we are interested in a particle close to the lattice end. For example, the leftmost particle $u_{h, 1}$ can interact with the nearest neighbour to the left---the virtual wall particle $u_{h, 0}$, see~\eqref{eq:282}---but the interaction can not reach further left. At the moment we thus restrict ourselves to the nearest neighbour interactions, and we shall speculate on possible multiple-neighbours interactions later.

The governing equations for the particles displacements $\{u_{h, j} (t)\}_{j=1}^{M}$ now read
\begin{equation}
  \label{eq:284}
  \ddd{u_{h,j}}{t}
  -
  \ccontinuous^2
  \left(
    \frac{u_{h, j+1} - 2u_{h, j} + u_{h, j-1}}{h^2}
  \right)
  =
  0
  ,
\end{equation}
with the convention that whenever we encounter the virtual wall particle $u_{h, 0}$ or $u_{h, M+1}$, then we use~\eqref{eq:281}. (Recall that we denote $\ccontinuous^2 = \frac{K}{\rho}$.) The system of ordinary differential equations can be also rewritten in the matrix--vector form as
  \begin{equation}
    \label{eq:285}
    \ddd{\vec{u}_h}{t} - \tensorq{L}_h^{\text{Dirichlet}} \vec{u}_h = \vec{0}
  \end{equation}
  with the obvious identification of vector $\vec{u}_h$ and matrix $\tensorq{L}_h^{\text{Dirichlet}}$. \emph{The matrix in \eqref{eq:285} is a banded symmetric matrix}. For example, if we consider the six particles lattice with the nearest neighbour interactions, then we get the following system of ordinary differential equations,
  \begin{equation}
    \label{eq:286}
        \ddd{}{t}
    \begin{bmatrix}
      u_{h, 1} \\
      u_{h, 2} \\
      u_{h, 3} \\
      u_{h, 4} \\
      u_{h, 5} \\
      u_{h, 6} \\
    \end{bmatrix}
    -
    \frac{\ccontinuous^2}{h^2}
    \begin{bmatrix}
      -2 & 1 &  &  &  &  \\
      1 & -2 & 1 &  &  &  \\
       & 1 & -2 & 1 &  &  \\
       &  & 1 & -2 & 1 &  \\
       &  &  & 1 & -2 & 1 \\
       &  &  &  & 1 & -2 \\
    \end{bmatrix}
    \begin{bmatrix}
      u_{h, 1} \\
      u_{h, 2} \\
      u_{h, 3} \\
      u_{h, 4} \\
      u_{h, 5} \\
      u_{h, 6} \\
    \end{bmatrix}
    =
    \begin{bmatrix}
      0 \\
      0 \\
      0 \\
      0 \\
      0 \\
      0
    \end{bmatrix}
    .
  \end{equation}
  The banded matrix in~\eqref{eq:286} is the well-known symmetric tridiagonal matrix that arises in the numerical analysis by the discretisation of the second order derivative operator (with zero Dirichlet boundary conditions) by the centred second order finite difference scheme.

  We also note that the matrix in~\eqref{eq:286} is almost the same as the matrix~\eqref{eq:186} encountered in the periodic case. The difference between the periodic lattice and the fixed ends lattice is in the top-right corner and the bottom-left corner of the system matrix---the system matrix in the periodic case is a \emph{circulant matrix}, while the system matrix in the fixed ends case is a \emph{banded matrix}. The loss of \emph{circulant matrix} structure might be a problem in analysis, since our analysis heavily relied on the convolution structure of the governing equations and on convolution-to-multiplication property of Fourier transform. But as we shall see,  the Fourier transform based toolchain resurfaces even in the fixed ends/zero Dirichlet boundary conditions case.

\subsection{Reconstruction procedure---from discrete grid values to a function of real variable with zero Dirichlet boundary conditions}
\label{sec:reconstr-proc-from-1}
We must first propose a reconstruction procedure that would allow us to reconstruct a real valued function $u_h(t, x)$ out of discrete grid values $ \{u_{h, j} (t)\}_{j=1}^{M}$ while preserving the zero Dirichlet boundary conditions~\eqref{eq:281} at the continuous level, that is
\begin{equation}
  \label{eq:287}
  \left. u_h(t, x) \right|_{x=0, \pi} = 0.
\end{equation}
Since we already know that the bandwidth limited interpolant is a good choice in the periodic case, we would like to exploit the same methodology also in the zero Dirichlet boundary conditions case. This brings us to the realm of discrete sine and cosine transforms, see, for example, \cite{martucci.sa:symmetric}, \cite{strang.g:discrete} and~\cite{britanak.v.yip.pc.ea:discrete}.

\subsubsection{Discrete sine transform and bandwidth limited interpolant}
\label{sec:discr-sine-transf}
We start with $M$ grid values $\{u_{h, j} (t)\}_{j=1}^{M}$ in the interval~$[0, \pi]$, and we extend them in such a way that we can exploit the setting for the periodic lattice, see Section~\eqref{sec:periodic-lattice}. The extension is constructed as follows, see Figure~\ref{fig:dirichlet-lattice}. First, we define the grid value at the right endpoint of $[0, \pi]$ using the desired zero Dirichlet boundary condition~\eqref{eq:283}, that is we set $u_{h, M+1} = 0$. Second, we extend the available grid values to the interval $[\pi, 2\pi]$. In particular, we do \emph{odd} extension of grid values along the point $x_{h, M+1} = \pi$, that is for $j=1, \dots, M$ we set
\begin{equation}
  \label{eq:odd-extension}
  u_{h, (M+1)+j} =_{\bydefinition}- u_{h, (M+1)-j},
\end{equation}
and we interpret $u_{h, (M+1)+j}$ as the grid value at point $x_{h, (M+1)+j} = \left((M+1)+j \right) h$. Finally, we define the grid value at the right endpoint of $[0, 2\pi]$ to be zero. Doing so we obtain grid values~$\{u_{h, j} (t)\}_{j=1}^{N}$,
\begin{equation}
  \label{eq:288}
  N =_{\bydefinition} 2(M+1) 
\end{equation}
at equispaced grid points $\left\{x_{h, n} \right\}_{n=1}^N$, $x_{h, l} =lh$, in the interval $[0, 2\pi]$ with the interpoint distance
\begin{equation}
  \label{eq:289}
  h = \frac{\pi}{M+1} = \frac{2 \pi}{N}.  
\end{equation}
This is the setting we know from the periodic case, and the rest of the reconstruction procedure is the same as in the periodic case---we exploit the discrete Fourier transform, see Section~\ref{sec:reconstr-proc-from-2}.

  In virtue of the proposed reconstruction procedure, the bandwidth limited interpolant $u_h(t, x)$ of extended grid values clearly satisfies the conditions
  \begin{subequations}
    \begin{align}
      \label{eq:290}
      \left. u_h(t, x) \right|_{x = x_{h, M+1} = \pi} &= 0\\      
      \label{eq:291}
      \left. u_h(t, x) \right|_{x = 0, 2\pi} = 0,
    \end{align}
  \end{subequations}
hence the bandwidth limited interpolant satisfies the desired zero Dirichlet boundary condition~\eqref{eq:281}. Furthermore, the odd extension guarantees that the reconstructed function is odd with respect to the midpoint~$\pi$.

The whole procedure in fact leads to discrete sine transformation. The extension procedure applied to the grid values $\{u_{h, j} (t)\}_{j=1}^{M}$ gives us
\begin{equation}
  \label{eq:292}
  \vec{u}_h
  =
  \begin{bmatrix}
    u_{h, 1} \\
    u_{h, 2} \\
    u_{h, 3} \\
    \vdots \\
    u_{h, M}
  \end{bmatrix}
  \mapsto
  \vec{u}_h^{\text{odd}}
  =
  \begin{bmatrix}
    u_{h, 1} \\
    u_{h, 2} \\
    u_{h, 3} \\
    \vdots \\
    u_{h, M} \\
    0 \\
    -u_{h, M} \\
    \vdots \\
    -u_{h, 3} \\
    -u_{h, 2} \\
    -u_{h, 1} \\
    0
  \end{bmatrix}
  ,
\end{equation}
where $\vec{u}_h^{\text{odd}}$ is a vector of length $N = 2 (M+1)$. The extended vector $\vec{u}_h^{\text{odd}}$ is then processed by the standard discrete Fourier transform/discrete inverse Fourier transform~\eqref{eq:discrete-fourier-transform}. If we evaluate the discrete Fourier transform of the extended vector~$\vec{u}_h^{\text{odd}}$, then we get
\begin{multline}
  \label{eq:293}
      \tensor{\left( \FourierTransformDiscrete{\vec{u}_h^{\text{odd}}} \right)}{_k}
       =
       h
       \sum_{j= 1}^{N}
       \tensor{\left(\vec{u}_h^{\text{odd}}\right)}{_j}
       \exponential{- \iunit k x_{h, j}}
       =
       h
       \sum_{j= 1}^{M}
       \tensor{\left(\vec{u}_h^{\text{odd}}\right)}{_j}
       \exponential{- \iunit k x_{h, j}}
       +
       h
       \sum_{n= 1}^{M}
       \tensor{\left(\vec{u}_h^{\text{odd}}\right)}{_{(M+1) + n}}
       \exponential{- \iunit k x_{h, (M+1) + n}}
       \\
       =
       h
       \sum_{j= 1}^{M}
       \tensor{\left(\vec{u}_h^{\text{odd}}\right)}{_j}
       \exponential{- \iunit k x_{h, j}}
       -
       h
       \sum_{n= 1}^{M}
       \tensor{\left(\vec{u}_h^{\text{odd}}\right)}{_{(M+1) - n}}
       \exponential{- \iunit k x_{h, (M+1) + n}}
       =
       h
       \sum_{j= 1}^{M}
       \tensor{\left(\vec{u}_h^{\text{odd}}\right)}{_j}
       \exponential{- \iunit k x_{h, j}}
       -
       h
       \sum_{l= M}^{1}
       \tensor{\left(\vec{u}_h^{\text{odd}}\right)}{_{l}}
       \exponential{- \iunit k x_{h, 2(M+1) - l}}
       \\
       =
       -
       h
       \sum_{j= 1}^{M}
       \tensor{\left(\vec{u}_h\right)}{_j}
       \left( \exponential{\iunit k x_{h, j}} - \exponential{ -\iunit k x_{h, j}} \right)
       =
       -
       2 \iunit h
       \sum_{j= 1}^{M}
       \tensor{\left(\vec{u}_h\right)}{_j}
       \sin \left(  k x_{h, j} \right)
       =
       -
       2 \iunit h
       \sum_{j= 1}^{M}
       \tensor{\left(\vec{u}_h\right)}{_j}
       \sin \left(  \frac{\pi}{M+1} k j \right)
       ,
     \end{multline}
     where we have exploited the odd symmetry~\eqref{eq:odd-extension} and the fact that $\exponential{- \iunit k x_{h, 2(M+1) - l}} = \exponential{- \iunit k \left( 2\pi - x_{h, l} \right)} = \exponential{\iunit k x_{h, l}}$. This manipulation clearly shows that the extension--discrete Fourier transform is, up to a scaling factor, equivalent to the discrete sine transform of type I, DST-I, see, for example, \cite[Section 2.7]{britanak.v.yip.pc.ea:discrete}. 

     Four our purposes we introduce the \emph{discrete sine transformation} as a transformation that transforms $M$ grid values $\{u_{h, j} (t)\}_{j=1}^{M}$ to $M$ values $\left\{ \tensor{\left(\SineTransformDiscrete{\vec{u}_n}\right)}{_k} \right\}_{k=1}^M$ in the Fourier space as 
     \begin{equation}
       \label{eq:discrete-sine-transform}
       \tensor{\left(\SineTransformDiscrete{\vec{u}_n}\right)}{_k} =_{\bydefinition} - \frac{1}{2} {\ensuremath{\Im}} \left( \tensor{\left( \FourierTransformDiscrete{\vec{u}_h^{\text{odd}}} \right)}{_k} \right), \qquad k=1, \dots, M,
     \end{equation}
     where $\Im$ denotes the imaginary part of the corresponding expression. We thus have the following definition.
     \begin{definition}[Discrete sine transform]
       \label{dfn:9}
       Let $\vec{u}_h = \{u_{h, j} (t)\}_{j=1}^{M}$ be a collection of grid values at the grid $\left\{x_{h, j} \right\}_{j=1}^M$, $x_{h, j}=jh$, $h = \frac{\pi}{M+1}$, in the interval~$[0, \pi]$. Let $\vec{u}_h^{\text{odd}}$ be the odd extension of the collection $\vec{u}_h$, see~\eqref{eq:292}, to the grid $\left\{x_{h, l} \right\}_{l=1}^N$, $x_{h, l}=lh$, $h = \frac{2\pi}{N}$, $N = 2(M+1)$. The collection $\SineTransformDiscrete{\vec{u}_n}$ whose components in the Fourier space are given as
       \begin{equation}
         \label{eq:442}
         \tensor{\left(\SineTransformDiscrete{\vec{u}_n}\right)}{_k} =_{\bydefinition} - \frac{1}{2} {\ensuremath{\Im}} \left( \tensor{\left( \FourierTransformDiscrete{\vec{u}_h^{\text{odd}}} \right)}{_k} \right), \qquad k=1, \dots, M,
       \end{equation}
       where ${\ensuremath{\Im}}$ denotes the imaginary part, is called the \emph{discrete sine transform} of collection $\vec{u}_h$.
     \end{definition}
     An analogue of the explicit formula for the bandwidth limited interpolant in the physical space, that is an analogue of~\eqref{eq:bandwidth-limited-interpolant-semidiscrete-physical-space} and~\eqref{eq:bandwidth-limited-interpolant-periodic-physical-space}, is obtained by the same manipulation as in Section~\ref{sec:bandw-limit-interp} and Section~\ref{sec:discr-four-transf}. We start with the $m$-th discrete Dirac distribution $\vec{\diracdelta}_{h, \mathrm{Dirichlet}}^{h, m}$, 
\begin{equation}
  \label{eq:294}
  \vec{\diracdelta}_{h, \mathrm{Dirichlet}}^{h, m}
  =_{\bydefinition}
  \begin{bmatrix}
    0 \\
    \vdots \\
    0 \\
    1 \\
    0 \\
    \vdots \\
    0
  \end{bmatrix}
  ,  
\end{equation}
where the nonzero element is placed at the $m$-th position in this vector of length $M$. (The $m$-th discrete Dirac distribution~$\vec{\diracdelta}_{h, \mathrm{Dirichlet}}^{h, m}$ is defined on the space of grid values $x_{h, j}$, $j=1, \dots, M$. The $m$-the position in the vector~\eqref{eq:294} corresponds to the grid value at the point $x_{h, m}$.) Following the construction based on odd periodic extension, see~\eqref{eq:292}, we extend this vector to the vector $\vec{\diracdelta}_{h, \mathrm{Dirichlet}}^{h, m, \text{odd}}$ of length $2(M+1)$. The vector $\vec{\diracdelta}_{h, \mathrm{Dirichlet}}^{h, m, \text{odd}}$ is zero everywhere except at its $m$-th position and $(2(M+1)-m)$-th position, and we can write it as a sum of two   denotes the discrete Dirac functions $\vec{\diracdelta}_h^{h, l}$ we have used in the periodic case, see~\eqref{eq:202},
\begin{equation}
  \label{eq:296}
  \vec{\diracdelta}_{h, \mathrm{Dirichlet}}^{h, m, \text{odd}}
  =
  \vec{\diracdelta}_h^{h, m}
  -
  \vec{\diracdelta}_h^{h, 2(M+1)-m},
\end{equation}
The odd extension procedure brings us to the periodic setting, hence we can from now on follow the corresponding discussion on the periodic case, see Section~\ref{sec:discr-four-transf}, especially formula~\eqref{eq:204}. The forward discrete Fourier transform~\eqref{eq:188} of $\vec{\diracdelta}_{h, \mathrm{Dirichlet}}^{h, m, \text{odd}}$ yields
\begin{equation}
  \label{eq:297}
  \tensor{
    \left(
      \FourierTransformDiscrete{\vec{\diracdelta}_{h, \mathrm{Dirichlet}}^{h, m, \text{odd}}}
    \right)
  }{_k}
  =
  h
  \exponential{- \iunit k x_{h, m}}
  -
  h
  \exponential{- \iunit k x_{h, 2(M+1) - m}}
  =
  h
  \exponential{- \iunit k x_{h, m}}
  -
  h
  \exponential{\iunit k x_{h, m}},
\end{equation}
and the inverse \emph{continuous} Fourier transform of~\eqref{eq:297} then gives us the sought explicit expression for the bandwidth limited interpolant $\diracdelta_{h, \mathrm{Dirichlet}}^{h, m}$ of the Discrete Dirac function for fixed ends lattice $\vec{\diracdelta}_{h, \mathrm{Dirichlet}}^{h, m}$,
\begin{equation}
  \label{eq:298}
  \diracdelta_{h, \mathrm{Dirichlet}}^{h, m}
  =
  \InverseFourierTransform{
    \FourierTransformDiscrete{\vec{\diracdelta}_{h, \mathrm{Dirichlet}}^{h, m, \text{odd}}}
  }
  =
  \Sdiracdelta_h^h \left( x - x_{h, m}\right)
  -
  \Sdiracdelta_h^h \left( x + x_{h, m}\right)
  ,
\end{equation}
where we have reused the calculation done in~\eqref{eq:204}. We can thus conclude that the bandwidth limited interpolant $u_h$ of grid values $\vec{u}_h$ is in the fixed ends/zero Dirichlet boundary conditions case given as
  \begin{subequations}
    \label{eq:bandwidth-limited-interpolant-dirichlet-physical-space}
    \begin{align}
      \label{eq:299}
      u_h(x) &= \sum_{j=1}^{M}\tensor{\left(\vec{u}_h\right)}{_j} \left( \Sdiracdelta_h^h \left( x - x_{h, j}\right) - \Sdiracdelta_h^h \left( x + x_{h, j}\right) \right), \\
      \label{eq:300}
      \Sdiracdelta_h^h \left(x \right)
             &=_{\bydefinition}
      \frac{h}{2\pi}
      \cos \left( \frac{x}{2} \right)
      \frac{\sin \left( \frac{\pi}{h} x \right)}{\sin \left( \frac{x}{2} \right)}
      .
    \end{align}
  \end{subequations}  
  This is the fixed ends/zero Dirichlet boundary conditions case analogue to the infinite domain formula~\eqref{eq:bandwidth-limited-interpolant-semidiscrete-physical-space} and the periodic domain formula~\eqref{eq:bandwidth-limited-interpolant-periodic-physical-space}.
  We can summarise our findings as a lemma.
  \begin{lemma}[Bandwidth limited interpolant, finite lattice, fixed ends/zero Dirichlet boundary conditions]
    \label{lm:10}
    Let $\vec{u}_h = \{u_{h, j} (t)\}_{j=1}^{M}$ be a collection of grid values at the grid $\left\{x_{h, j} \right\}_{j=1}^M$, $x_{h, j}=jh$, $h = \frac{\pi}{M+1}$, in the interval~$[0, \pi]$. Let $u_h$ be the bandwidth limited interpolant constructed out of the odd extension $\vec{u}_h^{\text{odd}}$ of the collection $\vec{u}_h$, see~\eqref{eq:292}, to the grid $\left\{x_{h, l} \right\}_{l=1}^N$, $x_{h, l}=lh$, $h = \frac{2\pi}{N}$, $N = 2(M+1)$, that is
    \begin{equation}
      \label{eq:445}
      u_h =_{\bydefinition}
      \InverseFourierTransform{
    \sideset{}{^\prime}\sum_{k = -\frac{N}{2}}^{\frac{N}{2}}
    \tensor{\left( \FourierTransformDiscrete{\vec{u}_h^{\text{odd}}} \right)}{_k}
    \diracdelta_{k}
  }.
\end{equation}
(Concerning the bandwidth limited interpolant construction see Section~\ref{sec:discr-four-transf}.) Then
    \begin{align}
      \label{eq:443}
      u_h(x) &= \sum_{j=1}^{M}\tensor{\left(\vec{u}_h\right)}{_j} \left( \Sdiracdelta_h^h \left( x - x_{h, j}\right) - \Sdiracdelta_h^h \left( x + x_{h, j}\right) \right), \\
      \label{eq:444}
      \Sdiracdelta_h^h \left(x \right)
             &=_{\bydefinition}
               \frac{h}{2\pi}
               \cos \left( \frac{x}{2} \right)
               \frac{\sin \left( \frac{\pi}{h} x \right)}{\sin \left( \frac{x}{2} \right)}
               .
    \end{align}
  \end{lemma}
  
\subsubsection{Discrete sine transform and evaluation of bandwidth limited interpolant at arbitrary point}
\label{sec:discr-sine-transf-2}

As in the periodic case, the bandwidth limited interpolant $u_h$ can be evaluated at any point $y \in [0, \pi]$ in the physical space \emph{without the need to sum the series representation} \eqref{eq:bandwidth-limited-interpolant-dirichlet-physical-space}. This can be done using the same manipulation as in Section~\eqref{sec:discr-four-transf-4}, but one needs to work with the extended grid values vector, see~\eqref{eq:292}. 

     \subsubsection{Discrete sine transform and symmetric tridiagonal matrices}
     \label{sec:discr-sine-transf-1}
     The matrix
     \begin{equation}
       \label{eq:301}
       \tensorq{L}_h^{\text{Dirichlet}}
       =_{\bydefinition}
       \begin{bmatrix}
         -2 & 1 &  &  &  &  \\
         1 & -2 & 1 &  &  &  \\
          & 1 & -2 & 1 &  &  \\
          &  & 1 & -2 & 1 &  \\
          &  &  & 1 & -2 & 1 \\
          &  &  &  & 1 & -2 \\
       \end{bmatrix}
       ,
     \end{equation}
     in~\eqref{eq:286} is not a circulant matrix. It thus seems that the convenient convolution-to-multiplication property we know from the continuous/semidiscrete/discrete Fourier transform is of no use in the case of fixed ends/Dirichlet lattice. In particular, if we are interested in the eigenvalues of the differentiation matrix~\eqref{eq:301}, it seems that we can non make use of the fact that the discrete Fourier transform diagonalises circulant matrices. However, it turns out that the circulant matrices are still present and can be used in solving the eigenproblem for the differentiation matrix~\eqref{eq:301}.

     We define \emph{extension}~$\tensorq{E}_{\left(2M + 2\right) \times M}$ and \emph{reduction}~$\tensorq{R}_{M \times \left(2M + 2\right)}$ matrices/operators as
     \begin{subequations}
       \label{eq:302}
       \begin{align}
         \label{eq:303}
         \tensorq{E}_{\left(2M + 2\right) \times M}
         =_{\bydefinition}
         \begin{bmatrix}
           \identity_{M \times M} \\
           \transpose{\vec{0}_{M}} \\
           - \jmatrix_{M \times M} \\
           \transpose{\vec{0}_{M}} \\
         \end{bmatrix}
         ,
         \\
         \label{eq:304}
         \tensorq{R}_{M \times \left(2M + 2\right)}
         =_{\bydefinition}
         \frac{1}{2}
         \begin{bmatrix}
           \identity_{M \times M} & \vec{0}_{M} & - \jmatrix_{M \times M} & \vec{0}_{M}
         \end{bmatrix}
         ,
       \end{align}
     \end{subequations}
     where the symbol $\vec{0}_{M}$ denotes the zero column vector of length $M$, and where $\identity_{M \times M}$ denotes the $M \times M$ identity matrix and $\jmatrix_{M \times M}$ denotes the $M \times M$ matrix with the secondary diagonal filled with ones and with zero elements otherwise,
     \begin{subequations}
       \label{eq:305}
       \begin{align}
         \label{eq:306}
         \identity_{M \times M} &=_{\bydefinition}
                                  \begin{bmatrix}
                                    1 &  &  &  &  \\
                                     & 1 &  &  &  \\
                                     &  & 1 &  &  \\
                                     &  &  & \ddots &  \\
                                     &  &  &  & 1
                                  \end{bmatrix}
                                  , \\
         \label{eq:307}
         \jmatrix_{M \times M} &=_{\bydefinition}
                                  \begin{bmatrix}
                                     &  &  &  & 1 \\
                                     &  &  & 1 &  \\
                                     &  & 1 &  &  \\
                                     & \reflectbox{$\ddots$} &  &  &  \\
                                    1 &  &  &  & 
                                  \end{bmatrix}
       \end{align}
     \end{subequations}
     Using these matrices we can represent the extension transform~\eqref{eq:292} as
     \begin{subequations}
       \label{eq:308}
       \begin{align}
         \label{eq:309}
         \vec{u}_h^{\text{odd}} = \tensorq{E}_{\left(2M + 2\right) \times M}  \vec{u}_h, \\
         \label{eq:310}
         \vec{u}_h = \tensorq{R}_{M \times \left(2M + 2\right)} \vec{u}_h^{\text{odd}}.
       \end{align}
     \end{subequations}
     Note that
     \begin{equation}
       \label{eq:311}
       \tensorq{R}_{M \times \left(2M + 2\right)} \tensorq{E}_{\left(2M + 2\right) \times M} = \identity_{M \times M},  
     \end{equation}
     and that $\tensorq{R}_{M \times \left(2M + 2\right)} = \frac{1}{2} \transpose{\left(\tensorq{E}_{\left(2M + 2\right) \times M}\right)}$. Furthermore, if we restrict ourselves to the subspace of vectors with the structure~$\vec{u}_h^{\text{odd}}$, then we even have the inverse relation between the extension and reduction matrices, 
     \begin{equation}
       \label{eq:312}
       \tensorq{E}_{\left(2M + 2\right) \times M} \tensorq{R}_{M \times \left(2M + 2\right)} = \identity_{\left(2M + 2\right) \times \left(2M + 2\right)}.
     \end{equation}
     Finally, we note that if the elements of collection $\FourierTransformDiscrete{\vec{u}_h^{\text{odd}}}$ are arranged accordingly into a column vector, then we get the matrix representation of discrete sine transform~\eqref{eq:discrete-sine-transform} in the form
     \begin{equation}
       \label{eq:313}
       \SineTransformDiscreteMatrix = \frac{\iunit}{2}   \tensorq{R}_{M \times \left(2M + 2\right)}    \FourierTransformDiscreteMatrix \tensorq{E}_{\left(2M + 2\right) \times M}.
     \end{equation}
     Here we assume that $\FourierTransformDiscreteMatrix$ is the matrix representing the discrete Fourier transform that takes the column vector of grid values arranged as
     $
     \transpose{
       \begin{bmatrix}
         v_{h, 1} & \cdots & v_{h, 2(M+1)}
       \end{bmatrix}
     }
     $
     and returns the collection $\FourierTransformDiscrete{\vec{v}_h}$ arranged in a column vector such that the elements $\tensor{\left( \FourierTransformDiscrete{\vec{u}_h^{\text{odd}}} \right)}{_1}$, \dots , $\tensor{\left( \FourierTransformDiscrete{\vec{u}_h^{\text{odd}}} \right)}{_{M+1}}$ are listed first, then they are followed by the elements $\tensor{\left( \FourierTransformDiscrete{\vec{u}_h^{\text{odd}}} \right)}{_{- \left(M+1 \right) +1}}$, \dots, $\tensor{\left( \FourierTransformDiscrete{\vec{u}_h^{\text{odd}}} \right)}{_{-1}}$, and the element corresponding to the zero wavenumber, that is the element $\tensor{\left( \FourierTransformDiscrete{\vec{u}_h^{\text{odd}}} \right)}{_{0}}$,  is listed as the \emph{last} component of this vector. (This might not be the typical arrangement in the computer codes for discrete Fourier transform.). Furthermore, the factor $h$ is present in our definition of discrete Fourier transform. We also note that in virtue of \eqref{eq:293} we have
     \begin{subequations}
       \label{eq:314}
       \begin{align}
         \label{eq:315}
         \tensor{\left( \FourierTransformDiscrete{\vec{u}_h^{\text{odd}}} \right)}{_k}
         &
           =
           \tensor{\left( \FourierTransformDiscrete{\vec{u}_h^{\text{odd}}} \right)}{_{-k}}, \\
         \label{eq:316}
         \tensor{\left( \FourierTransformDiscrete{\vec{u}_h^{\text{odd}}} \right)}{_{M+1}} & = 0, \\
         \label{eq:317}
         \tensor{\left( \FourierTransformDiscrete{\vec{u}_h^{\text{odd}}} \right)}{_0} & = 0,
       \end{align}
     \end{subequations}
     which means that the application of matrix $\FourierTransformDiscreteMatrix \tensorq{E}_{\left(2M + 2\right) \times M}$ generates the vector that has the same structure as the one implied by the extension operation~\eqref{eq:292}.

     Now we make the key observation---\emph{the application of a symmetric tridiagonal $M \times M$ matrix to a vector $\vec{u}_h$ can be rewritten as an application of a ``big'' circulant matrix to the extended vector~$\vec{u}_h^{\text{odd}}$}. Indeed, if $\tensorq{B}_{M \times M}$ denotes a symmetric tridiagonal matrix of size $M \times M$,
     \begin{subequations}
       \label{eq:318}
     \begin{equation}
       \label{eq:319}
       \tensorq{B}_{M \times M} =_{\bydefinition}
         \begin{bmatrix}
           c_{h, 0} & c_{h, 1}&  &  &  &  &  \\
           c_{h, 1}& c_{h, 0} & c_{h, 1}&  &  &  &  \\
            & c_{h, 1}& c_{h, 0} &  &  &  &  \\
            &  &  & \ddots &  &  & \\
            &  &  &  & c_{h, 0} & c_{h, 1}&  \\
            &  &  &  & c_{h, 1}& c_{h, 0} & c_{h, 1}\\
            &  &  &  &  & c_{h, 1}& c_{h, 0} \\
         \end{bmatrix}_{M \times M}
         ,
       \end{equation}
       and if
       \begin{equation}
         \label{eq:320}
         \tensorq{C}_{\left(2M + 2\right) \times \left(2M + 2\right)} =_{\bydefinition}
         \begin{bmatrix}
           c_{h, 0}& c_{h, 1}&  &  &  &  & c_{h, 1}\\
           c_{h, 1}& c_{h, 0}& c_{h, 1}&  &  &  &  \\
             & c_{h, 1}& c_{h, 0}&  &  &  &  \\
             &  &  & \ddots &  &  & \\
             &  &  &  & c_{h, 0}& c_{h, 1}&  \\
             &  &  &  & c_{h, 1}& c_{h, 0}& c_{h, 1}\\
           c_{h, 1}&  &  &  &  & c_{h, 1}& c_{h, 0}\\
         \end{bmatrix}_{\left(2M + 2\right) \times \left(2M + 2\right)}
       \end{equation}
       denotes the ``big'' circulant matrix of size $\left(2M + 2\right) \times \left(2M + 2\right)$, then it is straightforward to check that
       \begin{equation}
         \label{eq:321}
         \tensorq{R}_{M \times \left(2M + 2\right)}\tensorq{C}_{\left(2M + 2\right) \times \left(2M + 2\right)} \tensorq{E}_{\left(2M + 2\right) \times M} \vec{u}_h = \tensorq{B}_{M \times M} \vec{u}_h.
       \end{equation}
     \end{subequations}

     Naturally, we are mainly interested in the key observation~\eqref{eq:318} if the coefficients take values
     \begin{equation}
       \label{eq:322}
       c_{h, 0} =_{\bydefinition} - \frac{2}{h^2}, \qquad c_{h, 1} =_{\bydefinition} \frac{1}{h^2}.
     \end{equation}
     In this case the matrix $\tensorq{B}_{M \times M}$ is the matrix representing the standard second order finite differences approximation of the second derivative operator with fixed ends/zero Dirichlet boundary conditions, $\tensorq{B}_{M \times M} \equiv \tensorq{L}_h^{\text{Dirichlet}}$. The approximation of the second derivative operator then reads
     \begin{equation}
       \label{eq:323}
           \begin{bmatrix}
             \left. \ddd{u}{x} \right|_{x = x_{h, 1}} \\
             \left. \ddd{u}{x} \right|_{x = x_{h, 2}} \\
             \left. \ddd{u}{x} \right|_{x = x_{h, 3}} \\
             \vdots \\
             \left. \ddd{u}{x} \right|_{x = x_{h, M-1}} \\
             \left. \ddd{u}{x} \right|_{x = x_{h, M}} \\
           \end{bmatrix}
           \approx
           \begin{bmatrix}
             c_{h, 0} & c_{h, 1} &  &  &  &    \\
             c_{h, 1} & c_{h, 0} & c_{h, 1} &  &  &    \\
              & c_{h, 1} & c_{h, 0} & c_{h, 1} &  &    \\
                                 &  & \ddots & \ddots & \ddots &  \\
                                 &  &  & c_{h, 1} & c_{h, 0} & c_{h, 1} \\
                                 &  &  &  & c_{h, 1} & c_{h, 0} \\
           \end{bmatrix}
           \begin{bmatrix}
             u_{h, 1} \\
             u_{h, 2} \\
             u_{h, 3} \\
             \vdots \\
             u_{h, M-1} \\
             u_{h, M} \\
           \end{bmatrix}
           .
         \end{equation}
       
       Since the matrix $\tensorq{C}_{\left(2M + 2\right) \times \left(2M + 2\right)}$ in the key identity~\eqref{eq:321} is a \emph{circulant matrix}, it is still possible to solve the eigenvalue problem for the symmetric tridiagonal matrix
      \begin{equation}
        \label{eq:324}
        \tensorq{B}_{M \times M} \vec{u}_h = \lambda \vec{u}_h
      \end{equation}
      by referring to the diagonalisation properties of discrete Fourier transform, see diagonalisation of circulant matrices in Section~\ref{sec:discr-four-transf-3}. %
      First we rewrite the eigenvalue problem~\eqref{eq:324} using the key identity~\eqref{eq:321}, and we get
      \begin{equation}
        \label{eq:338}
        \tensorq{R}_{M \times \left(2M + 2\right)}\tensorq{C}_{\left(2M + 2\right) \times \left(2M + 2\right)} \tensorq{E}_{\left(2M + 2\right) \times M} \vec{u}_h
        =
        \lambda \vec{u}_h.
      \end{equation}
      This equation can be manipulated as
      \begin{equation}
        \label{eq:339}
        \tensorq{R}_{M \times \left(2M + 2\right)}
        \inverse{\left(\FourierTransformDiscreteMatrix \right)}
        \left(
        \FourierTransformDiscreteMatrix
        \tensorq{C}_{\left(2M + 2\right) \times \left(2M + 2\right)}
        \inverse{\left( \FourierTransformDiscreteMatrix \right)}
        \right)
        \FourierTransformDiscreteMatrix
        \tensorq{E}_{\left(2M + 2\right) \times M} \vec{u}_h
        =
        \lambda \vec{u}_h.
      \end{equation}
      Using the inverse relations~\eqref{eq:311} and \eqref{eq:312} we further rewrite~\eqref{eq:339} as
      \begin{multline}
        \label{eq:340}
        \left[
          \frac{2}{\iunit}
          \tensorq{R}_{M \times \left(2M + 2\right)}
          \inverse{\left(\FourierTransformDiscreteMatrix \right)}
          \tensorq{E}_{\left(2M + 2\right) \times M}
        \right]
        \left[
          \tensorq{R}_{M \times \left(2M + 2\right)}
          \left(
            \FourierTransformDiscreteMatrix
            \tensorq{C}_{\left(2M + 2\right) \times \left(2M + 2\right)}
            \inverse{\left( \FourierTransformDiscreteMatrix \right)}
          \right)
          \tensorq{E}_{\left(2M + 2\right) \times M}
        \right]
        \\
        \left[
          \frac{\iunit}{2}
          \tensorq{R}_{M \times \left(2M + 2\right)}
          \FourierTransformDiscreteMatrix
          \tensorq{E}_{\left(2M + 2\right) \times M}
        \right]
        \vec{u}_h
        =
        \lambda \vec{u}_h.
      \end{multline}
      We make use of the matrix form of the discrete sine transform~\eqref{eq:313} and the diagonalisation property of the discrete Fourier transform~\eqref{eq:234}, and we further rewrite~\eqref{eq:340} as
      \begin{equation}
        \label{eq:341}
        \inverse{\left(\SineTransformDiscreteMatrix\right)}
        \left[
          \tensorq{R}_{M \times \left(2M + 2\right)}
          \left(
            \diag \left( \FourierTransformDiscreteMatrix \vec{c}_h \right)
          \right)
          \tensorq{E}_{\left(2M + 2\right) \times M}
        \right]
        \SineTransformDiscreteMatrix
        \vec{u}_h
        =
        \lambda \vec{u}_h.
      \end{equation}
      We note that $\diag \left( \FourierTransformDiscreteMatrix \vec{c}_h \right)$ has a specific structure, see Lemma~\eqref{lm:1}, and we observe that the successive application of extension/reduction operation to a properly structured diagonal matrix yields the top left $M \times M$ block of the diagonal matrix $\diag \left( \FourierTransformDiscreteMatrix \vec{c}_h \right)$. In particular, we have
      \begin{multline}
        \label{eq:342}
        \tensorq{R}_{M \times \left(2M + 2\right)}
        \left(
          \diag \left( \FourierTransformDiscreteMatrix \vec{c}_h \right)
        \right)
        \tensorq{E}_{\left(2M + 2\right) \times M}
        =
        \tensorq{R}_{M \times \left(2M + 2\right)}
        \begin{bmatrix}
          \lambda_{h, 1} & & & & & & & \\
                    & \ddots & & & & & & \\ 
                    & & \lambda_{h, M} & & & & & \\
                    & & & \lambda_{h, M+1} & & & & \\
                    & & & & \lambda_{h, M} & & & \\
                    & & & & & \ddots & & \\
                    & & & & & & \lambda_{h, 1} & \\
                    & & & & & & & \lambda_{h, 0} 
        \end{bmatrix}
        \tensorq{E}_{\left(2M + 2\right) \times M}
        \\
        =
        \begin{bmatrix}
          \lambda_{h, 1} & & & & & \\
                    & \lambda_{h, 2} & & & & \\
           & & \ddots & & & \\
                    & & & & \lambda_{h, M-1}& \\ 
                    & & & & & \lambda_{h, M} 
        \end{bmatrix}
        ,
      \end{multline}
      with
      \begin{equation}
        \label{eq:343}
        \lambda_{h, m} = h \left( c_{h, 0} + 2c_{h, 1} \cos \left( m h \right) \right), \qquad m=1, \dots, M.
      \end{equation}

      This observation together with~\eqref{eq:341} recovers a well known fact usually phrased as ``discrete sine transform diagonalises symmetric tridiagonal matrices'' and that ``columns of discrete sine transform matrix are eigenvectors of symmetric tridiagonal matrices''. The typical analysis of the eigenvalue problem~\eqref{eq:324} however proceeds differently, see, for example, \cite[Example~7.2.5]{meyer.c:matrix} and \cite[Chapter~9, page~111]{gantmacher.fr.krein.mg:oscillation}. The eigenvalues are typically obtained directly by solving the eigenvalue problem via conversion to a difference equation. Once the eigenvalue formula~\eqref{eq:343} is obtained, it is then noted that it is in fact the cosine transform of the matrix elements. Here we have argued differently---we have shown that that Fourier transform tools are natural for the eigenproblem analysis, and that they directly reflect the properties of the corresponding continuous eigenproblem.
      
      \subsubsection{Discrete sine transform and differentiation}
\label{sec:discr-sine-transf-3}
Concerning the differentiation of a function reconstructed from the grid values $\vec{u}_h = \{u_{h, j} (t)\}_{j=1}^{N}$, we again define the differentiation via the differentiation of the corresponding bandwidth limited interpolant $u_h$. The grid values of the derivative are thus defined as
\begin{equation}
  \label{eq:344}
  \tensor{\left( \dd{^n}{x^n} \vec{u}_h \right)}{_j} =_{\bydefinition} \left. \dd{^n}{x^n} u_h \right|_{x = x_{h, j}},
\end{equation}
$j = 1, \dots, M$, where the bandwidth limited interpolant $u_h$ is the bandwidth limited interpolant constructed from the odd extension $\vec{u}_h^{\text{odd}}$ of grid values $\left\{u_{h, j} \right\}_{j=1}^m$, see Section~\ref{sec:discr-sine-transf}. The formula for the derivative is particularly interesting when $n$ is even. In this case the generic discrete Fourier transform formula~\eqref{eq:244} reduces, in its matrix form, as follows
\begin{equation}
  \label{eq:345}
  \dd{^n}{x^n} \vec{u}_h 
  =
  \left[
    \tensorq{R}_{M \times \left(2M + 2\right)}
    \inverse{\left( \FourierTransformDiscreteMatrix \right)}
  \right]
  \diag \left( (\iunit \vec{k})^n \right)
  \left[
    \vphantom{\inverse{\left( \FourierTransformDiscreteMatrix \right)}}
    \FourierTransformDiscreteMatrix
    \tensorq{E}_{\left(2M + 2\right) \times M}
    \vec{u}_h
  \right]
  ,
\end{equation}
where we have used the extension/reduction matrices and the matrix $\FourierTransformDiscreteMatrix$ representing the discrete Fourier transform. (For the discrete Fourier transform we use the same convention as in Section~\eqref{sec:discr-sine-transf-1}.) The matrix $\diag \left(\iunit \vec{k} \right)^n$ is the matrix with the main diagonal composed of properly ordered elements of the collection $\left(\iunit \vec{k} \right)^n$, see~\eqref{eq:244},
\begin{equation}
  \label{eq:346}
  \diag \left( \left(\iunit \vec{k} \right)^n \right)
  =
  \begin{bmatrix}
          (\iunit 1)^n & & & & & & & \\
                    & \ddots & & & & & & \\ 
                    & & (\iunit M)^n & & & & & \\
                    & & & \left(\iunit (M+1) \right)^n & & & & \\
                    & & & & (\iunit M)^n & & & \\
                    & & & & & \ddots & & \\
                    & & & & & & (\iunit 1)^n & \\
                    & & & & & & & (\iunit 0)^n     
  \end{bmatrix}.
\end{equation}
Using~\eqref{eq:346}, the inverse relations for extension/reduction matrices~\eqref{eq:311} and~\eqref{eq:312}, and the identity~\eqref{eq:342} for extension/reduction matrices and structured diagonal matrices, we see that~\eqref{eq:345} can be further manipulated as
\begin{multline}
  \label{eq:347}
    \dd{^n}{x^n} \vec{u}_h 
    =
    \left[
      \tensorq{R}_{M \times \left(2M + 2\right)}
      \inverse{\left( \FourierTransformDiscreteMatrix \right)}
    \right]
    \diag \left( (\iunit \vec{k})^n \right)
    \left[
      \vphantom{\inverse{\left( \FourierTransformDiscreteMatrix \right)}}
      \FourierTransformDiscreteMatrix
      \tensorq{E}_{\left(2M + 2\right) \times M}
      \vec{u}_h
    \right]
  \\
  =
  \left[
    \tensorq{R}_{M \times \left(2M + 2\right)}
    \inverse{\left( \FourierTransformDiscreteMatrix \right)}
    \tensorq{E}_{\left(2M + 2\right) \times M}
  \right]
  \left[
    \vphantom{\inverse{\left( \FourierTransformDiscreteMatrix \right)}}
    \tensorq{R}_{M \times \left(2M + 2\right)}
    \diag \left( (\iunit \vec{k})^n \right)
    \tensorq{E}_{\left(2M + 2\right) \times M}
  \right]
  \left[
    \vphantom{\inverse{\left( \FourierTransformDiscreteMatrix \right)}}
    \tensorq{R}_{M \times \left(2M + 2\right)}
    \FourierTransformDiscreteMatrix
    \tensorq{E}_{\left(2M + 2\right) \times M}
  \right]
  \vec{u}_h
    \\
    =
    \inverse{\left(\SineTransformDiscreteMatrix\right)}
    \begin{bmatrix}
      (\iunit 1)^n & &  \\
                   & \ddots &  \\
                   & & (\iunit M)^n \\
    \end{bmatrix}
    \SineTransformDiscreteMatrix
    \vec{u}_h
    .
  \end{multline}
  This shows that the differentiation of the bandwidth limited interpolant is straightforward even in the fixed ends/zero Dirichlet boundary conditions case. It suffices to take the discrete sine transform of grid values, multiply the resulting vector elements by fixed values and transform back to the physical space using the inverse discrete sine transform.

  The operator form of the Fourier based differentiation is summarised in the following definition. Note that the definition is just a variant of Definition~\eqref{dfn:8} applied to the odd extension and that the definition is of practical importance only for the even order differentiation.
  
  \begin{definition}[Fourier transform based differentiation, finite lattice, fixed ends/zero Dirichlet boundary conditions]
    \label{dfn:10}
    Let~$\vec{f}_h =  \{f_{h, j} \}_{j=1}^{M}$ be a vector of grid values at the grid $\{x_{h, j}\}_{j=1}^{M}$, $x_{h, j} = jh$, $h = \frac{\pi}{M+1}$, in the interval $[0, \pi]$, see~\eqref{eq:277}. Let $n$ be an even number. We set
    \begin{equation}
      \label{eq:446}
      \dd{^n}{x^n} \vec{f}_h =_{\bydefinition} \InverseSineTransformDiscrete{\left( \tensorschur{\iunit \vec{k} \right)^n}{\SineTransformDiscrete{\vec{f}_h}}},
    \end{equation}
    where the elements of the Fourier space collection $\left(\iunit \vec{k} \right)^n$, $m=1, \dots, M$, read
    \begin{equation}
      \label{eq:447}
      \tensor{
        \left(
          \left(\iunit \vec{k} \right)^n
        \right)
      }{_m}
      =_{\bydefinition}
      \left(\iunit m\right)^n.
  \end{equation}
\end{definition}

      \subsection{Correspondence between discrete and continuous models}
\label{sec:corr-betw-discr-2}

Now we are in a position to prove an analogue of Theorem~\eqref{thr:2} and Theorem~\ref{thr:3} on the equivalence between discrete and continuous models. In the fixed ends/zero Dirichlet boundary conditions case we are---at the moment---limited to the nearest neighbour interaction. The rationale for this restriction is discussed in the following Section~\ref{sec:beyond-near-neighb}

\subsubsection{Nearest neighbour interaction}
\label{sec:near-neighb-inter}
In the case of the nearest neighbour interaction it is straightforward to derive the following counterpart to Theorem~\ref{thr:2} and Theorem~\ref{thr:3}.

\begin{theorem}[Equivalence between discrete system of ordinary differential equations for grid values and the corresponding partial differential equation---nearest neighbour interaction, fixed ends/zero Dirichlet boundary condition]
  \label{thr:4}
  Let $M \in \N$ and let $\left\{x_{h, j}\right\}_{j=1}^{M}$, where $x_{h, j} =_{\bydefinition} jh$ with $h =_{\bydefinition} \frac{\pi}{M+1}$, be the corresponding grid on the interval $[0, \pi]$. Let~$\vec{u}_h(t) = \left\{ u_{h, j} (t)\right\}_{j=1}^{M}$ be a vector of grid values on the grid $\left\{x_{h, j} \right\}_{j=1}^{M}$, and let~$u_h(x,t)$ be the corresponding bandwidth limited interpolant of~$\vec{u}_h(t)$ obtained via the odd extension of grid values~$\vec{u}_h(t) = \left\{ u_{h, j} (t)\right\}_{j=1}^{M}$, that is
  \begin{subequations}
    \label{eq:348}
  \begin{equation}
    \label{eq:349}
    u_h(x)
    =
    \sum_{j=1}^{M}\tensor{\left(\vec{u}_h\right)}{_j} \left( \Sdiracdelta_h^h \left( x - x_{h, j}\right) - \Sdiracdelta_h^h \left( x + x_{h, j}\right) \right)
    ,
  \end{equation}
  where
  \begin{equation}
    \label{eq:350}
      \Sdiracdelta_h^h \left(x \right)
             =_{\bydefinition}
      \frac{h}{2\pi}
      \cos \left( \frac{x}{2} \right)
      \frac{\sin \left( \frac{\pi}{h} x \right)}{\sin \left( \frac{x}{2} \right)}.
  \end{equation}
  \end{subequations}
  Let the grid values $\vec{u}_h(t)  = \left\{ u_{h, j} (t)\right\}_{j=1}^{N}$ solve the initial value problem for the system of ordinary differential equations
  \begin{subequations}
    \label{eq:351}
    \begin{align}
      \label{eq:352}
      \ddd{u_{h,j}}{t}
      -
      \ccontinuous^2
      \left(
      \frac{u_{h, j+1} - 2u_{h, j} + u_{h, j-1}}{h^2}
      \right)
      &=
        0,
      \\
      \label{eq:353}
      \left. u_{h, j} \right|_{t=0} &= u_{h, j}^0,\\
      \label{eq:354}
      \left. \dd{u_{h, j}}{t} \right|_{t=0} &= v_{h, j}^0,
    \end{align}
  \end{subequations}
  where in~\eqref{eq:352} we set $u_{h, 0} =_{\bydefinition} 0$ and $u_{h, M+1} =_{\bydefinition} 0$. The matrix form of~\eqref{eq:352} thus reads
  \begin{equation}
    \label{eq:355}
    \ddd{}{t}
    \begin{bmatrix}
      u_{h, 1} \\
      u_{h, 2} \\
      u_{h, 3} \\
      \vdots \\
      u_{h, M-1} \\
      u_{h, M} \\
    \end{bmatrix}
    -
    \frac{\ccontinuous^2}{h^2}
    \begin{bmatrix}
      -2 & 1 &  &  &  &  \\
      1 & -2 & 1 &  &  &  \\
       & 1 & -2 & 1 &  &  \\
       &  &  & \ddots &  &  \\
       &  &  & 1 & -2 & 1 \\
       &  &  &  & 1 & -2 \\
    \end{bmatrix}
    \begin{bmatrix}
      u_{h, 1} \\
      u_{h, 2} \\
      u_{h, 3} \\
      \vdots \\
      u_{h, M-1} \\
      u_{h, M} \\
    \end{bmatrix}
    =
    \begin{bmatrix}
      0 \\
      0 \\
      0 \\
      \vdots \\
      0 \\
      0
    \end{bmatrix}
    .
  \end{equation}
  Let $u$ solve, for $x \in \R$ on the real line, the initial value problem 
  \begin{subequations}
    \label{eq:356}
    \begin{align}
      \label{eq:357}
    \ppd{u}{t}
    +
    \convolution{
      \left(
      \frac{1}{h}
      \InverseFourierTransform{\frac{\FourierTransformSemidiscrete{\vec{C}_h}}{\xi^2}}
      \right)
    }
    {
      \ppd{u}{x}
    }
    &=
    0
      ,
      \\
      \label{eq:358}
      \left. u \right|_{t=0} &= u_h^0,\\
      \label{eq:359}
      \left. \dd{u}{t} \right|_{t=0} &= v_h^0,
    \end{align}
  \end{subequations}
  where the symbol $\vec{C}_h$ denotes the infinite collection~$\left\{ C_{h, i}\right\}_{i=-\infty}^{+\infty}$ constructed from the vector $\widetilde{\vec{c}}_h$ of length $N = 2(M+1)$
  \begin{equation}
    \label{eq:360}
    \widetilde{\vec{c}}_h =
    \begin{bmatrix}
          c_{h, 0} \\ c_{h, 1} \\ c_{h, 2} \\  \vdots \\ c_{h, \frac{N}{2} - 2} \\ c_{h, \frac{N}{2} - 1} \\ c_{h, \frac{N}{2}} \\ c_{h, \frac{N}{2} -1} \\ \vdots \\ c_{h, 2} \\ c_{h, 1}
    \end{bmatrix}
    =_{\bydefinition}
    \frac{\ccontinuous^2}{h^2}
    \begin{bmatrix}
      -2 \\ 1 \\ 0 \\ \vdots \\ 0 \\ 0 \\ 0 \\ 0 \\ \vdots \\ 0 \\ 1 
    \end{bmatrix},
  \end{equation}
  via the procedure defined in~\eqref{eq:263}, and let the initial data $u_h^0$ and $v_h^0$ in~\eqref{eq:358} and \eqref{eq:359} be the bandwidth limited interpolants obtained by the odd extension of the initial data~\eqref{eq:353} and \eqref{eq:354}. The symbols $\FourierTransformSemidiscrete{\cdot}$ and $\InverseFourierTransform{\cdot}$ denote the \emph{semidiscrete} Fourier transform~\eqref{eq:69} and the \emph{continuous} inverse Fourier transform~\eqref{eq:20}. Then the function $u$ is the solution to the continuous problem~\eqref{eq:356} if and only if $u = u_h$, where $u_h$ is the bandwidth limited interpolant obtained by the odd extension of grid values $\vec{u}_h(t)  = \left\{ u_{h, j} (t)\right\}_{j=1}^{M}$ that solve the discrete problem~\eqref{eq:351}.
\end{theorem}

Since the solution $u$ to the problem~\eqref{eq:356} is the bandwidth limited interpolant $u_h$ obtained by the odd extension of grid values that solve~\eqref{eq:351}, we see that
\begin{equation}
  \label{eq:361}
  \left. u \right|_{x = l \pi} = 0,
\end{equation}
for all $l \in \Z$. (The same holds for bandwidth limited interpolants of the initial data.) In this sense $u$ represents a solution to the problem~\eqref{eq:351} in the interval $(0, \pi)$ with the ``boundary conditions'' $ \left. u  \right|_{x = 0, \pi} = 0$, which is the original domain of interest for the discrete problem.

\begin{proof}
  The matrix problem~\eqref{eq:355} can be rewritten in the form
  \begin{equation}
    \label{eq:362}
    \ddd{\vec{u}_h}{t}
    -
    \tensorq{B}_{M \times M}
    \vec{u}_h
    =
    \vec{0},
  \end{equation}
  where $\tensorq{B}_{M \times M}$ is the symmetric tridiagonal matrix introduced in~\eqref{eq:319} with the elements
  \begin{equation}
    \label{eq:363}
    c_{h, 0} =_{\bydefinition}   -2\frac{\ccontinuous^2}{h^2}, \qquad
    c_{h, 1} =_{\bydefinition}   \frac{\ccontinuous^2}{h^2}.
  \end{equation}
  We use the key observation~\eqref{eq:321} and we rewrite~\eqref{eq:362} in terms of the extension matrix~$\tensorq{E}_{\left(2M + 2\right) \times M}$, the reduction matrix~$\tensorq{R}_{M \times \left(2M + 2\right)}$ and the ``big'' circulant matrix $\tensorq{C}_{\left(2M + 2\right) \times \left(2M + 2\right)}$, see~\eqref{eq:320}. This yields
  \begin{equation}
    \label{eq:364}
    \ddd{\vec{u}_h}{t}
    -
    \tensorq{R}_{M \times \left(2M + 2\right)}\tensorq{C}_{\left(2M + 2\right) \times \left(2M + 2\right)} \tensorq{E}_{\left(2M + 2\right) \times M}
    \vec{u}_h
    =
    \vec{0}
    .
  \end{equation}
  Next we multiply the whole equation from the left by the extension matrix $\tensorq{E}_{\left(2M + 2\right) \times M}$, and we use inverse relation~\eqref{eq:312}. (Recall that the action of ``big'' circulant matrix preserves the symmetry of the odd extension vector~\eqref{eq:292}.) Doing so we arrive at
  \begin{equation}
    \label{eq:365}
    \ddd{\vec{u}_h^{\text{odd}}}{t}
    -
    \tensorq{C}_{\left(2M + 2\right) \times \left(2M + 2\right)} \vec{u}_h^{\text{odd}} 
    =
    \vec{0}
    ,
  \end{equation}
  where $\vec{u}_h^{\text{odd}} =  \tensorq{E}_{\left(2M + 2\right) \times M}  \vec{u}_h$ denotes the odd extension of the original vector $\vec{u}_h$. Problem~\eqref{eq:365} for grid values $\vec{u}_h^{\text{odd}}$ is however a problem we have already dealt with in the periodic case. In particular~\eqref{eq:365} has the desired convolution structure as it can be rewritten as
  \begin{equation}
    \label{eq:366}
    \ddd{\vec{u}_h^{\text{odd}}}{t}
    -
    \frac{1}{h}
    \discreteperiodicconvolution{\widetilde{\vec{c}}_h}{\vec{u}_h^{\text{odd}}}
    =
    \vec{0},
  \end{equation}
  where the vector $\widetilde{\vec{c}}_h$ is the first column of the ``big'' circulant matrix $ \tensorq{C}_{\left(2M + 2\right) \times \left(2M + 2\right)}$. At this point we are at the same position as in the proof of Theorem~\eqref{thr:3} for the periodic lattice, see equation~\eqref{eq:261}. We can thus follow the same steps as in proof of Theorem~\eqref{thr:3} and conclude that $u_h$, which is the bandwidth limited interpolant obtained from the odd extension~$\vec{u}_h^{\text{odd}}$ of grid values~$\vec{u}_h$, solves
  \begin{equation}
    \label{eq:367}
    \ppd{u_h}{t}
    +
    \convolution{
      \left(
      \frac{1}{h}
      \InverseFourierTransform{\frac{\FourierTransformSemidiscrete{\vec{C}_h}}{\xi^2}}
      \right)
    }
    {
      \ppd{u_h}{x}
    }
    =
    0,
  \end{equation}
  which was to prove.

  On the other hand, assume that $u$ is a function that solves~\eqref{eq:356}. Following similar arguments as in the proof of Theorem~\ref{thr:3} we can conclude that the function $u$ inherits the properties of initial conditions. In particular, if the initial conditions are constructed as bandwidth limited interpolants obtained by odd extensions of grid values, then the function~$u$ itself must be, for all positive times, also a bandwidth limited interpolant obtained by an odd extension of grid values. Once we get this structure preserving property we can follow word by word the proof of Theorem~\ref{thr:3}, and we get the desired conclusion that $u = u_h$, where $u_h$ is constructed as the bandwidth limited interpolant based on odd extension of grid values that solve~\eqref{eq:351}.
  
\end{proof}

\subsubsection{Beyond the nearest neighbour interaction}
\label{sec:beyond-near-neighb}

Theorem~\eqref{thr:4} provides complete characterisation of discrete \emph{nearest} neighbour interaction lattices with fixed ends/zero Dirichlet boundary conditions and their continuous counterparts. The extension to \emph{multiple} neighbours interaction is however not so straightforward as in the infinite/periodic setting. The problem is that a finite domain has boundaries, and if a particle is close to the boundary, then it is not \emph{a priori} clear as how the particle should interact with its neighbours beyond the domain boundary---in fact there are no such neighbours. In what follows we investigate how such beyond-the-domain interaction should look like if we want to exploit the circulant structure.

\label{sec:pent-matr}
As in Section~\eqref{sec:corr-betw-discr-2} we start with a ``big'' circulant matrix $\tensorq{C}_{\left(2M + 2\right) \times \left(2M + 2\right)}$, see~\eqref{eq:320}. However instead of a tridiagonal matrix we now consider a \emph{pentadiagonal circulant matrix}
\begin{subequations}
  \label{eq:big-circulant-to-small-multidiagonal}  
\begin{equation}
  \label{eq:368}
  \tensorq{C}_{\left(2M + 2\right) \times \left(2M + 2\right)} 
  =_{\bydefinition}
  \begin{bmatrix}
    c_{h, 0}& c_{h, 1}& c_{h, 2}& &  &  &   c_{h, 2} & c_{h, 1}\\
    c_{h, 1}& c_{h, 0}& c_{h, 1}& c_{h, 2}&  &  &  &   c_{h, 2} \\
    c_{h, 2}& c_{h, 1}& c_{h, 0}&  c_{h, 1} & c_{h, 2} &  &  &  \\
            & c_{h, 2}& c_{h, 1}& c_{h, 0}&   c_{h, 1} & c_{h, 2} &  &  & \\
            &  &  \ddots &  \ddots & \ddots &  \ddots &  \ddots  & \\
            &  &  & c_{h, 2} & c_{h, 1} & c_{h, 0}& c_{h, 1}&  c_{h, 2}\\
    c_{h, 2}       &  &  & & c_{h, 2} & c_{h, 1}& c_{h, 0}& c_{h, 1}\\
    c_{h, 1}&  c_{h, 2}  &  & & & c_{h,2}  & c_{h, 1}& c_{h, 0}\\
  \end{bmatrix}_{\left(2M + 2\right) \times \left(2M + 2\right)}
  .
\end{equation}
(The treatment of a pentadiagonal matrix is straightforward to extend to multidiagonal case.) Now we apply extension/reduction matrices to the ``big'' circulant matrix, see the same step in the previous Section~\eqref{sec:near-neighb-inter}, equation~\eqref{eq:321}. The application of extension/reduction matrices yields
\begin{equation}
  \label{eq:369}
  \tensorq{R}_{M \times \left(2M + 2\right)}\tensorq{C}_{\left(2M + 2\right) \times \left(2M + 2\right)} \tensorq{E}_{\left(2M + 2\right) \times M} = \tensorq{B}_{M \times M},
\end{equation}
where
\begin{equation}
  \label{eq:370}
  \tensorq{B}_{M \times M}
  =
  \begin{bmatrix}
 \left( c_{h, 0}-c_{h, 2} \right) & c_{h, 1} & c_{h, 2} &  &  &  &  &  \\
 c_{h, 1} & c_{h, 0} & c_{h, 1} & c_{h, 2} &  &  &  &  \\
 c_{h, 2} & c_{h, 1} & c_{h, 0} & c_{h, 1} & c_{h, 2} &  &  &  \\
  & c_{h, 2} & c_{h, 1} & c_{h, 0} & c_{h, 1} & c_{h, 2} &  &  \\
  &  & \ddots & \ddots & \ddots & \ddots & \ddots &  \\
  &  &  & c_{h, 2} & c_{h, 1} & c_{h, 0} & c_{h, 1} & c_{h, 2} \\
  &  &  &  & c_{h, 2} & c_{h, 1} & c_{h, 0} & c_{h, 1} \\
  &  &  &  &  & c_{h, 2} & c_{h, 1} & \left( c_{h, 0}-c_{h, 2} \right) \\
\end{bmatrix}
.
\end{equation}
\end{subequations}
The matrix $\tensorq{B}_{M \times M}$ now describes the interactions in between $M$ particles in the original domain $[0, \pi]$. If we want to interpret $\tensorq{B}_{M \times M}$ as a finite difference type scheme for the second derivative, that is
\begin{equation}
  \label{eq:371}
  \left. \ddd{u}{x} \right|_{x = \vec{x}_h}
  =_{\bydefinition}
  \begin{bmatrix}
    \left. \ddd{u}{x} \right|_{x = x_{h, 1}} \\
    \left. \ddd{u}{x} \right|_{x = x_{h, 2}} \\
    \left. \ddd{u}{x} \right|_{x = x_{h, 3}} \\
    \left. \ddd{u}{x} \right|_{x = x_{h, 4}} \\
    \vdots \\
    \left. \ddd{u}{x} \right|_{x = x_{h, M-2}} \\
    \left. \ddd{u}{x} \right|_{x = x_{h, M-1}} \\
    \left. \ddd{u}{x} \right|_{x = x_{h, M}} \\
  \end{bmatrix}
  \stackrel{?}{\approx}
    \begin{bmatrix}
 \left(c_{h, 0}-c_{h, 2}\right) & c_{h, 1} & c_{h, 2} &  &  &  &  &  \\
 c_{h, 1} & c_{h, 0} & c_{h, 1} & c_{h, 2} &  &  &  &  \\
 c_{h, 2} & c_{h, 1} & c_{h, 0} & c_{h, 1} & c_{h, 2} &  &  &  \\
  & c_{h, 2} & c_{h, 1} & c_{h, 0} & c_{h, 1} & c_{h, 2} &  &  \\
  &  & \ddots & \ddots & \ddots & \ddots & \ddots &  \\
  &  &  & c_{h, 2} & c_{h, 1} & c_{h, 0} & c_{h, 1} & c_{h, 2} \\
  &  &  &  & c_{h, 2} & c_{h, 1} & c_{h, 0} & c_{h, 1} \\
  &  &  &  &  & c_{h, 2} & c_{h, 1} & \left(c_{h, 0}-c_{h, 2}\right) \\
    \end{bmatrix}
    \begin{bmatrix}
          u_{h, 1} \\
          u_{h, 2} \\
          u_{h, 3} \\
          u_{h, 4} \\
          \vdots \\
          u_{h, M-2} \\
          u_{h, M-1} \\
          u_{h, M} \\
    \end{bmatrix}
\end{equation}
we must identify the corresponding values of $\left\{ c_{h, i} \right\}_{i=0}^2$. Clearly, the middle rows of  $\tensorq{B}_{M \times M}$ must be tantamount to the centred higher order finite difference schemes for the second derivative. The values $\left\{ c_{h, i} \right\}_{i=0}^2$ are well-known and they read
\begin{equation}
  \label{eq:372}
  c_{h, 0} = -\frac{5}{2} \frac{1}{h^2}, \qquad c_{h, 1} = \frac{4}{3} \frac{1}{h^2}, \qquad c_{h, 2} = - \frac{1}{12} \frac{1}{h^2},
\end{equation}
(The same can be done for higher order schemes---multidiagonal matrices---see especially~\cite{fornberg.b:calculation}.) Using these values we get the following approximation of the second derivative $\left. \ddd{u}{x} \right|_{x = x_i}$ at point $x_i$,
\begin{equation}
  \label{eq:373}
  \left. \ddd{u}{x} \right|_{x = x_i}
  =
  \frac{-u_{h, i-2} + 16 u_{h, i-1} - 30 u_{h, i} + 16 u_{h, i+1} - u_{h, i+2}}{12 h^2}
  +
  \bigo{h^4}
  ,
\end{equation}
and this formula holds for $i=2, \dots, M-1$. (Provided that we use the fixed ends/zero Dirichlet boundary condition and convention $u_{h, 0} = 0$, $u_{h, M+1}$). The application of matrix $\tensorq{B}_{M \times M}$ with the coefficients chosen as in~\eqref{eq:372} to the column vector of grid values is thus identical to the application of the finite differences approximation~\eqref{eq:373}.  The only problematic rows in~\eqref{eq:371} are the \emph{first} and the \emph{last} row, where the finite differences formula~\eqref{eq:373} would require us to work with~$u_{h, -1}$ and~$u_{h, M+2}$.

Using the Taylor expansion centred at $x_{h, 1}$ and $x_{h, M}$ we however find that the application of matrix~$\tensorq{B}_{M \times M}$ with the coefficients chosen as in~\eqref{eq:372} to the column vector of grid values yields in the first row the following approximation
\begin{equation}
    \label{eq:374}
    \left( c_{h, 0}-c_{h, 2} \right) \left. u  \right|_{x = x_{h, 1}} +  c_{h, 1}  \left. u  \right|_{x = x_{h, 2}} +  c_{h, 2}  \left. u  \right|_{x = x_{h, 3}}
    \\
    =
    -\frac{7}{6 h^2} \left. u \right|_{x = x_{h, 1}} + \frac{7}{6 h} \left. \dd{u}{x} \right|_{x = x_{h, 1}}+ \frac{1}{2} \left. \ddd{u}{x} \right|_{x = x_{h, 1}} + \bigo{h}.
  \end{equation}
  (A similar formula is obtained also for the last row.) It thus seems that the first/last row do not give an approximation of the second derivative. However, if we exploit the fixed ends/zero Dirichlet boundary conditions $\left. u \right|_{x =x_{h, 0}} = 0$ and $\left. u \right|_{x =x_{h, M+1}} = 0$, we can use the Taylor expansion centred at $x_{h, 0}$ and $x_{h, M+1}$ respectively, and we find that
  \begin{equation}
    \label{eq:375}
    0 = \left. u \right|_{x = x_{h, 0} = x_{h, 1} - h} = \left. u \right|_{x =x_{h, 1}} - \left. \dd{u}{x} \right|_{x =x_{h, 1}}h + \frac{1}{2} \left. \ddd{u}{x} \right|_{x =x_{h, 1}} h^2 + \bigo{h^3},
  \end{equation}
  which reveals that
  \begin{equation}
    \label{eq:376}
    \left. u \right|_{x =x_{h, 1}} =  \left. \dd{u}{x} \right|_{x =x_{h, 1}}h - \frac{1}{2} \left. \ddd{u}{x} \right|_{x =x_{h, 1}} h^2 + \bigo{h^3}.
  \end{equation}
  Substituting~\eqref{eq:376} into~\eqref{eq:374} then yields
  \begin{equation}
    \label{eq:377}
    \left( c_{h, 0}-c_{h, 2} \right) \left. u  \right|_{x = x_{h, 1}} +  c_{h, 1}  \left. u  \right|_{x = x_{h, 2}} +  c_{h, 2}  \left. u  \right|_{x = x_{h, 3}}
    =
    \frac{13}{12} \left. \ddd{u}{x} \right|_{x =x_{h, 1}} + \bigo{h},
  \end{equation}
  and the same formula can be derived for the last row as well. We can thus conclude that the application of matrix $\tensorq{B}_{M \times M}$ to the vector of grid values \emph{does not provide a consistent finite differences approximation of the second derivative}, the consistency is ruined by the factor $\frac{13}{12}$.

  But this inconsistency is easy to fix by premultiplication of  $\tensorq{B}_{M \times M}$ by a simple diagonal matrix $\tensorq{M}_{M \times M}$. If we thus set
  \begin{equation}
    \label{eq:378}
    \tensorq{M}_{M \times M}
    =_{\bydefinition}
    \begin{bmatrix}
 \frac{12}{13} &  &  &  &  &  &  &  \\
  & 1 &  &  &  &  &  &  \\
  &  & 1 &  &  &  &  &  \\
  &  &  & 1 &  &  &  &  \\
  &  &  &  & \ddots &  &  &  \\
  &  &  &  &  & 1 &  &  \\
  &  &  &  &  &  & 1 &  \\
  &  &  &  &  &  &  & \frac{12}{13}
    \end{bmatrix}
    ,
  \end{equation}
  and if we choose the coefficients $\left\{ c_{h, i} \right\}_{i=0}^2$ as in~\eqref{eq:372}, then the matrix $\tensorq{M}_{M \times M}\tensorq{B}_{M \times M}$ provides us a consistent approximation of the second derivative in the sense that
  \begin{equation}
    \label{eq:379}
    \left. \ddd{u}{x} \right|_{x = \vec{x}_h} \approx \tensorq{M}_{M \times M}\tensorq{B}_{M \times M} \left. \vec{u} \right|_{x = \vec{x}_h},
  \end{equation}
which, if written in full, reads  
  \begin{multline}
    \label{eq:380}
  \begin{bmatrix}
    \left. \ddd{u}{x} \right|_{x = x_{h, 1}} \\
    \left. \ddd{u}{x} \right|_{x = x_{h, 2}} \\
    \left. \ddd{u}{x} \right|_{x = x_{h, 3}} \\
    \left. \ddd{u}{x} \right|_{x = x_{h, 4}} \\
    \vdots \\
    \left. \ddd{u}{x} \right|_{x = x_{h, M-2}} \\
    \left. \ddd{u}{x} \right|_{x = x_{h, M-1}} \\
    \left. \ddd{u}{x} \right|_{x = x_{h, M}} \\
  \end{bmatrix}
  \approx
      \begin{bmatrix}
 \frac{12}{13} &  &  &  &  &  &  &  \\
  & 1 &  &  &  &  &  &  \\
  &  & 1 &  &  &  &  &  \\
  &  &  & 1 &  &  &  &  \\
  &  &  &  & \ddots &  &  &  \\
  &  &  &  &  & 1 &  &  \\
  &  &  &  &  &  & 1 &  \\
  &  &  &  &  &  &  & \frac{12}{13}
      \end{bmatrix}
      \\
    \begin{bmatrix}
 \left( c_{h, 0}-c_{h, 2} \right) & c_{h, 1} & c_{h, 2} &  &  &  &  &  \\
 c_{h, 1} & c_{h, 0} & c_{h, 1} & c_{h, 2} &  &  &  &  \\
 c_{h, 2} & c_{h, 1} & c_{h, 0} & c_{h, 1} & c_{h, 2} &  &  &  \\
  & c_{h, 2} & c_{h, 1} & c_{h, 0} & c_{h, 1} & c_{h, 2} &  &  \\
  &  & \ddots & \ddots & \ddots & \ddots & \ddots &  \\
  &  &  & c_{h, 2} & c_{h, 1} & c_{h, 0} & c_{h, 1} & c_{h, 2} \\
  &  &  &  & c_{h, 2} & c_{h, 1} & c_{h, 0} & c_{h, 1} \\
  &  &  &  &  & c_{h, 2} & c_{h, 1} & \left( c_{h, 0}-c_{h, 2} \right) \\
    \end{bmatrix}
    \begin{bmatrix}
          u_{h, 1} \\
          u_{h, 2} \\
          u_{h, 3} \\
          u_{h, 4} \\
          \vdots \\
          u_{h, M-2} \\
          u_{h, M-1} \\
          u_{h, M} \\
    \end{bmatrix}
    .
  \end{multline}
  The order of the approximation is $\bigo{h}$ in the first and the last row, and $\bigo{h^4}$ in the remaining ``middle rows''. We note that the matrix $\tensorq{M}_{M \times M}\tensorq{B}_{M \times M}$ is not a \emph{symmetric} matrix, hence we \emph{so far} do not know whether its eigenvalues are real and negative, which is something we would expect from a differentiation matrix.

  The calculation above for pentadiagonal circulant matrices can be done for multidiagonal circulants as well. This leads to higher order finite differences schemes in ``middle rows'' and modified schemes in the ``boundary layer'', that is in the first few/last few rows. For example, if we start with a \emph{heptadiagonal circulant matrix}, and if we fix the coefficients in the ``big'' circulant matrix as
  \begin{subequations}
    \label{eq:381}
    \begin{equation}
      \label{eq:382}
      c_{h, 0}= -\frac{49}{18 h^2}, \qquad c_{h, 1}= \frac{3}{2 h^2}, \qquad c_{h, 2}= -\frac{3}{20 h^2}, \qquad c_{h, 3}= \frac{1}{90 h^2},
    \end{equation}
    then the corresponding discretisation of the second order derivative operator reads
    \begin{multline}
      \label{eq:383}
      \begin{bmatrix}
        \left. \ddd{u}{x} \right|_{x = x_{h, 1}} \\
        \left. \ddd{u}{x} \right|_{x = x_{h, 2}} \\
        \left. \ddd{u}{x} \right|_{x = x_{h, 3}} \\
        \left. \ddd{u}{x} \right|_{x = x_{h, 4}} \\
        \vdots \\
        \left. \ddd{u}{x} \right|_{x = x_{h, M-2}} \\
        \left. \ddd{u}{x} \right|_{x = x_{h, M-1}} \\
        \left. \ddd{u}{x} \right|_{x = x_{h, M}} \\
      \end{bmatrix}
      \approx
      \begin{bmatrix}
        \frac{180}{199} &  &  &  &  &  &  &  \\
                       & \frac{90}{89} &  &  &  &  &  &  \\
                       &  & 1 &  &  &  &  &  \\
                       &  &  & 1 &  &  &  &  \\
                       &  &  &  & \ddots &  &  &  \\
                       &  &  &  &  & 1 &  &  \\
                       &  &  &  &  &  & \frac{90}{89} &  \\
                       &  &  &  &  &  &  & \frac{180}{199} \\
      \end{bmatrix}
      \\
      \begin{bmatrix}
 \left( c_{h, 0}-c_{h, 2} \right)  & \left( c_{h, 1} -c_{h, 3} \right)  & c_{h, 2}  & c_{h, 3}  &  &  &  &  &  &  \\
 \left( c_{h, 1} -c_{h, 3} \right)  & c_{h, 0} & c_{h, 1}  & c_{h, 2}  & c_{h, 3}  &  &  &  &  &  \\
 c_{h, 2}  & c_{h, 1}  & c_{h, 0} & c_{h, 1}  & c_{h, 2}  & c_{h, 3}  &  &  &  &  \\
 c_{h, 3}  & c_{h, 2}  & c_{h, 1}  & c_{h, 0} & c_{h, 1}  & c_{h, 2}  & c_{h, 3}  &  &  &  \\
  & c_{h, 3}  & c_{h, 2}  & c_{h, 1}  & c_{h, 0} & c_{h, 1}  & c_{h, 2}  & c_{h, 3}  &  &  \\
  &  & \ddots & \ddots & \ddots & \ddots & \ddots & \ddots & \ddots &  \\
  &  &  & c_{h, 3}  & c_{h, 2}  & c_{h, 1}  & c_{h, 0} & c_{h, 1}  & c_{h, 2}  & c_{h, 3}  \\
  &  &  &  & c_{h, 3}  & c_{h, 2}  & c_{h, 1}  & c_{h, 0} & c_{h, 1}  & c_{h, 2}  \\
  &  &  &  &  & c_{h, 3}  & c_{h, 2}  & c_{h, 1}  & c_{h, 0} & \left( c_{h, 1} -c_{h, 3} \right)  \\
  &  &  &  &  &  & c_{h, 3}  & c_{h, 2}  & \left( c_{h, 1} -c_{h, 3} \right)  & \left( c_{h, 0}-c_{h, 2} \right)  \\
      \end{bmatrix}
      \begin{bmatrix}
          u_{h, 1} \\
          u_{h, 2} \\
          u_{h, 3} \\
          u_{h, 4} \\
          \vdots \\
          u_{h, M-2} \\
          u_{h, M-1} \\
          u_{h, M} \\
    \end{bmatrix}
    ,
  \end{multline}
\end{subequations}
where the order of the second derivative approximation is $\bigo{h}$ in the first and second row and in the last and next to the last row. In the remaining rows we however get approximation order $\bigo{h^6}$.

Concerning the eigenvalues of~$\tensorq{M}_{M \times M}\tensorq{B}_{M \times M}$ we can not easily find closed formulae for the eigenvalues, but we can provide a \emph{tight bound on the eigenvalues}. We exploit the fact that we know the eigenvalues of~$\tensorq{B}_{M \times M}$, see Lemma~\eqref{lm:2}, and that the correction matrix~$\tensorq{M}_{M \times M}$ is close to the identity matrix. These two facts allow us to derive a tight bound on the Rayleigh quotient and hence on the eigenvalues of~$\tensorq{M}_{M \times M}\tensorq{B}_{M \times M}$, see Lemma~\ref{lm:3}.

\begin{lemma}[Eigenvalues of matrices $\tensorq{B}_{M \times M}$ generated by extension/reduction procedure]
  \label{lm:2}
  Let $\tensorq{B}_{M \times M}$ be the $M \times M$ matrix generated from the $\left(2M + 2\right) \times \left(2M + 2\right)$ circulant matrix $\tensorq{C}_{\left(2M + 2\right) \times \left(2M + 2\right)}$,
  \begin{equation}
    \label{eq:384}
    \tensorq{C}_{\left(2M + 2\right) \times \left(2M + 2\right)}
    =_{\bydefinition}
    \begin{bmatrix}
      c_{h, 0} & c_{h, 1} & \cdots &  c_{h, \frac{N}{2} - 1} &  c_{h, \frac{N}{2}} &  c_{h, \frac{N}{2} - 1} & \cdots & c_{h, 2} & c_{h, 1} \\
      c_{h, 1} & c_{h, 0} & \cdots &  c_{h, \frac{N}{2} - 2} &  c_{h, \frac{N}{2} - 1} &  c_{h, \frac{N}{2} } & \cdots & c_{h, 3} & c_{h, 2} \\
      \vdots & \vdots & \ddots & \vdots & \vdots & \vdots & \reflectbox{$\ddots$} & \vdots & \vdots \\
      c_{h, \frac{N}{2} - 2} & c_{h, \frac{N}{2} - 1} & \cdots & c_{h, 0}  &  c_{h, 1} & c_{h, 2} & \cdots & c_{h, \frac{N}{2}} & c_{h, \frac{N}{2} - 1} \\

      c_{h, \frac{N}{2} - 1} & c_{h, \frac{N}{2} - 2} & \cdots & c_{h, 1}  &  c_{h, 0} & c_{h, 1} & \cdots & c_{h, \frac{N}{2} - 1} & c_{h, \frac{N}{2}} \\
      c_{h, \frac{N}{2}} & c_{h, \frac{N}{2} - 1} & \cdots & c_{h, 2}  &  c_{h, 1} & c_{h, 0} & \cdots & c_{h, \frac{N}{2} - 2} & c_{h, \frac{N}{2} - 1} \\
      \vdots & \vdots & \reflectbox{$\ddots$} & \vdots & \vdots & \vdots & \ddots & \vdots & \vdots \\
      c_{h, 2} & c_{h, 3} & \cdots &  c_{h, \frac{N}{2}-1} &  c_{h, \frac{N}{2}-2} &  c_{h, \frac{N}{2} - 1} & \cdots & c_{h, 0} & c_{h, 1} \\
      c_{h, 1} & c_{h, 2} & \cdots &  c_{h, \frac{N}{2}} &  c_{h, \frac{N}{2}-1} &  c_{h, \frac{N}{2} - 2} & \cdots & c_{h, 1} & c_{h, 0}
    \end{bmatrix}
    ,
  \end{equation}
  with $c_{h, \frac{N}{2}} = 0$ via the extension/reduction procedure
  \begin{equation}
    \label{eq:385}
    \tensorq{B}_{M \times M} =_{\bydefinition} \tensorq{R}_{M \times \left(2M + 2\right)}\tensorq{C}_{\left(2M + 2\right) \times \left(2M + 2\right)} \tensorq{E}_{\left(2M + 2\right) \times M},
  \end{equation}
  where $\tensorq{E}_{\left(2M + 2\right) \times M}$ and $\tensorq{R}_{M \times \left(2M + 2\right)}$ denote the extension/reduction matrices defined in~\eqref{eq:302}. Then the eigenvalues of~$\tensorq{B}_{M \times M}$ are given by the formula
  \begin{equation}
    \label{eq:386}
        \lambda_{h, m}
    =
    c_{h, 0}
    +
    2
    \sum_{j= 1}^{M}
    c_{h, j}
    \cos \left(j mh \right)
    ,
    \qquad
    m = 1, \dots, M.
  \end{equation}
\end{lemma}

\begin{proof}
  Lemma~\ref{lm:1} on eigenvalues of circulant matrices states that the circulant matrix $\tensorq{C}_{\left(2M + 2\right) \times \left(2M + 2\right)}$ has $M$ eigenvalues of algebraic multiplicity \emph{two}---these are the eigenvalues associated to the wavenumbers $m=1, \dots, M = \frac{N}{2} -1$---and two eigenvalues of multiplicity \emph{one}, namely the eigenvalues associated to wavenumbers $m = 0$ and $m = M + 1 = \frac{N}{2}$. The diagonalisation of the circulant matrix $\tensorq{C}_{\left(2M + 2\right) \times \left(2M + 2\right)} $ via the discrete Fourier transform thus leads to the diagonal matrix with the structure
  \begin{equation}
    \label{eq:387}
    \begin{bmatrix}
      \lambda_{h, 1} & & & & & & & \\
                & \ddots & & & & & & \\ 
                & & \lambda_{h, M} & & & & & \\
                & & & \lambda_{h, M+1} & & & & \\
                & & & & \lambda_{h, M} & & & \\
                & & & & & \ddots & & \\
                & & & & & & \lambda_{h, 1} & \\
                & & & & & & & \lambda_{h, 0} 
    \end{bmatrix}
    .
  \end{equation}

Once we know that the diagonalisation of $\tensorq{C}_{\left(2M + 2\right) \times \left(2M + 2\right)}$ invariably leads to the diagonal matrix with the structure~\eqref{eq:387}, we can follow the same steps as in Section~\eqref{sec:discr-four-transf-3}, see the discussion following equation~\eqref{eq:338}. In particular, we can conclude that any matrix $\tensorq{B}_{M \times M}$ generated out of the ``big'' circulant matrix~$\tensorq{C}_{\left(2M + 2\right) \times \left(2M + 2\right)} $ by the reduction/extension procedure
\begin{equation}
  \label{eq:388}
  \tensorq{B}_{M \times M} = \tensorq{R}_{M \times \left(2M + 2\right)}\tensorq{C}_{\left(2M + 2\right) \times \left(2M + 2\right)} \tensorq{E}_{\left(2M + 2\right) \times M}
\end{equation}
is diagonalised by the discrete sine transform. Furthermore, its eigenvalues are determined by the non-simple eigenvalues $\lambda_{h, 1}, \dots \lambda_{h, M}$ of the ``big'' circulant matrix $\tensorq{C}_{\left(2M + 2\right) \times \left(2M + 2\right)}$, which are calculated by the discrete Fourier transform of the generating vector of the circulant matrix $\tensorq{C}_{\left(2M + 2\right) \times \left(2M + 2\right)}$, see formula~\eqref{eq:238}. Using~\eqref{eq:238}, we thus see that the eigenvalues of $\tensorq{B}_{M \times M}$ are given by~\eqref{eq:386}.  
\end{proof}

\begin{lemma}[Bound on eigenvalues of product $-\tensorq{M}_{M \times M}\tensorq{B}_{M \times M}$]
  \label{lm:3}
  Let $\tensorq{B}_{M \times M}$ be a negative definite matrix, and let $\tensorq{M}_{M \times M}$ be a symmetric positive definite matrix. Then the eigenvalues~$\left\{\lambda_{- \tensorq{M}_{M \times M} \tensorq{B}_{M \times M}, i} \right\}_{i = 1}^M$ in the eigenvalue problem
  \begin{equation}
    \label{eq:389}
    - \tensorq{M}_{M \times M}\tensorq{B}_{M \times M} \vec{v} = \lambda_{-\tensorq{M}_{M \times M} \tensorq{B}_{M \times M}} \vec{v} 
  \end{equation}
  can be estimated in terms of eigenvalues  $\left\{ \lambda_{-\tensorq{B}_{M \times M}, i} \right\}_{i=1}^M$ in the eigenvalue problem
  \begin{equation}
    \label{eq:390}
    -\tensorq{B}_{M \times M} \vec{x} = \lambda_{-\tensorq{B}_{M \times M}} \vec{x}
  \end{equation}
  as
  \begin{equation}
    \label{eq:391}
    \lambda_{\tensorq{M}_{M \times M}, \, \min}
    \lambda_{-\tensorq{B}_{M \times M}, i}
    \leq
    \lambda_{-\tensorq{M}_{M \times M} \tensorq{B}_{M \times M}, i}
    \leq
    \lambda_{\tensorq{M}_{M \times M}, \, \max}
    \lambda_{-\tensorq{B}_{M \times M}, i},
  \end{equation}
  where $\lambda_{\tensorq{M}_{M \times M}, \, \min}$ and $\lambda_{\tensorq{M}_{M \times M}, \, \max}$ denote the smallest/largest eigenvalue of $\tensorq{M}_{M \times M}$.
\end{lemma}

\begin{proof}
  Since $\tensorq{M}_{M \times M}$ is a symmetric positive definite matrix, we can take its square root, and we can write
  \begin{equation}
    \label{eq:392}
    \tensorq{M}_{M \times M}\tensorq{B}_{M \times M} = \tensorq{M}_{M \times M}^{\frac{1}{2}} \left( \tensorq{M}_{M \times M}^{\frac{1}{2}}\tensorq{B}_{M \times M} \tensorq{M}_{M \times M}^{\frac{1}{2}} \right) \tensorq{M}_{M \times M}^{-\frac{1}{2}},
  \end{equation}
  This reveals that the matrix $-\tensorq{M}_{M \times M}\tensorq{B}_{M \times M}$ is similar, but not unitarily similar, to the \emph{symmetric positive definite matrix} $-\tensorq{M}_{M \times M}^{\frac{1}{2}}\tensorq{B}_{M \times M} \tensorq{M}_{M \times M}^{\frac{1}{2}}$. (Provided that $-\tensorq{B}_{M \times M}$ is symmetric positive definite matrix, which is our case.) Furthermore, the eigenvalue problem
  \begin{equation}
    \label{eq:393}
    -\tensorq{M}_{M \times M}\tensorq{B}_{M \times M} \vec{v} = \lambda \vec{v}
  \end{equation}
  can be rewritten as
  \begin{equation}
    \label{eq:394}
    -\tensorq{M}_{M \times M}^{\frac{1}{2}} \left( \tensorq{M}_{M \times M}^{\frac{1}{2}}\tensorq{B}_{M \times M} \tensorq{M}_{M \times M}^{\frac{1}{2}} \right) \tensorq{M}_{M \times M}^{-\frac{1}{2}} \vec{v} = \lambda \vec{v}.
  \end{equation}
  This implies that the eigenvalues of matrix $-\tensorq{M}_{M \times M}\tensorq{B}_{M \times M}$ are the same as the eigenvalues of matrix $-\tensorq{M}_{M \times M}^{\frac{1}{2}}\tensorq{B}_{M \times M} \tensorq{M}_{M \times M}^{\frac{1}{2}}$, and that the eigenvectors of matrix $-\tensorq{M}_{M \times M}\tensorq{B}_{M \times M}$ are obtained by the transformation $\vec{v} =  \tensorq{M}_{M \times M}^{\frac{1}{2}} \vec{w}$, where $\vec{w}$ are eigenvectors of  matrix $-\tensorq{M}_{M \times M}^{\frac{1}{2}}\tensorq{B}_{M \times M} \tensorq{M}_{M \times M}^{\frac{1}{2}}$, that is
  \begin{equation}
    \label{eq:395}
     -\tensorq{M}_{M \times M}^{\frac{1}{2}}\tensorq{B}_{M \times M} \tensorq{M}_{M \times M}^{\frac{1}{2}} \vec{w} = \lambda \vec{w}.
   \end{equation}

   The eigenvalue problem~\eqref{eq:395} for the symmetric matrix $- \tensorq{M}_{M \times M}^{\frac{1}{2}}\tensorq{B}_{M \times M} \tensorq{M}_{M \times M}^{\frac{1}{2}}$ can be solved using the Rayleigh quotient technique. The Rayleigh quotient for $- \tensorq{M}_{M \times M}^{\frac{1}{2}}\tensorq{B}_{M \times M} \tensorq{M}_{M \times M}^{\frac{1}{2}}$ can be however rewritten as follows
   \begin{multline}
     \label{eq:396}
     -
     \frac{\left( \vectordot{\tensorq{M}_{M \times M}^{\frac{1}{2}}\tensorq{B}_{M \times M} \tensorq{M}_{M \times M}^{\frac{1}{2}} \vec{u} \right)}{\vec{u}}}{\vectordot{\vec{u}}{\vec{u}}}
     =
     -
     \frac{\left( \vectordot{\tensorq{B}_{M \times M} \tensorq{M}_{M \times M}^{\frac{1}{2}} \vec{u} \right)}{\tensorq{M}_{M \times M}^{\frac{1}{2}}\vec{u}}}{\vectordot{\tensorq{M}_{M \times M}^{\frac{1}{2}} \vec{u}}{\tensorq{M}_{M \times M}^{\frac{1}{2}} \vec{u}}}
     \frac{\vectordot{\tensorq{M}_{M \times M}^{\frac{1}{2}} \vec{u}}{\tensorq{M}_{M \times M}^{\frac{1}{2}} \vec{u}}}{\vectordot{\vec{u}}{\vec{u}}}
     \\
     =
     -
     \frac{\vectordot{\tensorq{B}_{M \times M} \vec{x}}{\vec{x}}}{\vectordot{\vec{x}}{\vec{x}}}
     \frac{\vectordot{\tensorq{M}_{M \times M} \vec{u}}{\vec{u}}}{\vectordot{\vec{u}}{\vec{u}}}
     ,
   \end{multline}
   where we have denoted $\vec{x} = _{\bydefinition}  \tensorq{M}_{M \times M}^{\frac{1}{2}} \vec{u}$. This reveals that
   \begin{equation}
     \label{eq:397}
     -
     \lambda_{\tensorq{M}_{M \times M}, \, \min}
     \frac{\vectordot{\tensorq{B}_{M \times M} \vec{x}}{\vec{x}}}{\vectordot{\vec{x}}{\vec{x}}}
     \leq
     -
     \frac{\left( \vectordot{\tensorq{M}_{M \times M}^{\frac{1}{2}}\tensorq{B}_{M \times M} \tensorq{M}_{M \times M}^{\frac{1}{2}} \vec{u} \right)}{\vec{u}}}{\vectordot{\vec{u}}{\vec{u}}}
     \leq
     -
     \lambda_{\tensorq{M}_{M \times M}, \, \max}
     \frac{\vectordot{\tensorq{B}_{M \times M} \vec{x}}{\vec{x}}}{\vectordot{\vec{x}}{\vec{x}}}
     ,
    \end{equation}
    where $\lambda_{\tensorq{M}_{M \times M}, \, \min}$ and $\lambda_{\tensorq{M}_{M \times M}, \, \max}$ denote the smallest/largest eigenvalue of $\tensorq{M}_{M \times M}$. The Courant–-Fischer variational characterisation of eigenvalues, see, for example~\cite[Chapter 7]{meyer.c:matrix}, then gives us
    \begin{equation}
      \label{eq:398}
      \lambda_{\generictensor, i}
      =
      \max_{\dim V = i}
      \min_{\vec{y} \in V, \vec{u} \not = \vec{0}}
      \frac{\vectordot{\generictensor \vec{y}}{\vec{y}}}{\vectordot{\vec{y}}{\vec{y}}},
    \end{equation}
    where $\lambda_{\generictensor, i}$ is the $i$-th eigenvalue of the corresponding matrix $\generictensor \in \R^{m \times m}$ and $V$ is a subspace of $\R^{m}$. (The eigenvalues are sorted in ascending order.) In our case thus have
    \begin{equation}
      \label{eq:399}
      \lambda_{-\tensorq{M}_{M \times M}^{\frac{1}{2}}\tensorq{B}_{M \times M} \tensorq{M}_{M \times M}^{\frac{1}{2}}, i}
      =
      \max_{\dim V = i}
      \min_{\vec{u} \in V, \  \vec{u} \not = \vec{0}}
      \frac{\left( \vectordot{-\tensorq{M}_{M \times M}^{\frac{1}{2}}\tensorq{B}_{M \times M} \tensorq{M}_{M \times M}^{\frac{1}{2}} \vec{u} \right)}{\vec{u}}}{\vectordot{\vec{u}}{\vec{u}}},
    \end{equation}
    which in the virtue of~\eqref{eq:397} implies that the $i$-th eigenvalue of matrix $-\tensorq{M}_{M \times M}^{\frac{1}{2}}\tensorq{B}_{M \times M} \tensorq{M}_{M \times M}^{\frac{1}{2}}$, and thus also the $i$-th eigenvalue of matrix $-\tensorq{M}_{M \times M} \tensorq{B}_{M \times M}$ is bounded by rescaled $i$-th eigenvalue of matrix $-\tensorq{B}_{M \times M}$ 
    \begin{equation}
      \label{eq:400}
      \lambda_{\tensorq{M}_{M \times M}, \, \min}
      \lambda_{-\tensorq{B}_{M \times M}, i}
      \leq
      \lambda_{-\tensorq{M}_{M \times M} \tensorq{B}_{M \times M}, i}
      \leq
      \lambda_{\tensorq{M}_{M \times M}, \, \max}
      \lambda_{-\tensorq{B}_{M \times M}, i}.
    \end{equation}
  \end{proof}

Let us now investigate how good are various approximations of the second derivative operator with fixed ends/zero Dirichlet boundary conditions. In particular, we focus on the quality of approximation of \emph{eigenvalues} $\lambda^{\text{continuous}}$  of the \emph{continuous eigenproblem}
\begin{subequations}
  \label{eq:401}
  \begin{align}
    \label{eq:402}
    -\ddd{u}{x} &= \lambda^{\text{continuous}} u, \\
    \label{eq:403}
    \left. u \right|_{x=0, \pi} &= 0.
  \end{align}
\end{subequations}
The first $M$ eigenvalues of the continuous eigenproblem are given by the formula $\lambda_n^{\text{continuous}} =  n^2$, $n = 1, \dots , M$. We want to compare these eigenvalues with the $M$ eigenvalues of the approximations of~\eqref{eq:401} by various $M \times M$ differentiation matrices $\tensorq{M}_{M \times M} \tensorq{B}_{M \times M}$, that is with the eigenvalues obtained by the solution of the \emph{discrete eigenproblem}
\begin{equation}
  \label{eq:404}
  -\tensorq{M}_{M \times M}\tensorq{B}_{M \times M} \vec{u} =  \lambda^{\text{discrete}} \vec{u}.
\end{equation}

The comparison results are summarised in Table~\ref{tab:spectrum-fd-versus-dst}, where $\tensorq{M}_{M \times M,\, 2j+1} \tensorq{B}_{M \times M,\, 2j+1}$ denote the matrices approximating the second derivative operator with fixed ends/zero Dirichlet boundary conditions, and where $2j+1$ is the number of non-zero diagonals in $\tensorq{B}_{M \times M,\, 2j+1}$. (For $j=1$ the matrix $\tensorq{B}_{M \times M,\, 3}$ represents the second derivative approximation by the classical tridiagonal matrix~\eqref{eq:323}; for $j=2$ we get the pentadiagonal matrix $\tensorq{B}_{M \times M,\, 5}$ and the diagonal correction matrix $\tensorq{M}_{M \times M,\, 5}$ representing the finite differences approximation~\eqref{eq:379}; for $j=3$ we get the heptadiagonal matrix $\tensorq{B}_{M \times M,\, 7}$ and the diagonal correction matrix $\tensorq{M}_{M \times M,\, 7}$ representing the finite differences approximation~\eqref{eq:381}, and so forth.) Computations reported in Table~\ref{tab:spectrum-fd-versus-dst} confirm the theoretically expected behaviour.

\emph{The first few eigenvalues for the discrete problems correspond very well to the first few eigenvalues of the continuous problem}. However, the \emph{high index eigenvalues of the discrete problems are nowhere near to the corresponding eigenvalues of the continuous problem}, albeit the quality of the eigenvalue approximation increases with the order of the approximation (number of non-zero diagonals in $\tensorq{B}_{M \times M,\, 2j+1}$). On the other hand, the sine transform based discretisation of eigenproblem~\eqref{eq:401}, see Section~\eqref{sec:discr-sine-transf-3} and especially formula~\eqref{eq:347}, gives by the construction \emph{$M$ discrete eigenvalues that exactly correspond to the first $M$ eigenvalues} of the continuous problem~\eqref{eq:401}. 

%
%
\begin{table}[h]
  \centering
  \begin{tabular}{lllllll}
    \toprule
    & Continuous & \multicolumn{5}{c}{Numerically computed eigenvalues $\lambda_n^{\text{discrete}, \, 2j+1}$ of} \\
                           & operator $-\ddd{}{x}$ & \multicolumn{5}{c}{discrete operator $-\tensorq{M}_{M \times M,\, 2j+1} \tensorq{B}_{M \times M,\, 2j+1}$} \\
    \cmidrule(lr){2-2} \cmidrule(lr){3-7}
    $n$  & $\lambda_n^{\text{continuous}}$ & $\lambda_n^{\text{discrete}, \,3}$ & $\lambda_n^{\text{discrete}, \,5}$ & $\lambda_n^{\text{discrete}, \,7}$ & $\lambda_n^{\text{discrete}, \,9}$ & $\lambda_n^{\text{discrete}, \,{11}}$ \\
    \midrule
    1& 1& 0.999684& 0.999975& 0.999982& 0.999983& 0.999983\\ 
2& 4& 3.99494& 3.99959& 3.99971& 3.99972& 3.99973\\ 
3& 9& 8.97442& 8.99789& 8.9985& 8.9986& 8.99863\\ 
4& 16& 15.9192& 15.9931& 15.9952& 15.9955& 15.9956\\ 
    5& 25& 24.803& 24.9824& 24.9881& 24.989& 24.9893\\
    \multicolumn{7}{c}{$\vdots$} \\
    46& 2116& 1029.34& 1363.73& 1538.35& 1647.91& 1723.98\\ 
47& 2209& 1038.23& 1378.65& 1557.79& 1671.17& 1750.56\\ 
48& 2304& 1045.17& 1390.35& 1573.07& 1689.5& 1771.59\\ 
49& 2401& 1050.15& 1398.76& 1584.07& 1702.74& 1786.8\\ 
50& 2500& 1053.15& 1403.83& 1590.71& 1710.73& 1796\\
    \bottomrule
  \end{tabular}
  \caption{Eigenvalues of discrete operators $-\tensorq{M}_{M \times M,\, 2j+1} \tensorq{B}_{M \times M,\, 2j+1}$  approximating the second derivative operator $-\ddd{}{x}$ with fixed ends/zero Dirichlet boundary conditions. Matrix size $M \times M$, $M=50$.}
  \label{tab:spectrum-fd-versus-dst}
\end{table}

Finally, the performance of the eigenvalue estimate~\eqref{eq:400} is documented in Table~\ref{tab:spectrum-fd-estimate-via-circulant}. Here we report numerically computed eigenvalues of differentiation matrix $-\tensorq{M}_{M \times M,\, 2j+1} \tensorq{B}_{M \times M,\, 2j+1}$ for $j=3$ and the eigenvalue estimates based on the knowledge of eigenvalues of $\tensorq{B}_{M \times M,\, 2j+1}$ and the matrix $\tensorq{M}_{M \times M,\, 2j+1}$. We see that the estimates give decent bounds on the actual eigenvalues and that the eigenvalues of  $\tensorq{M}_{M \times M,\, 2j+1} \tensorq{B}_{M \times M,\, 2j+1}$ are, in this particular case, very close to the eigenvalues of~$\tensorq{B}_{M \times M,\, 2j+1}$.

%
%
\begin{table}[h]
  \centering
  \begin{tabular}{lllll}
    \toprule
    &Numerically computed eigenvalues $\lambda_n^{\text{discrete}, \, 2j+1}$ of & Eigenvalues of & Lower bound on & Upper bound on\\
    &discrete operator $-\tensorq{M}_{M \times M,\, 2j+1} \tensorq{B}_{M \times M,\, 2j+1}$ & $-\tensorq{B}_{M \times M,\, 2j+1}$ & $\lambda_n^{\text{discrete}, \, 2j+1}$ & $\lambda_n^{\text{discrete}, \, 2j+1}$ \\
    \cmidrule(lr){2-2} \cmidrule(lr){3-3} \cmidrule(lr){4-5}
    $n$ & $\lambda_n^{\text{discrete}, \,5}$ & $\lambda_{-\tensorq{B}_{M \times M, \, 2j + 1}, n}^{\text{discrete}, \,5}$ & & \\
    \midrule
    1& 0.999982& 1.& 0.904523& 1.01124\\ 
    2& 3.99971& 4.& 3.61809& 4.04494\\ 
    3& 8.9985& 9.& 8.1407& 9.10112\\ 
    4& 15.9952& 16.& 14.4724& 16.1798\\ 
    5& 24.9881& 25.& 22.613& 25.2809\\
    \multicolumn{5}{c}{$\vdots$} \\
    46& 1538.35& 1538.87& 1391.94& 1556.16\\ 
    47& 1557.79& 1558.12& 1409.35& 1575.62\\ 
    48& 1573.07& 1573.25& 1423.04& 1590.93\\ 
    49& 1584.07& 1584.16& 1432.9& 1601.95\\ 
    50& 1590.71& 1590.73& 1438.85& 1608.61\\                                                                          \bottomrule
  \end{tabular}
  \caption{Eigenvalues of discrete differentiation matrix  $-\tensorq{M}_{M \times M,\, 2j+1} \tensorq{B}_{M \times M,\, 2j+1}$, $j=3$ (heptadiagonal matrix~$\tensorq{B}_{M \times M,\, 2j+1}$), and their estimate~\eqref{eq:400} based on theoretically known eigenvalues of $-\tensorq{B}_{M \times M,\, 2j+1}$. Matrix size $M \times M$, $M = 50$.}
  \label{tab:spectrum-fd-estimate-via-circulant}
\end{table}

\section{Remark on more complex second order spatial differential operators}
\label{sec:remark-more-complex}

\subsection{Regular Sturm--Liouville operators and Liouville substitution}
\label{sec:regul-sturm-liouv}

The analysis outlined in previous sections heavily relied on the \emph{a priori} knowledge of the eigenvalues/eigenvectors of the differential operator $\ddd{}{x}$ with the corresponding boundary conditions. It might thus seem that such analysis is inapplicable for a generic wave-like equation of type
\begin{equation}
  \label{eq:405}
  \rho \ppd{u}{t} - \left( \pd{}{x} \left( p \pd{u}{x} \right) + q u \right) = 0,
\end{equation}
where $\rho$, $p$ and $q$ are given functions of position such that $p(x) >  0$ and $\rho(x) > 0$ for all $x$ in the domain of interest, $x \in (a, b)$. We however show that the Fourier transform tools, in particular the discrete sine transform, are of use even in this case.

The key qualitative features of solution to~\eqref{eq:405} are again encoded in the dispersion relation, while the dispersion relation reduces to an eigenvalue problem. Indeed, taking the Fourier transform in \emph{time} we get the eigenvalue problem
\begin{subequations}
  \label{eq:406}
  \begin{align}
    \label{eq:407}
    \dd{}{x} \left( p \dd{\widehat{u}}{x} \right) + q \widehat{u} &= - \omega^2 \rho \widehat{u}, \\
    \label{eq:408}
    \left. \widehat{u} \right|_{x=a, b} &=0,
  \end{align}
\end{subequations}
where $\widehat{u} = \fouriertransforms_{t \to \omega}\left[u\right]$ is the time Fourier transform of the original function $u$. Using the standard Liouville substitution, see, for example, \cite[Chapter 10, Theorem 6]{birkhoff.g.rota.g:ordinary}, we can convert the eigenvalue problem~\eqref{eq:406} into the canonical Schr\"odinger form
 \begin{equation}
   \label{eq:410}
   \ddd{\widehat{w}}{z} - \widehat{q} \widehat{w} = - \omega^2 \widehat{w},
 \end{equation}
 where the new unknown quantity $w$, the new spatial variable $z$, and the coefficient $\widehat{q}$ are given by the formulae
 \begin{equation}
   \label{eq:liouville-substitution}
   \widehat{u}  =_{\bydefinition} \frac{\widehat{w}}{ \left( p \rho \right)^{\frac{1}{4}}}, \\
     \qquad
     z =_{\bydefinition} \int_{\xi = a}^x \left( \frac{\rho}{p} \right)^{\frac{1}{2}} \, \diff \xi, \\
     \qquad
     \widehat{q} =_{\bydefinition} \frac{q}{\rho} + \frac{1}{\left( p \rho \right)^{\frac{1}{4}}} \ddd{}{z} \left( \left( p \rho \right)^{\frac{1}{4}} \right).
 \end{equation}
 The new eigenvalue problem~\eqref{eq:410} has the same eigenvalues as the original eigenvalue problem~\eqref{eq:406}, but it has to be solved in the interval
 \begin{equation}
   \label{eq:414}
   z \in \left( 0, \int_{\xi = a}^b \left( \frac{p}{\rho} \right)^{\frac{1}{2}} \, \diff \xi \right).
 \end{equation}
 (The fixed ends/zero Dirichlet boundary conditions are inherited from the original problem, but they are enforced at the endpoints of the interval for the $z$ variable.)
 We can further rescale the spatial variable $z$ as
 \begin{equation}
   \label{eq:415}
   Z =_{\bydefinition} \frac{\pi z}{z_{\max}}, \qquad z_{\max} =_{\bydefinition} \int_{\xi = a}^b \left( \frac{p}{\rho} \right)^{\frac{1}{2}} \, \diff \xi,
 \end{equation}
 which gives us the eigenvalue problem
 \begin{subequations}
   \label{eq:416}
   \begin{align}
     \label{eq:417}
     \ddd{W}{Z} - Q W &= - \Omega^2 W, \\
     \label{eq:418}
     \left. W \right|_{Z = 0, \pi} &= 0, 
   \end{align}
 \end{subequations}
 on the canonical interval $Z \in \left( 0, \pi \right)$, where we use the obvious identification
   \begin{equation}
     \label{eq:420}
     W (Z) =_{\bydefinition} \widehat{w} \left( \frac{Z \pi}{z_{\max}} \right), \\
     \qquad
     Q (Z)  =_{\bydefinition}  \left( \frac{z_{\max}}{\pi} \right)^2 \widehat{q}\left( \frac{Z \pi}{z_{\max}} \right), \\
     \qquad
     \Omega^2 =_{\bydefinition}4 \left( \frac{z_{\max}}{\pi} \right)^2 \omega^2.
   \end{equation}
   The eigenvalues of the original problem~\eqref{eq:406} are thus rescaled eigenvalues of the rescaled problem~\eqref{eq:420}.

   The transformed/rescaled eigenvalue problem~\eqref{eq:416} is however almost the same as the eigenvalue problem we have been discussing in previous sections---problem~\eqref{eq:416} is an eigenvalue problem for the operator that consists of the second derivative and a multiplication by a known function. In general, there is no simple analytical formula for the eigenvalues and eigenvectors of~\eqref{eq:416}, hence we seemingly gained nothing by the transformation of the original eigenvalue~\eqref{eq:406} into the form~\eqref{eq:416}. However, the \emph{asymptotic behaviour of eigenfunctions and eigenvalues} for the transformed/rescaled problem~\eqref{eq:416} is well known, and it turns out that \emph{the high index eigenvalues in~\eqref{eq:416} are}, up to a simple shift, \emph{asymptotically the same as the eigenvalues of the simple problem}
   \begin{subequations}
     \label{eq:424}
     \begin{align}
       \label{eq:425}
       \ddd{W}{Z} &= - \Omega^2 W, \\
       \label{eq:426}
       \left. W \right|_{Z = 0, \pi} &= 0,
     \end{align}
   \end{subequations}
   see~\cite[Corollary 3]{fix.g:asymptotic} and comprehensive analysis and references therein. Namely, the asymptotic formula for the $n$-th eigenvalue of~\eqref{eq:416} reads
   \begin{equation}
     \label{eq:428}
     \Omega^2_{n} =  n ^2 + \frac{1}{\pi} \int_{\xi = 0}^\pi Q(\xi) \, \diff \xi + \bigo{\frac{1}{n^2}}.
   \end{equation}

   Since the sine transform based differentiation~\eqref{eq:347} is the best option regarding the discretisation of~\eqref{eq:424}---differentiation matrix of size $M \times M$ exactly preserves the first $M$ eigenvalues of the continuous operator---we can \emph{conjecture} that the discretisation of~\eqref{eq:416} using the sine transform based differentiation will perform very well. The asymptotic behaviour  (large eigenvalues, highly oscillatory eigenfunctions) is well captured by the differentiation matrix due to the asymptotic formula~\eqref{eq:428}, while the first few eigenvalues (slowly varying eigenfunctions) are well captured due to the good approximation properties of Fourier series.

   \subsection{Numerical experiment}
\label{sec:numerical-experiment}

   This expectation is confirmed by the numerical experiment reported in Table~\eqref{tab:sample-problem}. In the domain $(0, \pi)$ we solve eigenvalue problem
   \begin{subequations}
     \label{eq:429}
     \begin{align}
       \label{eq:430}
       - \ddd{u}{x} + \exponential{x} u &= \lambda u, \\
       \label{eq:431}
       \left. u \right|_{x =0, \pi} & =0,
     \end{align}
   \end{subequations}
   which is the standard test problem in numerical solution of eigenvalue problems for regular Sturm--Liouville operators, see~\cite{paine.jw.hoog.fr.ea:on}, \cite{ledoux.v.van-daele.m.ea:efficient} or~\cite{zhang.x:mapped}. In particular, we numerically compute the eigenvalues of the discrete operator
   \begin{equation}
     \label{eq:432}
     -
     \left(
       \inverse{\left(\SineTransformDiscreteMatrix\right)}
       \begin{bmatrix}
         -1 & & & \\
            & -4 &  &\\
            & & \ddots & \\
            & & & -M^2
       \end{bmatrix}
       \SineTransformDiscreteMatrix
     \right)
    \vec{u}_h
    +
    \diag \left( \vec{q}_h \right)
    \vec{u}_h
    =
    \lambda_n^{\text{discrete}}
    \vec{u}_h
    ,
  \end{equation}
  which is a straightforward discretisation of~\eqref{eq:429} that exploits the sine transform based differentiation~\eqref{eq:347}; the symbol~$\vec{q}_h$ denotes the collection of grid values of function $q(x) = \exponential{x}$.

  In Table~\eqref{tab:sample-problem}, second column, we show several eigenvalues of~\eqref{eq:429} computed by the mapped barycentric Chebyshev differentiation matrix (MBCDM) of size $800 \times 800$, see \cite{zhang.x:mapped}. These eigenvalues are assumed to provide accurate representation of the corresponding eigenvalues of the continuous problem~\eqref{eq:429}. The third column in Table~\ref{tab:sample-problem} shows the eigenvalues of the discrete problem~\eqref{eq:432} computed using $M \times M$ differentiation matrix, $M = 500$, and the default matrix eigenvalue computation algorithm in \texttt{Wolfram Mathematica}. \emph{Note that the size of the differentiation matrix in~\eqref{eq:432} corresponds to the highest reported index eigenvalue~$\lambda_n^{\text{discrete}}$.} We see that the computed eigenvalues of the discrete problem~\eqref{eq:432} perfectly match the ``exact'' eigenvalues reported by~\cite{zhang.x:mapped}. 
  
   %
   %
   \begin{table}[h]
     \centering
     \begin{tabular}[h]{lll}
       \toprule
       & \cite[Table 2]{zhang.x:mapped} & discrete eigenproblem~\eqref{eq:432}\\
       \cmidrule(lr){2-3}
       $n$ & $\lambda_n^{\text{discrete}}$ & $\lambda_n^{\text{discrete}}$\\
       \midrule
       1& 4.89666938& 4.89666938\\ 
50& 2507.050434& 2507.050434\\ 
100& 10007.04831& 10007.04831\\ 
150& 22507.04792& 22507.04791\\ 
200& 40007.04778& 40007.04777\\ 
250& 62507.04771& 62507.0477\\ 
300& 90007.04768& 90007.04765\\ 
350& 122507.0477& 122507.0476\\ 
400& 160007.0476& 160007.0475\\ 
450& 202507.0476& 202507.047\\ 
500& 250007.0476& 250005.6482\\
       \bottomrule
     \end{tabular}
     \caption{Eigenvalues of~\eqref{eq:429} numerically computed by the mapped barycentric Chebyshev differentiation matrix (MBCDM) of size $800 \times 800$, see \cite{zhang.x:mapped} and eigenvalues of~\eqref{eq:429} computed numerically using the discrete sine transform based discretisation~\eqref{eq:432} with matrix size $500 \times 500$.}
     \label{tab:sample-problem}
   \end{table}

   Contrary to~\cite{zhang.x:mapped} we have however matched the first $M$ ``exact'' eigenvalues of the continuous problem~\eqref{eq:429} via discrete eigenvalue problem~\eqref{eq:432} with $M \times M$ matrix. This anecdotical observation suggests that the sine transform based discretisation of regular Sturm--Liouville operators might indeed provide \emph{a method to construct discrete operators represented by $M \times M$ matrices whose eigenvalues very well match the first $M$ eigenvalues of the corresponding continuous operator}. Such conjecture however requires a careful analysis which is beyond the scope of the current contribution and will be carried out elsewhere.
   
\section{Conclusion}
\label{sec:conclusion}
We have investigated the interplay between the continuous problem~\eqref{eq:434} and its discrete counterpart~\eqref{eq:435},
\begin{subequations}
  \label{eq:433}
  \begin{align}
    \label{eq:434}
    \ppd{u_h}{t} - \mathcal{L}_h u_h &= 0, \\
    \label{eq:435}
    \ddd{\vec{u}_h}{t} - \tensorq{L}_h \vec{u}_h &= 0,
  \end{align}
\end{subequations}
where $u_h = u_h(x, t)$ and $\vec{u}_h$ is a vector of point values obtained by a sampling of $u_h$ on a grid, and where $\mathcal{L}_h$ is a linear differential operator acting in the spatial/physical domain, and $\tensorq{L}_h$ is a matrix representing the discretisation of~$\mathcal{L}_h$. The interplay between~\eqref{eq:434} and \eqref{eq:435} has been analysed especially from the perspective of the \emph{dispersion relation}, which is also tantamount to the \emph{comparison of eigenvalues} $\mathcal{L}_h$ and $ \tensorq{L}_h$.

First, we have analysed the interplay problem~\eqref{eq:433} for an \emph{infinite} and for \emph{a periodic domain}. Using Fourier analysis methods we have provided a complete characterisation of the equivalence between~\eqref{eq:434} and~\eqref{eq:435} for a class of operators~$\mathcal{L}_h$/$ \tensorq{L}_h$. In particular, we have studied the case where the matrix~$\tensorq{L}_h$ is an (infinite) symmetric circulant matrix, see Theorem~\ref{thr:2} and Theorem~\ref{thr:3}, which is the physically relevant situation where~\eqref{eq:435} represents a lattice of multiple-neighbours interacting particles. The obtained characterisation provides answer to the \emph{discretisation problem}---how to construct $\tensorq{L}_h$ out of given $\mathcal{L}_h$ such that important qualitative properties of~\eqref{eq:434} are preserved---and to the \emph{continualisation problem}---how to construct $\mathcal{L}_h$ and $u_h$ out of given $\tensorq{L}_h$ and $\vec{u}_h$  such that important qualitative properties of~\eqref{eq:435} are preserved. Besides that the characterisation also answers the \emph{backward error analysis problem} for~\eqref{eq:435}, see~\cite{higham.nj:accuracy}, meaning that it provides an answer to the question: ``If~$\vec{u}_h$ is an \emph{approximate} solution to~\eqref{eq:434} in the sense that it solves~\eqref{eq:435} where $\tensorq{L}_h$ is a discrete counterpart of $\mathcal{L}_h$, and if $u_h$ is constructed out of the solution $\vec{u}_h$, what \emph{nearby} problem of type~\eqref{eq:434} is actually \emph{exactly} solved by $u_h$?''

Second, the previous analysis in \emph{infinite} and \emph{periodic domain} has been further extended to \emph{finite domain} with \emph{fixed ends/zero Dirichlet} boundary conditions. Even here we have been able to exploit the Fourier analysis toolbox, and we have provided complete characterisation of the equivalence between~\eqref{eq:434} and~\eqref{eq:435} in the case of nearest-neighbour interactions. For the multiple-neighbour interactions we have been only able to characterise eigenvalues of a class of discrete operators representing higher order finite differences type discretisations of the second derivative operator. (We have not been able to construct continuous counterparts of these operators.) We have also briefly discussed the application of Fourier analysis toolbox beyond the analysis of the simple second derivative $\ddd{}{x}$ operator. Namely, we have focused on the discrete sine transform based discretisation of regular Sturm--Liouville operators, which is typically discouraged practice especially in the numerical mathematics community. For example, \cite[Chapter 3, page 47]{shen.j.tang.t.ea:spectral} claim:
\begin{quote}
The Fourier spectral method is only appropriate for problems with periodic
boundary conditions. If a Fourier method is applied to a non-periodic problem,
it inevitably induces the so-called Gibbs phenomenon, and reduces the global
convergence rate to $\bigo{N^{-1}}$ (cf. \cite{gottlieb.d.orszag.sa:numerical}
). Consequently,
one should not apply a Fourier method to problems with non-periodic boundary
conditions. Instead, one should use orthogonal polynomials which are eigenfunctions of some singular Sturm--Liouville problems. The commonly used orthogonal
polynomials include the Legendre, Chebyshev, Hermite and Laguerre polynomials.
\end{quote}
While it is true that Fourier spectral method based of the discrete sine transform does not---in the non-periodic setting---provide a particularly appealing \emph{function approximation} properties, the discrete sine transform based discretisation has another extremely convenient property. The discrete sine transform based discretisation is superior in terms of \emph{eigenvalue approximation} properties of the simple second derivative operator $\ddd{}{x}$, and our numerical experiment indicates that this property is preserved even for complex regular Sturm--Liouville type differential operators. Based on this observation, we have conjectured that the discrete sine transform is an ideal candidate for the \emph{dispersion relation preserving discretisation} of~\eqref{eq:434} where $\mathcal{L}_h$ is a regular Sturm--Liouville operator. This conjecture has to be further investigated.

Finally, we note that in our analysis we have worked with the continuous time variable, which is the setting interesting from the perspective of correspondence between lattice and continuous models in physics. From the perspective of numerical analysis it would be however interesting to discretise the time variable as well. So far some analysis of this type exists only for conventional time-stepping schemes, see, for example, \cite{trefethen.ln:group}, but it might be worthwhile to analyse complete space-time discretisation based on the \emph{Fourier transform both in the spatial and temporal variable}. This might lead to dispersion relation preserving schemes even on the fully discrete level.


\newpage
\appendix

\section{Implementation}
\label{sec:implementation-2}

As we have already noted the particular arrangement of grid values collections $\left\{ u_{h, j}\right\}_{j=1}^N$ does matter if we want to use standard computer codes dealing with the discrete Fourier transform and the discrete sine transform. Below we discuss how to implement the key manipulations using standard computer codes.

\subsection{Periodic lattice}
\label{sec:periodic-lattice-1}
We start with the periodic lattice, that is with the \emph{discrete Fourier transform} toolbox.

\subsubsection{Discrete Fourier transform}
\label{sec:discr-four-transf-5}
 \texttt{Matlab} implementation of discrete Fourier transform \texttt{fft} transforms vector $\vec{X}$ of length~$N$ to vector $\vec{Y}$ of length~$N$ according to the formula
\begin{equation}
  \label{eq:466}
  Y(k) = \sum_{j=1}^N X(j) \exponential{- \frac{2 \pi \iunit}{n} \left(k-1\right)\left(j-1\right)},
\end{equation}
while the inverse transformation \texttt{ifft} is defined as
\begin{equation}
  \label{eq:467}
  X(j) = \frac{1}{N}\sum_{k=1}^N Y(k) \exponential{\frac{2 \pi \iunit}{n} \left(k-1\right)\left(j-1\right)},
\end{equation}
where the symbols $Y(k)$ and $X(j)$ denote the $k$-th component and $j$-th component of the corresponding vectors. Components of vectors are in \texttt{Matlab} numbered from one.

\texttt{Wolfram Mathematica} implementation of discrete Fourier transform transforms a list $\vec{u}$ of length $N$ to the list $\vec{v}$ of length~$N$ via the formula
\begin{equation}
  \label{eq:468}
  v_s
  =
  \frac{1}{N^{\frac{1-a}{2}}}
  \sum_{r=1}^N
  u_r \exponential{\frac{2\pi \iunit}{n} b \left(r-1\right)\left(s-1\right)}
  ,
\end{equation}
where $a$ and $b$ are parameters with default values $a = 0$, $b = 1$. The parameters \texttt{\{a, b\}} are called \texttt{FourierParameters}, while the corresponding function is simply called \texttt{Fourier}. The inverse discrete Fourier transform is defined as
\begin{equation}
  \label{eq:469}
  u_s
  =
  \frac{1}{N^{\frac{1+a}{2}}}
  \sum_{r=1}^N
  v_r \exponential{-\frac{2\pi \iunit}{n} b \left(r-1\right)\left(s-1\right)}
  ,
\end{equation}
and the corresponding function is called \texttt{InverseFourier}. If we however set the option \texttt{FourierParameters} as
\begin{equation}
  \label{eq:470}
  \mathtt{FourierParameters -> \{1, -1\}},
\end{equation}
then we get the same definition as in \texttt{Matlab}. Components of lists are in \texttt{Wolfram Mathematica} numbered from one. We thus have the following ``equality'' between \texttt{Matlab} and \texttt{Wolfram Mathematica} implementation of the discrete Fourier transform, 
\begin{equation}
  \label{eq:471}
  \mathtt{fft(X)} = \mathtt{Fourier[X, FourierParameters -> \{1, -1\}]}.
\end{equation}
Finally, we note that \texttt{Matlab} and \texttt{Wolfram Mathematica} can generate matrices representing the application of discrete Fourier transform, check \texttt{FourierMatrix} function in \texttt{Wolfram Mathematica}.

Formula for the Fourier transform~\eqref{eq:466} can be manipulated as follows
\begin{multline}
  \label{eq:472}
  Y(k)
  =
  \sum_{j=1}^N X(j) \exponential{- \frac{2 \pi \iunit}{n} \left(k-1\right)\left(j-1\right)}
  =
  \sum_{j=1}^N X(j) \exponential{- \iunit \left(k-1\right) x_{h, j-1}}
  =
  \sum_{j=1}^N \tensor{\left( \vec{u}_h \right)}{_{j-1}} \exponential{- \iunit \left(k-1\right) x_{h, j-1}}
  =
  \tensor{\left( \vec{u}_h \right)}{_{N}}
  +
  \sum_{j=2}^N \tensor{\left( \vec{u}_h \right)}{_{j-1}} \exponential{- \iunit \left(k-1\right) x_{h, j-1}}
  \\
  =
  \tensor{\left( \vec{u}_h \right)}{_{N}}
  +
  \sum_{l=1}^{N-1} \tensor{\left( \vec{u}_h \right)}{_l} \exponential{- \iunit \left(k-1\right) x_{h, l}}
  =
  \sum_{l=1}^{N} \tensor{\left( \vec{u}_h \right)}{_l} \exponential{- \iunit \left(k-1\right) x_{h, l}}
  =
  \frac{1}{h}
  \tensor{\left( \FourierTransformDiscrete{\vec{u}_h} \right)}{_{k-1}}
  ,
\end{multline}
where we have exploited the definition of grid points $x_{h, j} = hj = \frac{2 \pi}{N} j$, the periodicity $u_{h, 0} = u_{h, N}$, and where we have arranged the grid values $\{u_{h, j} (t)\}_{j=1}^{N}$ into the column vector $\vec{X}$ such that $ X(j) =_{\bydefinition} \tensor{\left( \vec{u}_h \right)}{_{j-1}}$, that is
\begin{equation}
  \label{eq:473}
  \begin{bmatrix}
    X_1 \\
    X_2 \\
    X_3 \\
    \vdots \\
    X_{N-1} \\
    X_N \\
  \end{bmatrix}
  =_{\bydefinition}
  \begin{bmatrix}
    u_{h, N} \\
    u_{h, 1} \\
    u_{h, 2} \\
    \vdots \\
    u_{h, N-2} \\
    u_{h, N-1}
  \end{bmatrix}
  .
\end{equation}
Note that the grid values vector on the right-hand side of~\eqref{eq:439} can be, in virtue of the periodicity, also rewritten as
\begin{equation}
  \label{eq:474}
    \begin{bmatrix}
    u_{h, 0} \\
    u_{h, 1} \\
    u_{h, 2} \\
    \vdots \\
    u_{h, N-2} \\
    u_{h, N-1}
    \end{bmatrix}
    .
\end{equation}

Consequently, if we arrange the grid values $\vec{u}_h$ to a vector in a lexicographical order\footnote{Recall that $u_{h, N} = u_{h, 0}$ hence the order is indeed lexicographical.}, that is if we set 
\begin{equation}
  \label{eq:475}
  \vec{u}_h
  =
  \begin{bmatrix}
    u_{h, N} \\
    u_{h, 1} \\
    u_{h, 2} \\
    u_{h, 3} \\
    \vdots \\
    u_{h, N-1}
  \end{bmatrix}
  ,
\end{equation}
then we can find the elements of the discrete Fourier transform~\eqref{eq:188},
\begin{equation}
  \label{eq:476}
      \tensor{\left( \widehat{\vec{u}_h} \right)}{_l}
       =
       h
       \sum_{j= 1}^{N}
       \tensor{\left(\vec{u}_h\right)}{_j}
       \exponential{- \frac{2 \pi}{N} \iunit l j}
       ,
\end{equation}
by the application of \texttt{fft} to the vector~\eqref{eq:475}. The application of $\texttt{fft}$ to the vector~\eqref{eq:475} yields a vector wherein the elements of $\widehat{\vec{u}_h}$ are arranged as
\begin{equation}
  \label{eq:477}
  \begin{bmatrix}
    \tensor{\left(\widehat{\vec{u}_h}\right)}{_0} \\
    \tensor{\left(\widehat{\vec{u}_h}\right)}{_1} \\
    \tensor{\left(\widehat{\vec{u}_h}\right)}{_2} \\
    \vdots \\
    \tensor{\left(\widehat{\vec{u}_h}\right)}{_{\frac{N}{2} -1}} \\
    \tensor{\left(\widehat{\vec{u}_h}\right)}{_{\frac{N}{2}}}  \\
    \tensor{\left(\widehat{\vec{u}_h}\right)}{_{-\frac{N}{2} + 1}}  \\
    \tensor{\left(\widehat{\vec{u}_h}\right)}{_{-\frac{N}{2} + 2}}  \\
    \vdots \\
    \tensor{\left(\widehat{\vec{u}_h}\right)}{_{-1}}
  \end{bmatrix}
  =
  h
  \,
  \mathtt{fft}
  \left(
    \begin{bmatrix}
      u_{h, N} \\
      u_{h, 1} \\
      u_{h, 2} \\
      \vdots \\
      u_{h, \frac{N}{2} -1} \\
      u_{h, \frac{N}{2} } \\
      u_{h, \frac{N}{2} + 1} \\
      u_{h, \frac{N}{2} + 2} \\
      \vdots \\
      u_{h, N-1}
    \end{bmatrix}
  \right)
  .
\end{equation}
Similarly, for the inverse discrete Fourier transform we have
\begin{equation}
  \label{eq:478}
     \begin{bmatrix}
      u_{h, N} \\
      u_{h, 1} \\
      u_{h, 2} \\
      \vdots \\
      u_{h, \frac{N}{2} -1} \\
      u_{h, \frac{N}{2} } \\
      u_{h, \frac{N}{2} + 1} \\
      u_{h, \frac{N}{2} + 2} \\
      \vdots \\
      u_{h, N-1}
     \end{bmatrix}
     =
     \frac{N}{2 \pi}
     \;
     \mathtt{ifft}
     \left(
       \begin{bmatrix}
         \tensor{\left(\widehat{\vec{u}_h}\right)}{_0} \\
         \tensor{\left(\widehat{\vec{u}_h}\right)}{_1} \\
       \tensor{\left(\widehat{\vec{u}_h}\right)}{_2} \\
         \vdots \\
         \tensor{\left(\widehat{\vec{u}_h}\right)}{_{\frac{N}{2} -1}} \\
         \tensor{\left(\widehat{\vec{u}_h}\right)}{_{\frac{N}{2}}}  \\
         \tensor{\left(\widehat{\vec{u}_h}\right)}{_{-\frac{N}{2} + 1}}  \\
         \tensor{\left(\widehat{\vec{u}_h}\right)}{_{-\frac{N}{2} + 2}}  \\
         \vdots \\
         \tensor{\left(\widehat{\vec{u}_h}\right)}{_{-1}}
       \end{bmatrix}
     \right)
     .
   \end{equation}

   \subsubsection{Differentiation}
   \label{sec:differentiation}
   The formula~\eqref{eq:244} for derivative can be implemented as
\begin{equation}
  \label{eq:479}
  \begin{bmatrix}
    \tensor{\left( \dd{^nu_h}{x^n} \right)}{_N} \\
    \tensor{\left( \dd{^nu_h}{x^n} \right)}{_1} \\
    \tensor{\left( \dd{^nu_h}{x^n} \right)}{_2} \\
    \vdots \\
    \tensor{\left( \dd{^nu_h}{x^n} \right)}{_{\frac{N}{2} -1}} \\
    \tensor{\left( \dd{^nu_h}{x^n} \right)}{_{\frac{N}{2}}} \\
    \tensor{\left( \dd{^nu_h}{x^n} \right)}{_{\frac{N}{2} + 1}} \\
    \tensor{\left( \dd{^nu_h}{x^n} \right)}{_{\frac{N}{2} + 2}} \\
    \vdots \\
    \tensor{\left( \dd{^nu_h}{x^n} \right)}{_{N-1}}
  \end{bmatrix}
  =
  \begin{cases}
    \mathtt{ifft}
    \left(
      \tensorschur{
        \begin{bmatrix}
          0 \\
          \left( \iunit 1 \right)^n \\
          \left( \iunit 2 \right)^n \\
          \vdots \\
          \left( \iunit \left( \frac{N}{2} - 1 \right)  \right)^n \\
          \left( \iunit \frac{N}{2} \right)^n \\
          \left( \iunit \left( - \frac{N}{2} + 1\right) \right)^n \\
          \left( \iunit \left( - \frac{N}{2} + 2\right) \right)^n \\
          \vdots \\
          \left( - \iunit  \right)^n
        \end{bmatrix}
      }
      {
        \mathtt{fft}
        \left(
          \begin{bmatrix}
            u_{h, N} \\
            u_{h, 1} \\
            u_{h, 2} \\
            \vdots \\
            u_{h, \frac{N}{2} -1} \\
            u_{h, \frac{N}{2} } \\
            u_{h, \frac{N}{2} + 1} \\
            u_{h, \frac{N}{2} + 2} \\
            \vdots \\
            u_{h, N-1}
          \end{bmatrix}
        \right)
      }
    \right), & $n$ \text{ is even}, \\
    & \\
        \mathtt{ifft}
    \left(
      \tensorschur{
        \begin{bmatrix}
          0 \\
          \left( \iunit 1 \right)^n \\
          \left( \iunit 2 \right)^n \\
          \vdots \\
          \left( \iunit \left( \frac{N}{2} - 1 \right)  \right)^n \\
          0 \\
          \left( \iunit \left( - \frac{N}{2} + 1\right) \right)^n \\
          \left( \iunit \left( - \frac{N}{2} + 2\right) \right)^n \\
          \vdots \\
          \left( - \iunit  \right)^n
        \end{bmatrix}
      }
      {
        \mathtt{fft}
        \left(
          \begin{bmatrix}
            u_{h, N} \\
            u_{h, 1} \\
            u_{h, 2} \\
            \vdots \\
            u_{h, \frac{N}{2} -1} \\
            u_{h, \frac{N}{2} } \\
            u_{h, \frac{N}{2} + 1} \\
            u_{h, \frac{N}{2} + 2} \\
            \vdots \\
            u_{h, N-1}
          \end{bmatrix}
        \right)
      }
    \right),
    &
       $n$ \text{ is odd}.
  \end{cases}
    \end{equation}
Recall that for the odd order differentiation we set $\tensor{\left( \dd{^nu_h}{x^n} \right)}{_{\frac{N}{2}}} =_{\bydefinition} 0$, see also \cite[Chapter 3]{trefethen.ln:spectral}.

\subsubsection{Eigenvalue problem for circulant matrices}
\label{sec:eigenv-probl-circ}

The eigenvalues in the eigenvalue problem~\eqref{eq:218} for a circulant matrix are given as discrete Fourier transform of the first column. In particular, if we have the matrix
\begin{equation}
  \label{eq:480}
  \tensorq{A}_h
  =_{\bydefinition}
      \begin{bmatrix}
      a_N & a_{N-1} & a_{N-2} & \dots & a_1 \\
      a_1 & a_N & a_{N-1} & \dots & a_2 \\
      a_2 & a_1 & a_N & \dots & a_3 \\
      \vdots & \vdots & \vdots& \ddots & \vdots \\
      a_{N-1} & a_{N-2} & a_{N-3} & \dots & a_N
      \end{bmatrix}
      ,
    \end{equation}
    then its eigenvalues are given by the formula
    \begin{equation}
      \label{eq:481}
    \vec{\lambda}_{\tensorq{A}_h}
    =
    \mathtt{fft}
    \left(
      \begin{bmatrix}
        a_{N} \\
        a_1 \\
        a_2 \\
        \vdots \\
        a_{N-1} \\
      \end{bmatrix}
    \right)
    .
  \end{equation}
  Note that due to the particular arrangement of vector elements we are literally taking $\mathtt{fft}$ of the first column. Compare with the definition of the generating vector in~\eqref{eq:216}.

\subsection{Fixed ends/zero Dirichlet boundary conditions lattice}
\label{sec:fixed-endsz-dirichl}

\subsubsection{Discrete sine transform}
\label{sec:discr-sine-transf-4}
Due to particular arrangement of grid values in a vector, see~\eqref{eq:475}, we must shift the odd extension~\eqref{eq:292} accordingly. (This also holds for the extension/reduction matrices.) Furthermore, the discrete sine transform \texttt{dst} in \texttt{Matlab} is defined as
\begin{equation}
  \label{eq:482}
  Y(k)
  =
  \sum_{j= 1}^{M}
  X(j)
  \sin \left(  \frac{\pi}{M+1} k j \right),
\end{equation}
where $\vec{X}$ and $\vec{Y}$ are vectors of length $M$. (Note that in \texttt{Matlab} the \texttt{dst} function is labeled as ``not recommended''.) \texttt{Wolfram Mathematica} implements the discrete sine transform (DST-I) with a different normalisation, namely
\begin{equation}
  \label{eq:462}
  v_k
  =
  \sqrt{\frac{2}{M+1}}
  \sum_{j= 1}^{M}
  u_j
  \sin \left(  \frac{\pi}{M+1} k j \right),
\end{equation}
where $\vec{u}$ and $\vec{v}$ are lists of length $M$. The function name is \texttt{FourierDST}[$\cdot$, \texttt{I}]. Taking into account the normalisation factor in Wolfram Mathematica we thus have the following equality
\begin{equation}
  \label{eq:483}
  \begin{bmatrix}
    0
    \\
    \sqrt{\frac{M+1}{2}}
    \mathtt{FourierDST}
    [
    \begin{bmatrix}
      u_{h, 1} \\
        u_{h, 2} \\
        u_{h, 3} \\
        \vdots \\
        u_{h, M} \\
    \end{bmatrix},
    \mathtt{I}
    ]
    \\
    0
    \\
    -
    \mathtt{Reverse}
    [
    \sqrt{\frac{M+1}{2}}
    \mathtt{FourierDST}
    [
    \begin{bmatrix}
      u_{h, 1} \\
        u_{h, 2} \\
        u_{h, 3} \\
        \vdots \\
        u_{h, M} \\
    \end{bmatrix}
    ],
    \mathtt{I}
    ]
  \end{bmatrix}
  =
  -
  \frac{1}{2}
  \mathtt{Im}
  [
    \mathtt{Fourier}
    [
      \begin{bmatrix}
        0 \\
        u_{h, 1} \\
        u_{h, 2} \\
        u_{h, 3} \\
        \vdots \\
        u_{h, M} \\
        0 \\
        -u_{h, M} \\
        \vdots \\
        -u_{h, 3} \\
        -u_{h, 2} \\
        -u_{h, 1} \\
      \end{bmatrix}
      ,
      \mathtt{FourierParameters->\{1, -1\}}
    ]
  ],
\end{equation}
see also \eqref{eq:477} for a particular arrangement of Fourier space components in a column vector/list. In \texttt{Matlab} we however have
\begin{equation}
  \label{eq:464}
  \begin{bmatrix}
    0
    \\
    \mathtt{dst}
    (
    \begin{bmatrix}
      u_{h, 1} \\
        u_{h, 2} \\
        u_{h, 3} \\
        \vdots \\
        u_{h, M} \\
    \end{bmatrix}
    )
    \\
    0
    \\
    -
    \mathtt{flip}
    (
    \mathtt{dst}
    (
    \begin{bmatrix}
      u_{h, 1} \\
        u_{h, 2} \\
        u_{h, 3} \\
        \vdots \\
        u_{h, M} \\
    \end{bmatrix}
    )
    )
  \end{bmatrix}
  =
  -
  \frac{1}{2}
  \mathtt{imag}
  (
    \mathtt{fft}
    (
      \begin{bmatrix}
        0 \\
        u_{h, 1} \\
        u_{h, 2} \\
        u_{h, 3} \\
        \vdots \\
        u_{h, M} \\
        0 \\
        -u_{h, M} \\
        \vdots \\
        -u_{h, 3} \\
        -u_{h, 2} \\
        -u_{h, 1} \\
      \end{bmatrix}
    )
  ),
\end{equation}
Equalities~\eqref{eq:483} are \eqref{eq:464} are particular implementations of~\eqref{eq:discrete-sine-transform}.

\subsubsection{Differentiation}
\label{sec:differentiation-1}
As the inverse to the discrete sine transform DST-I is the discrete sine transform DST-I, the differentiation formula~\eqref{eq:446} reads
\begin{equation}
  \label{eq:465}
  \begin{bmatrix}
    \left. \dd{^n}{x^n} u_h \right|_{x = x_{h, 1}} \\
    \left. \dd{^n}{x^n} u_h \right|_{x = x_{h, 2}} \\
    \left. \dd{^n}{x^n} u_h \right|_{x = x_{h, 3}} \\
    \vdots \\
    \left. \dd{^n}{x^n} u_h \right|_{x = x_{h, M}}
  \end{bmatrix}
  =
  \mathtt{FourierDST}
  [
  \tensorschur{
    \begin{bmatrix}
      \left( \iunit \right)^n \\
      \left( 2 \iunit \right)^n \\
      \left( 3 \iunit \right)^n \\
      \vdots \\
      \left( M \iunit \right)^n
    \end{bmatrix}
  }
  {
    \mathtt{FourierDST}
    [
    \begin{bmatrix}
      u_{h, 1} \\
      u_{h, 2} \\
      u_{h, 3} \\
      \vdots \\
      u_{h, M} \\
    \end{bmatrix}
    ,
    \mathtt{I}
    ]
  },
  \mathtt{I}
  ].
\end{equation}

\subsubsection{Eigenvalues computation}
\label{sec:eigenv-comp}

The solution to the eigenvalues problem
\begin{subequations}
  \label{eq:456}
     \begin{align}
       \label{eq:457}
       - \ddd{u}{x} + q(x) u &= \lambda u, \\
       \label{eq:458}
       \left. u \right|_{x =0, \pi} & =0,
     \end{align}
   \end{subequations}
in the domain $(0, \pi)$ can be found using the code shown in Listing~\ref{lst:matlab-eigs}. The function \lstinline{dst_eig(mm, q)} takes a given function handler \lstinline{q} and matrix size \lstinline{mm} and it returns \lstinline{mm} eigenvalues of the operator~\eqref{eq:456} discretised using the Fourier transform technique discussed in Section~\ref{sec:numerical-experiment}.
   
\begin{lstlisting}[float, language=Matlab, morekeywords={dst, idst, arguments}, caption={\texttt{Matlab} code for DST based solution of generic eigenvalue problem~\eqref{eq:456} with a given function~$q$; function~\lstinline{dst_eig}, for source code see~\href{matlab/dst\_eig.m}{matlab/dst\_eig.m}.}, label={lst:matlab-eigs}]
function lambdas = dst_eig(mm, q)
% Returns mm eigenvalues of linear operator -y'' + q(x)y on interval [0, pi]
% with Dirichlet boundary conditions, operator is discretised using DST
% based technique with matrix size mm.
arguments
    % size of matrix/number of eigenvalues
    mm (1, 1) {mustBeInteger, mustBePositive}
    % function q
    q function_handle
end

% Second derivative out of vector of grid values x, operator form
% Grid values are arranged in a column vector
% Symbol kvec denotes the vector of negative squares, differentiation in DST space
kvec = @(x) -((1:length(x)).^2);
D2op = @(x) idst(kvec(x)' .* dst(x));

% Multiplication of grid values vector x with grid values q(x) function
% Grid values are arranged in a column vector
xgrid = @(x) pi/(length(x)+1) * (1:length(x));
Qop = @(x) (q(xgrid(x)') .* x);

% Sturm--Liouville operator -y'' + q(x)y
% Grid values are arranged in a column vector -- we need this for eigs
% solver in the matrix--free form
Lop = @(x) -D2op(x) + Qop(x);

% Solve the eigenvalue problem using explicitely formed matrix, return
% eigenvalues
L = Lop(eye(mm));
[~, D] = eig(L);
lambdas = sort(diag(D));

end
\end{lstlisting}

\begin{lstlisting}[float, language=Matlab, morekeywords={dst, idst, arguments}, caption={\texttt{Matlab} code for DST based solution of eigenvalue problem~\eqref{eq:429}; sample use of function~\lstinline{dst_eig}, for source code see~\href{matlab/dst\_eig\_run.m.m}{matlab/dst\_eig\_run.m}.}, label={lst:matlab-eigs-run}, linerange={4-}]
% Specify function q(x)
q = @(x) exp(x);
% Specify number of eigenvalues/matrix size
mm = 50;

% Call dst_eig and get the eigenvalues
tic
lambdas = dst_eig(mm, q);
t = toc;

fprintf('** All eigenvalues of discrete operator L, m=%d\n', mm)
fprintf('** Time elapsed: %d\n', t)

fprintf('%g\n', lambdas')
\end{lstlisting}



\bibliographystyle{chicago}
\bibliography{vit-prusa}

\end{document}